\documentclass[11pt]{article}
\usepackage[usenames,dvipsnames]{xcolor}
\definecolor{darkgreen}{rgb}{0.0, 0.5, 0.0}
\usepackage{fullpage}
\usepackage{amsmath} 
\usepackage{amsfonts}
\usepackage{amssymb}
\usepackage{authblk}
\usepackage{dsfont} 
\usepackage{graphicx}
\usepackage{tikz}
\usepackage{subcaption}

\usepackage[colorlinks=true, linkcolor=blue, citecolor=blue, hypertexnames=false]{hyperref}  
\usepackage[authoryear,round]{natbib}
\setcitestyle{authoryear, round, aysep={,}}
\newcommand{\journalversion}[1]{#1}
\newcommand{\blindversion}[1]{}
\usepackage{array}
\usepackage{amsthm}
\usepackage{thmtools}
\usepackage{thm-restate}

\usepackage{graphicx}
\usepackage{pgfplots}
\usepackage{tikz}
\usepackage{subcaption}
\usepackage{array}
\usepackage{amsthm}
\usepackage{thmtools}
\usepackage{thm-restate}
\usepackage{nameref}

\usepackage{latexsym}
\usepackage{amsmath,amssymb}
\usepackage{epsfig}
\usepackage{setspace}
\usepackage{hyperref}
\usepackage{cleveref}
\usepackage{amssymb}
\usepackage{xargs}
\usepackage[normalem]{ulem}

\usepackage{color}
\newtheorem{theorem}{Theorem}

\newtheorem{lemma}{Lemma}

\newtheorem{definition}{Definition}

\newtheorem{proposition}{Proposition}%[section]

\usepackage{enumerate}

\usepackage{calrsfs}

\usepackage{nicefrac}       % compact symbols for 1/2, etc.
\usepackage{microtype}      % microtypography
 \usepackage{float}
\usepackage{dsfont}
\usepackage{enumerate}

\usepackage{amsmath}
\usepackage{enumerate}
\numberwithin{equation}{section}
\usepackage{bbm}

\DeclareMathAlphabet{\pazocal}{OMS}{zplm}{m}{n}

\newcommand{\IN}{\mathbb{A}}

\DeclareMathOperator*{\argmax}{arg\,max}
\usepackage{calrsfs}
\usepackage{graphicx}
\usepackage{xcolor}

\newcommand{\qm}[1][s]{q_{\frac{r_F}{2}}}
\newcommand{\qtss}[1][s]{q_{t_s}}
\newcommand{\qs}[1][s]{q^*}
\newcommand{\qw}[1][s]{q_w}
\newcommand{\qmin}[1][s]{\alpha^{\frac{1}{\alpha-1}}}
\newcommand{\rft}[1][s]{\frac{r_F}{2}}
\newcommand{\qsht}[1][s]{q_{\hat{s}}}

\newcommand{\Gad}[1]{\Gamma\left(#1 \right)}
\newcommand{\Gainv}[1]{\Gamma^{-1}\left(#1 \right)}

\newcommand{\qga}[1][s]{q_{\gamma r_F}}

\newcommand{\qg}[1][s]{q_{\frac{r_F}{\gamma}}}
\newcommand{\qgg}[1][s]{q_{\frac{r_F}{\gamma^2}}}
\newcommand{\qss}[1][s]{q_{s}}

\newcommand{\qip}{q_{i+1}}

\newcommand{\bF}{\overline{F}}

\newcommand {\bear}{\begin{eqnarray}}
\newcommand {\eear}{\end{eqnarray}}
\newcommand {\bearn}{\begin{eqnarray*}}
\newcommand {\eearn}{\end{eqnarray*}}
\newcommand {\opt}{\mbox{\normalfont opt}}

\newcommand{\Fb}{\pazocal{F}}
\newcommand{\Cb}{\pazocal{C}}

\newcommand{\Gb}{\pazocal{G}}

\newcommand{\Lb}{\pazocal{L}}

\newcommand{\Expect}{\mathbb{E}}

\def\rs{{p^*}}
\def\qs{{q^*}}

\newcommand{\distSet}{\pazocal{D}_{[\lb, \ub]}}
\newcommand{\MechSet}{\pazocal{P}}

\newcommand{\aDistSetqa}[2][]{\pazocal{F}_{#1}( #2)}

\usepackage{tikz}
\usepackage{pgfpages}
\usepackage{pgfplots}
\usetikzlibrary{arrows,decorations,patterns,trees,matrix,calc,shapes}

\usepackage{graphicx}

\usetikzlibrary{calc,trees,positioning,arrows,chains,shapes.geometric,%
    decorations.pathreplacing,decorations.pathmorphing,shapes,%
    matrix,shapes.symbols}
\usetikzlibrary{matrix}
\usepgfplotslibrary{fillbetween}
\pgfdeclareplotmark{bluestar}{
    \node[star,star point ratio=2.25,minimum size=6pt,
          inner sep=0pt,draw=black,solid,fill=blue] {};
}
\pgfdeclareplotmark{redstar}{
    \node[star,star point ratio=2.25,minimum size=6pt,
          inner sep=0pt,draw=black,solid,fill=red] {};
}

\makeatletter
\DeclareRobustCommand*\cal{\@fontswitch\relax\mathcal}
\makeatother

\usepackage{ifthen}
\usepackage{csvsimple}
\usepackage{pgfplotstable}
\usepackage{booktabs}
\newcommand{\pvect}{\mathbf{p}}
\newcommand{\qvect}{\mathbf{q}}

\newcommand{\Iset}{\mathcal{I}_{\pvect, \qvect}}
\newcommand{\IsetPlus}{\mathcal{I}_{\pvect, \qvect}^{+}}

\newcommand{\Rev}{\mathrm{Rev}}

\newcommand{\lb}{\underline{v}}
\newcommand{\ub}{\overline{v}}

\newcommandx{\rl}[2][1=\Cb, 2=\Iset]{\underline{r}_{#1}\left(#2\right)}
\newcommandx{\rh}[2][1=\Cb, 2=\Iset]{\overline{r}_{#1}\left(#2\right)}

\newcommand{\Ratio}[1][\Psi,F]{\mathrm{Ratio}\left(#1\right)}
\newcommand{\MMRatio}[1][\MechSet,\Cb(\Iset)]{\mathfrak{R}\left(#1\right)}
\newcommandx{\LRegret}[2][1=\lambda, 2={\Psi,F}]{R_{#1}\left(#2\right)}

\usepackage{nameref}
\usepackage{cleveref}

\usepackage{amsmath} 
\usepackage{amssymb}
\usepackage{mathtools}

\usepackage{latexsym}
\usepackage{amsmath,amssymb}
\usepackage{epsfig}
\usepackage[right=1.0in, top=1in, bottom=1.0in, left=1.0in]{geometry}
\usepackage{setspace}
\usepackage{cleveref}
\usepackage{amssymb}
\usepackage{amsthm}
\usepackage{thmtools}
\usepackage{color}
\usepackage{bm}

\usepackage{enumerate}

\usepackage{calrsfs}

\usepackage{setspace}
\usepackage[utf8]{inputenc} % allow utf-8 input
\usepackage{nameref}
\usepackage{cleveref}     % hyperlinks
\usepackage{url}            % simple URL typesetting
\usepackage{booktabs}       % professional-quality tables
\usepackage{amsfonts}       % blackboard math symbols
\usepackage{nicefrac}       % compact symbols for 1/2, etc.
\usepackage{microtype}      % microtypography
\usepackage{amsmath,mathtools}
\usepackage{cleveref}
 \usepackage{float}
\usepackage{dsfont}
\usepackage{enumerate}

\usepackage{amsmath}
\usepackage{enumerate}
\numberwithin{equation}{section}
\usepackage[all]{nowidow}
\usepackage[utf8]{inputenc}
\usepackage[inline]{enumitem}

\usepackage{bbm}

\usepackage{tikz}
\usepackage{pgfpages}
\usepackage{pgfplots}
\usepgfplotslibrary{groupplots}
\usepackage{pgfplotstable}
\pgfplotsset{compat=1.18}
\usetikzlibrary{arrows,decorations,patterns,trees,matrix,calc,shapes}

\usepackage{graphicx}

\usepackage{filecontents}
\usetikzlibrary{calc,trees,positioning,arrows,chains,shapes.geometric,%
    decorations.pathreplacing,decorations.pathmorphing,shapes,%
    matrix,shapes.symbols}
\usetikzlibrary{matrix}
\usepgfplotslibrary{fillbetween}

\pgfplotsset{
  discard if not/.style 2 args={
    x filter/.code={
      \edef\temp{\thisrow{#1}}%
      \edef\val{#2}%
      \ifx\temp\val
      \else
      \fi
    }
  }
}

\pgfdeclareplotmark{bluestar}{
    \node[star,star point ratio=2.25,minimum size=6pt,
          inner sep=0pt,draw=black,solid,fill=blue] {};
}
\pgfdeclareplotmark{redstar}{
    \node[star,star point ratio=2.25,minimum size=6pt,
          inner sep=0pt,draw=black,solid,fill=red] {};
}

\makeatletter
\DeclareRobustCommand*\cal{\@fontswitch\relax\mathcal}
\makeatother

\usepackage{ifthen}

\usepackage{csvsimple}
\usepackage{booktabs}

\usepackage[ruled]{algorithm2e} % For algorithms
\SetCommentSty{mycommfont}

\begin{document}

\setstretch{1.5}
	
	\title{Fast Revenue Maximization}

	\journalversion{
		\date{first version: July 11, 2024; last revised: May 15, 2026}
		\author[1]{Achraf Bahamou}
		\author[2]{Omar Besbes}
		\author[3]{Omar Mouchtaki}
		\affil[1]{Columbia University, IEOR department, \texttt{achraf.bahamou@columbia.edu}}
		\affil[2]{Columbia University, Graduate School of Business, \texttt{obesbes@columbia.edu}}
		\affil[3]{NYU Stern School of Business, \texttt{om2166@nyu.edu}}
		%\date{}
		\maketitle
	}
	
	\blindversion{
		\author{}
		\date{}
		\maketitle
		\vspace{0cm}
	}

\def\gqi{{\Gainv{q_{i}}}}
\def\gqim{{\Gainv{q_{i-1}}}}
\def\gqip{{\Gainv{q_{i+1}}}}
\def\gqimm{{\Gainv{q_{i-2}}}}

\def\pii{{p_{i}}}
\def\pim{{p_{i-1}}}
\def\pip{{p_{i+1}}}
\def\pimm{{p_{i-2}}}

\def\qi{{q_{i}}}
\def\qim{{q_{i-1}}}
\def\qip{{q_{i+1}}}
\def\qimm{{q_{i-2}}}

\def\mi{{m_{i}}}
\def\mim{{m_{i-1}}}
\def\mimm{{m_{i-2}}}
\def\ri{{r_{i-1,i}}}

\newcommand{\pars}[1]{\left(#1 \right)}
\newcommand{\acols}[1]{\left[#1 \right]}

\begin{abstract}
\textbf{Problem definition:} We study a data-driven pricing problem in which a seller sets a price for a single item based on demand observed at a limited number of historical prices. Our goal is to quantify the value of such information and to guide efficient price experimentation under practical constraints. 
\textbf{Methodology/results:} Our main methodological contribution is an exact reduction that characterizes the maximin revenue ratio, defined as the worst-case revenue achievable using only past data relative to the optimal revenue under full information. This reduction transforms an infinite-dimensional problem into a tractable one-dimensional optimization problem, allowing us to compute near-optimal pricing policies with explicit guarantees and to precisely quantify the value of historical data. 
\textbf{Managerial implications:} Motivated by practical constraints that limit price changes, we first evaluate the value of local information and show that the sign of the revenue gradient at a single price can provide significant guidance. We then use our framework to design efficient price experiments: we develop a method to select the next price to test so as to maximize future robust performance, and show how to substantially reduce the number of experiments needed to achieve target revenue guarantees in dynamic pricing. Finally, we show that our approach remains effective with noisy demand data, achieving near-optimal performance with as few as $25$ to $100$ samples per price.		

\end{abstract}

\section{Introduction} \label{sec:intro_problem_formulation}
In this work, we investigate the monopoly pricing problem in which a decision-maker must select a price to sell a single good to a buyer. This operational problem crystallizes a fundamental trade-off between sales volume and profit margin: setting a price too low increases demand but reduces margins, whereas setting it too high improves margins at the cost of lost sales.

An important practical challenge in pricing is that the distribution of buyers’ values is typically unknown to the decision-maker. As a result, pricing decisions must rely on historical information or transaction data, which commonly consist of prices previously offered and the corresponding purchase decisions observed at those prices. While one could, in principle, experiment with a wide range of prices to learn buyer behavior, in practice firms often restrict themselves to a relatively small set of distinct prices.
There are several reasons for this constraint. From the customer side, frequent and unpredictable price changes may undermine perceptions of fairness, discourage loyalty, or induce strategic behavior such as delaying purchases in anticipation of future discounts. From the firm side, implementing a large menu of prices entails significant operational overhead: pricing systems must be updated and coordinated across channels. Moreover, excessive price variation can blur product positioning and dilute brand value. These considerations imply that, although rich experimentation may be theoretically appealing, effective pricing in practice must be designed around a limited number of price points.

Motivated by this setting, our goal in this paper is twofold. First, we seek to \textit{quantify the value of historical pricing data} in the form of prices and their associated conversion rates by characterizing what level of revenue performance can be guaranteed based solely on such information. Second, we use this framework to \textit{guide efficient price experimentation} by identifying when additional price observations are most valuable and how they should be chosen.  To this end, we derive precise guarantees on the informativeness of these data structures, enabling decision-makers to effectively and efficiently learn the revenue curve and, ultimately, achieve ``fast" revenue maximization.

To study these questions, we consider a stylized monopoly pricing environment in which a seller offers a single indivisible good to a buyer whose willingness to pay is drawn from an unknown distribution supported on a bounded interval. The seller does not observe buyers’ values directly and instead relies on conversion rates observed at a finite number of previously offered prices, where the conversion rate at a given price corresponds to the fraction of buyers whose willingness to pay exceeds that price. Because many distinct demand distributions may be consistent with the same observations and can lead to very different revenue outcomes, we evaluate pricing decisions through a robust perspective: we measure the fraction of optimal revenue that can be guaranteed across all distributions consistent with the observed data. This worst-case performance criterion allows us to interpret historical pricing data through the lens of information quality and provides a principled way to assess the value of past price observations.

\subsection{Main Contributions}
\noindent\textbf{Performance evaluation and a tractable reduction.}
A central question in this setting is to evaluate the performance of a pricing policy given only a finite set of observed conversion rates. This requires optimizing over all demand distributions that are consistent with the observed data and asking how the policy performs under the most adverse such distribution. This evaluation problem is inherently challenging. The set of demand distributions consistent with a given collection of prices and conversion rates is infinite-dimensional and nonparametric, and may additionally be subject to constraints. The constraints arise in the form of shape constraints (e.g., distributions have to be regular) or constraints on quantiles.

Our first main contribution is to show that this problem admits a tractable reduction. 
In \Cref{thm:nature_reduction}, we establish that for any pricing policy and any information set consisting of $N$ prices and conversion rates, the worst-case ratio can be obtained by considering a \emph{one-dimensional family} of carefully constructed candidate demand distributions. This reduction transforms the original infinite-dimensional problem into a one-dimensional optimization problem.
Conceptually, the proof proceeds by explicitly constructing extremal demand distributions that are consistent with the observed conversion rates and that minimize revenue for a given policy. While the argument builds on ideas developed in \cite{ABBSinglePoint}, our analysis generalizes prior work along several dimensions. First, the reduction applies not only to regular demand distributions but also to the class of all distributions supported on the value range, allowing us to study settings without shape restrictions. Second, our results hold for information sets containing an arbitrary number of observed prices, whereas existing characterizations are limited to a single observed price. Finally, although the exposition focuses on revenue ratios, the underlying arguments extend naturally to absolute regret objectives.

\noindent\textbf{Near-optimal policies and the value of information.}
In \Cref{sec:value_information} we turn to the question of characterizing the exact value of a given information set, as well as the construction of pricing policies achieving near-optimal worst-case performance. A priori, this problem introduces an additional layer of complexity, as it requires solving a maximin optimization problem over two infinite-dimensional spaces: the space of randomized pricing policies and the space of value distributions.

Leveraging \Cref{thm:nature_reduction}, we provide in \Cref{sec:value_information} a tractable linear program whose solution yields a provably near-optimal pricing policy.
Our approach builds on the discretization framework of \cite{ABBSinglePoint} for the case $N=1$. When multiple price observations are available, however, the structural properties underlying that discretization no longer hold, which prevents a direct extension.
We address this issue by introducing a refined discretization of the price space that restores the required structure. This leads to a finite linear program whose value provides a certified lower bound on the maximin worst-case ratio and whose solution directly yields a near-optimal randomized pricing policy. Moreover, we show that the discretization error scales as $O(1/M)$, where $M$ is the number of grid points, improving over the $O(1/\sqrt{M})$ rate obtained in \cite{ABBSinglePoint} when $N=1$.

This characterization plays a central conceptual role in the paper. It allows us to interpret the maximin worst-case ratio as a quantitative measure of the \emph{intrinsic value of information}: a higher ratio indicates that the observed prices and conversion rates are sufficiently informative to support near-optimal pricing, whereas a lower ratio reflects that the data are not informative enough, regardless of the pricing policy used. This perspective enables direct comparisons across different information structures and serves as the foundation for our analysis of price experimentation, where we study how additional observations improve guaranteed revenue performance.

\noindent \textbf{New insights on price experimentation.} 
 We consider in \Cref{sec:gradient} the practical setting in which the seller cannot substantially vary the offered price and instead observes demand realizations at two neighboring prices. Leveraging our theoretical characterization of the value of information, we assess the value of such local experiments and quantify the impact of knowing the gradient of the revenue curve at a given price as a function of its magnitude. Our results show that when demand is regular, the mere sign of the gradient may already provide substantial information to the seller. More importantly, our framework enables a rigorous cost–benefit analysis by identifying scenarios in which the sign of the gradient (which can be estimated through cheaper experiments) already provides sufficient information, and scenarios in which the numerical value of the gradient (which requires longer and more expensive experiments) is needed to achieve better robust performance.

We then use our framework as a way to develop or enrich algorithms for price experimentation. In \Cref{sec:extra_point}, we consider a setting in which the seller observed the demand at a single historical point and then may select a new price to experiment with before committing to a future price. We show that by judiciously selecting the second price to experiment with, the seller can significantly increase the worst-case ratio,  by $10$ to $20\%$ across all possible conversion rates observed at the first price.
In \Cref{sec:ternary}, we demonstrate that our framework can be used to improve dynamic pricing algorithms. We illustrate that, by better quantifying the value of data, one can  reduce the extent of price experimentation needed to reach a certain revenue guarantee.
We consider the case in which the value distribution is regular (and hence the revenue curve is unimodal), and compare our stopping criterion with the commonly used one when experimenting with the Ternary search algorithm. In our experiments, our refined stopping criterion allows to divide by $2$ to $5$ the number of prices offered to ensure a ratio larger than $99\%$.

\noindent \textbf{Noisy observations.} 
Finally, we numerically investigate robust pricing with noisy demand feedback and analyze how sampling noise affects the maximin pricing policy computed via the LP formulation in \Cref{sec:value_information}. We introduce a projection-based procedure that maps empirical conversion-rate estimates to quantiles that are consistent with the underlying class of value distribution. We show that policies computed from projected noisy quantiles achieve near-optimal worst-case performance once the number of observations per price is moderately large. For instance, our proposed policy is within $10\%$ of the optimal performance with only $25$ to $100$ noisy buy/no-buy decisions at each price.

\subsection{Literature Review}\label{sec:lit_rev}

Pricing when facing an unknown value distribution has been considered in a distributionally robust framework with support information \citep{Bergemann08,Eren2009}, with moment information \citep{azar2013parametric,carrasco2018optimal,chen2022distribution,chen2023screening,wang2024minimax} and with  knowledge of quantiles \citep{azar2013optimal,ABBSinglePoint}. Variants of this problem were also analyzed in the presence of strategic customers \citep{caldentey2016intertemporal}.

The pricing problem has also been studied in the offline data-driven setting where the decision-maker has past data about the value distribution and makes a single pricing decision. 
In the case where past data takes the form of i.i.d. values sampled from the value distribution, \cite{huang2015making} quantifies the number of samples required to be within a factor $(1-\epsilon)$ from the optimal revenue achievable with full knowledge of the distribution.
\cite{fu2015randomization,Babaioff2018,Daskalakis2020,ABBSamples} 
provide performance guarantees when the seller has access to a small number of samples. \cite{chen2023modelfree} investigate the performance of data-driven pricing algorithms with transactional data and \cite{biggs2022convex} proposes pricing policies based on surrogate loss functions to incorporate contextual information. \cite{hu2019data} study a robust joint inventory–pricing problem with demand samples at different prices, characterizing optimal strategies when the ambiguity set includes non-decreasing, piecewise linear, convex (or concave) demand functions that fit historical data within a prescribed mean square error.

In the online setting, \cite{aviv2002pricing,araman2009dynamic,FariasRoy2010Dynamic} examine dynamic pricing within a Bayesian framework, in which the decision-maker possesses prior knowledge about the parameters of the value distribution.
The data structure we consider in this work is related to the one studied in the dynamic pricing and learning literature where the decision-maker observes noisy demand observation at an offered price. \cite{Kleinberg2003,Besbes2009,BroderRusmevichientong2012Dynamic,kerkinzeevi,singla2023bandit} 
study the performance of dynamic pricing policies trading-off between exploration and exploitation in parametric and non-parametric settings.
\cite[Chapter 3]{broder2011online};\cite{cheung2017dynamic,chen2020data,perakis2023dynamic} develop and analyze data-driven pricing policies with limited price changes.
\cite{bu2022online} investigates a dynamic pricing problem with linear demand, in which the decision-maker uses both offline and online data. They analyze the value of historical (offline) data by quantifying the achievable regret rate as a function of three statistics of the offline transactional data (size, location, and dispersion). 
Our work differs from this previous literature along various dimensions: we measure the performance of pricing algorithms with respect to a single-shot pricing decision as opposed to the commonly studied cumulative regret, hence there is no tradeoff between exploration and exploitation in our work. Furthermore, we 
focus on the setting where the decision-maker observes the exact conversion rate at the offered prices and we
derive an \textit{exact} characterization of the worst-case performance of algorithms that allows to exactly quantify the value of past data. In contrast to \cite{bu2022online}, our quantification of the value of data is a function of the whole dataset and therefore enables to provide data-dependent insights. The insights we derive also relate to how a decision-maker should choose prices at which to experiment before committing to a final action. In this sense, our work connects to the literature on best-arm identification and sequential experimentation, which studies how to allocate samples to identify the best action under uncertainty; see, e.g., \cite{chernoff1959sequential,araman2022diffusion,garivier2016optimal,kaufmann2016complexity,russo2020simple}. Relative to this literature, our setting features a continuum of related actions, and our focus is on quantifying the exact value of a given set of historical demand observations. We believe that our results and insights may motivate new algorithmic principles for dynamic pricing and price experimentation.

Our work relates to \cite{leme2023pricing}, which compares the informational value of transactional data and actual observations of each buyer's value. They show that, asymptotically, observing buyers' values provides significantly more information for revenue maximization. Our theoretical results differ by assuming that the decision-maker has a perfect estimate of the conversion rate at a few historical prices and characterizing the \textit{exact} value of such historical data.
It is more closely related to \cite{ABBSinglePoint}, which provides a tractable procedure to arbitrarily approximate the maximin ratio when the number of conversion rates $N=1$. Our work extends their results to settings in which the decision-maker has access to demand information across several prices. This allows us to derive new insights and prescriptions regarding efficient price experimentation. In parallel developments, that appeared after we posted an initial draft of this paper, \cite{daei2024robust} solve the problem when $N=2$ and when the distribution is assumed to be regular. They provide a closed-form characterization of the optimal deterministic price experiment in the same setting as the one we numerically investigate in \Cref{sec:extra_point}. Their focus on deterministic prices, closed forms, and associated insights, complements the understanding of this fundamental problem. \cite[Section 4]{Eren2009} also analyzes the pricing problem with limited experimentation and characterizes the worst-case ratio for general distributions with information about $N$ conversion rates under the assumption that the value distribution is discrete. Our work generalizes this analysis to arbitrary distributions (both continuous and discrete) and provides a different and unified proof which applies not only to  general but also to regular distributions.

Finally, we note that our work broadly relates to \cite{chehrazi2010monotone} which considers a general robust decision problem with limited measurements and applies their methods in the context of optimal debt-settlement, and to  \cite{meister2021learning,leme2023description,okoroafor2023non}  which study the sample complexity required to estimate the cumulative distribution functions with binary feedback. We also contribute to the broad literature which aims at deriving an exact characterization of the value of data in the small to moderate data regime \citep{schlag2006eleven,stoye2009minimax,besbes2023big,besbes2023contextual}.

\section{Problem Formulation}\label{sec:formulation}

We consider the problem of a seller trying to sell one indivisible good to one buyer. We assume that the buyer's value $v$ is drawn from some \textit{unknown} distribution $F$ represented by its cumulative distribution function (cdf) and with support included in $[\lb,\ub]$, where $0 \le \lb < \ub < \infty$. The seller knows that the distribution of values belongs to some non-parametric class of distributions $\Cb$. 
The seller also has information about conversion rates observed at historical prices offered. Formally,  let $\pvect = (p_i)_{i\in \{0,\ldots N+1\}}$ be a non-decreasing sequence of prices such that $p_0 = \lb$ and $p_{N+1} = \ub$ and let $\qvect = (q_i)_{i\in \{0,\ldots N+1\}}$ be a sequence of historical conversion rates observed such that $q_0 =1$ and $q_{N+1} = 0$. The decision-maker information set is defined as
\begin{equation*}
\Iset =  \{(p_i, q_i); \; i \in \{0, \ldots, N+1\} \}.
\end{equation*}
The problem we are interested in is the following: how can the seller leverage the information contained in $\Iset$ to ``robustly'' maximize her revenue, and how should one guide price experimentation, given that such information has already been collected.

\noindent \textbf{Pricing mechanisms and performance.} In this problem, the decision-maker chooses a (potentially) randomized pricing mechanism characterized by the cdf of prices that the seller posts. Formally, we let $\MechSet = \{\Psi \in \distSet \}$ be the set of randomized mechanisms, where $\distSet$ is the set of probability measures on $[\lb,\ub]$. Then, for every $\Psi \in \MechSet$ and any distribution of consumer value $F$, the expected revenue generated from the randomized mechanism is defined as,
 \begin{equation*}
\Rev\left( \Psi,  F \right) = \int_{\lb}^{\ub}\left[\int_{\lb}^{\ub} p \mathbbm{1}\{v \geq p\} d F(v)\right]  d \Psi(p) = \int_{\lb}^{\ub} p \cdot \mathbb{P}_{v \sim F}\left( v \geq p \right) d \Psi(p) = \int_{\lb}^{\ub} \Rev\left(p\vert  F \right) d \Psi(p),
\end{equation*}
where we let  $\Rev \left( p \vert  F\right) = p \cdot \mathbb{P}_{v \sim F}\left( v \geq p \right).$ To ease notation, let $\overline F(p)=1-F(p)$ denote the complementary cumulative distribution function. Since a distribution may have atoms, the sale probability at price $p$ is given by the left limit $\mathbb{P}_{v\sim F}(v\ge p)=\overline F(p-)$. Thus, $\Rev(p\vert F)=p \cdot \overline F(p-)$

In an ideal scenario where the seller has \textit{full knowledge} of the distribution of values $F$, the optimal revenue is obtained by a posted price mechanism \citep{Riley1983} and we denote it by $\opt(F)= \sup_{p \in [\lb,\ub]} \Rev\left(p\vert F\right).$ We also define the set of optimal prices as $\mathcal{O}(F) = \arg \max_{p \in [\lb,\ub]} \Rev\left(p\vert F\right)$\footnote{We note that $\mathcal{O}(F)$ is non-empty as $\Rev\left( \cdot \vert  F\right)$ is upper semi-continuous on a compact segment.}. In this work, we aim at understanding the revenue gap incurred by the seller due to not knowing the actual distribution. Given a mechanism $\Psi \in \MechSet$ and a value distribution $F$, we define this gap through the following ratio, 
\begin{equation*}
\label{eq:gap}
\Ratio = \frac{\int_{\lb}^{\ub} \Rev\left(p\vert F\right) d \Psi(p) }{\operatorname{opt}(F)}.
\end{equation*}

\noindent \textbf{Performance for a given information set.} 
Our  goal  is to understand the spectrum of achievable performances for various information sets $\Iset$. In particular, we assume that the seller is given the information set $\Iset$. We then define the performance of a pricing mechanism $\Psi \in \MechSet$ as the worst-case ratio, where the worst-case is taken over a class of admissible value distributions denoted as $\Cb(\Iset)$ and defined as,
\begin{equation*}
\Cb(\Iset) = \left \{ F \in \Cb \text{ s.t. }  \bF(p_i-)= q_i \text{ for all } (p_i, q_i) \in \Iset \right \}.
\end{equation*}
We note that the set $\Cb(\Iset)$ contains all value distributions in $\Cb$ which are consistent with the information set $\Iset$ observed by the decision-maker.

Consequently, the worst-case performance of any pricing mechanism $\Psi \in \MechSet$ given $\Iset$ is the value of the following optimization problem.
\begin{equation}
\label{eq:problem_nature}
\inf_{F \in \Cb(\Iset)} \Ratio.
\end{equation}
In what follows, we will refer to \eqref{eq:problem_nature} as Nature's problem. We note that a high value of this worst-case ratio indicates that the information set $\Iset$ is highly informative in terms of revenue maximization and that the policy $\Psi \in \MechSet$ used by the decision-maker effectively takes advantage of the information in $\Iset$. 

Although  \eqref{eq:problem_nature} quantifies the performance of a \textit{fixed} policy, we also study the intrinsic value of certain information sets by computing the best achievable worst-case ratio across all pricing mechanisms. This quantity is formally defined as,
\begin{equation}
\tag{Maximin}
\label{eq:minimax}
\MMRatio = \sup_{\Psi \in \MechSet} \inf_{F \in \Cb(\Iset)} \Ratio.
\end{equation}

\noindent \textbf{Distribution classes.}
In this work, we derive results on the performance of pricing mechanisms for two classes of distributions. 
The class of general distributions denoted as $\Gb$ that contains all distributions supported on $[\lb,\ub]$. 
 We also consider the subset of $\Gb$ denoted as $\Fb$ and which denotes the set of regular distributions, i.e., distributions that admit a density, except potentially at the upper boundary of their support and that have a non-decreasing virtual value $v - \frac{1-F(v)}{f(v)}$. We note that the  class of regular distributions contains a wide variety of distributions; see, e.g., \cite{ewerhart2013regular}. This class is central to the pricing and mechanism design literature and can be alternatively described as the class of distributions that induce a concave revenue function in the quantity space.

\section{Reduction of Nature's Problem}\label{sec:Noptpricing}

To evaluate the performance of a pricing mechanism, one needs to solve Nature's optimization problem in \eqref{eq:problem_nature}. This task is \textit{a priori} challenging, as \eqref{eq:problem_nature} is an infinite-dimensional optimization problem over the nonparametric class $\Cb(\Iset)$. The constraints defining $\Cb(\Iset)$ combine (i) compatibility with the historical information set $\Iset$ through the equalities $\overline F(p_i-)=q_i$, and (ii) potential structural restrictions on the value distribution, such as regularity when $\Cb=\Fb$. 

In this section, we focus on the worst-case ratio objective in \eqref{eq:problem_nature} and show that Nature’s infinite-dimensional optimization problem can be reduced to a one-dimensional search over a family of extremal distributions parameterized by a feasible optimal price. Our proof generalizes the reasoning developed in \citet{ABBSinglePoint} for one quantile and proceeds in three steps. First, in \Cref{sec:shape_distribution}, we characterize the shape of distributions that are feasible given the information set $\Iset$ and derive envelope bounds on their cumulative distribution functions. Second, in \Cref{sec:candidate_worst_dist}, we leverage these bounds to identify a family of candidate worst-case distributions indexed by feasible optimal price-conversion pairs. Third, in \Cref{sec:main_reduction}, we combine these ingredients to show that Nature’s problem reduces to a one-dimensional optimization over feasible optimal prices. While our exposition focuses on the ratio objective for clarity, the same argument extends to regret-based performance measures. In particular, in \Cref{sec:apx_Proof_Thm1}  we present a more general formulation  that also captures the absolute regret metric.

\paragraph{Notation.}
For any $F \in \distSet$ define the generalized inverse
$F^{-1}(q) = \inf \left\{ v \in [\lb,\ub] : F(v) \ge q \right\},$ for every $q \in [0,1].$
For a distribution with a positive density $f$ on its support $[a,b]$, where $0 \le a < \infty$ and $a \le b \le \infty$, we define the virtual value function for $v \in [a,b]$ by
$\phi_F(v) = v - \frac{\overline F(v)}{f(v)}.$
Finally, for $v \ge 0$, define $\Gamma(v) = \frac{1}{1+v}$ and its inverse $\Gamma^{-1}(q) = \frac{1}{q} - 1$ for $q \in (0,1]$.

\subsection{Shape of Feasible Distributions}
\label{sec:shape_distribution}
Our construction relies on identifying extremal distributions that are consistent with a given pair of price–conversion observations. Consider two prices $s \le s'$ with associated conversion rates $q=\overline F(s-)$ and $q'=\overline F(s'-)$. We next show that any distribution $F \in \Cb$ that is consistent with these observations must satisfy structural bounds between the two prices. In particular, there exist extremal ccdfs that represent the smallest and largest possible values of $\overline F(v)$ compatible with these two observations and with the distribution class $\Cb$. These extremal objects will serve as local building blocks for constructing global envelope distributions that bound all feasible distributions in $\Cb(\Iset)$. To formalize this idea, for any pair $(s,s')$ such that $0 \le s < s'$ and $1 \ge q \ge q' >0$, we define on $[\lb,\ub]$ the following functions
\begin{align} 
\overline{G}_{\Gb}(v \vert  (s,q),(s',q')) &= \mathbbm{1} \{v \in [\lb,s)\} + \mathbbm{1} \{v \in [s,\ub)\} \cdot q' \label{eq:GG} \\
\overline{G}_{\Fb}(v \vert  (s,q),(s',q')) &=  \Gad{\Gainv{q} + \frac{\Gainv{q'}-\Gainv{q}}{s'-s}(v-s)}. \label{eq:GF}
\end{align}
Furthermore, we set the convention that $\overline{G}_{\Fb}(v \vert  (s,q),(s',0)) = q \cdot \mathbbm{1} \left \{ v \leq s \right\}$ when $q' =0$.

These functions are motivated by the fact that they allow us to derive the following bounds on distributions in $\Cb$ by deriving a so-called single crossing property. 
\begin{lemma}[Local Bounds]\label{lemma:singlecross}
	Fix $\Cb \in \{\Fb,\Gb\}$. Let $F$ in $\Cb$ and values $(s,s')$ such that $\lb \le s < s'$ and $0 < q' = \bF(s'-) \le q = \bF(s-)$. For every $v \in [\lb, \ub]$, we have that,
	\begin{align*}
	\overline{F}(v)&\ge  \overline{G}_{\Cb}(v \vert  (s,q),(s',q'))  
	 \qquad \mbox{ if } \:  v \in[s, s') \\
	\overline{F}(v) &\le \overline{G}_{\Cb}(v \vert  (s,q),(s',q')) 
	 \qquad \mbox{ if } \:  v \geq s' \text{ or } v < s.
	\end{align*}	
\end{lemma}
\Cref{lemma:singlecross} shows that $\overline{G}_{\Cb}$ provides extremal bounds for any distribution in $\Cb$ sharing the same conversion rates at $s$ and $s'$. Given that all distributions in $\Cb(\Iset)$ share the same conversion rates at the prices $\mathbf{p}$, we combine these local bounds across price intervals to construct global envelope distributions that bound all feasible ccdfs.
In particular, given a feasible information set $\Iset = \{(p_i, q_i); \; i \in \{0, \ldots, N+1 \} \}$ we define, $G_{\Cb,i}(v) = {G}_{\Cb}(v \vert  (p_{i},q_{i}),(p_{i+1},q_{i+1}))$ for every $i \in \{ 0, \ldots, N \}$. 
For both the regular and general classes of distributions, we define on $[\lb, \ub)$: 
\begin{equation}
\label{eq:L}
\overline{L}_{\Cb}(v \vert  \Iset) = \overline{G}_{\Cb,i}(v), \quad \mbox{if } v \in [p_{i},p_{i+1}), \mbox{ for } i \in \{ 0, \ldots, N\}.
\end{equation}
Furthermore, we let
\begin{align*}
\overline{U}_{\Gb}(v \vert  \Iset) &= \overline{G}_{\Gb,i-1}(v ) = q_i \quad \mbox{if } v \in [p_{i}, p_{i+1}), \mbox{for } i \in \{0, \ldots, N\}.\\
\overline{U}_{\Fb}(v \vert  \Iset) &=\begin{cases}
    \min\{1, \overline{G}_{\Fb,1}(v) \} &\quad \mbox{if } v \in [p_0,p_{1})\\
	 \min\{\overline{G}_{\Fb,{i-1}}(v),\overline{G}_{\Fb,{i+1}}(v) \} &\quad \mbox{if } v \in[p_{i},p_{i+1}), \mbox{for } i \in \{ 1, \ldots, N-2\}\\
	  \overline{G}_{\Fb,{i-1}}(v) &\quad \mbox{if } v \in[p_{i},p_{i+1}), \mbox{ for } i \in \{N-1, N\}.
	  \end{cases}  
\end{align*}
Intuitively, \Cref{lemma:singlecross} implies that the ccdf $\overline{L}_{\Cb}(\cdot \vert  \Iset)$ is a lower bound of the ccdf of any distribution in $\Cb(\Iset)$. Similarly, the ccdf $\overline{U}_{\Cb}(\cdot \vert  \Iset)$ provides an upper bound. We illustrate this in \Cref{fig_feasible}. The information set is $\{(p_i,q_i), \, i \in \{1,\ldots, 4\} \} = \{ (0.47,0.25), (0.70,0.055), (0.89,0.017), (0.98,0.006)   \}$ and we set $\lb = 0, \, \ub =1$.
\begin{figure}[h!]
\centering
% =======================
% (a) CCDF bounds
% =======================
\begin{subfigure}{0.48\textwidth}
\centering
\begin{tikzpicture}
\begin{axis}[
    xmin=0, xmax=1.1,
    xlabel={$p$},
    ylabel={ccdf $\overline{F}(p)$},
    ylabel style={font=\scriptsize},
    grid=both,
    unbounded coords=jump,
    table/col sep=comma,
    legend style={
    font=\tiny,
    row sep=2pt,
    draw=none,
    fill=white,
    fill opacity=0.8,
    text opacity=1,
},
]

\addplot[name path=ccdfUreg, draw=none, forget plot]
  table[x=v, y=ccdf_upper_reg] {Data/figure1_feasible_bounds.csv};
\addplot[name path=ccdfLreg, draw=none, forget plot]
  table[x=v, y=ccdf_lower_reg] {Data/figure1_feasible_bounds.csv};

% General (blue): paths
\addplot[name path=ccdfUgen, draw=none, forget plot]
  table[x=v, y=ccdf_upper_gen] {Data/figure1_feasible_bounds.csv};
\addplot[name path=ccdfLgen, draw=none, forget plot]
  table[x=v, y=ccdf_lower_gen] {Data/figure1_feasible_bounds.csv};

% ---------- Shaded feasible regions (background)
\addplot[fill=green!50!black, fill opacity=0.20, draw=none, forget plot]
  fill between[of=ccdfUreg and ccdfLreg];
\addplot[fill=cyan!50, fill opacity=0.18, draw=none, forget plot]
  fill between[of=ccdfUgen and ccdfLgen];

% ---------- Visible envelope curves
% Regular (green): L (solid) and U (dashed)
\addplot[thick, green!60!black, forget plot]
  table[x=v, y=ccdf_lower_reg] {Data/figure1_feasible_bounds.csv};
\addplot[thick, dashed, green!60!black, forget plot]
  table[x=v, y=ccdf_upper_reg] {Data/figure1_feasible_bounds.csv};

% General (blue): L (solid) and U (dashed)
\addplot[thick, cyan!70!black, forget plot]
  table[x=v, y=ccdf_lower_gen] {Data/figure1_feasible_bounds.csv};
\addplot[thick, dashed, cyan!70!black, forget plot]
  table[x=v, y=ccdf_upper_gen] {Data/figure1_feasible_bounds.csv};

% ---------- Information set (red dots)
\addplot[
  only marks,
  mark=*,
  mark size=1.5pt,
  draw=red,
  fill=red,
  forget plot
] table[x=v, y=info_ccdf] {Data/figure1_feasible_bounds.csv};

% ---------- Legend (manual, encodes class by color and envelope by line style)
\addlegendimage{area legend, fill=green!50!black, fill opacity=0.20, draw=none}
\addlegendentry{Regular class}
\addlegendimage{area legend, fill=cyan!50, fill opacity=0.18, draw=none}
\addlegendentry{General class}
\addlegendimage{thick, green!60!black}
\addlegendentry{$\overline L_{\Fb}(\cdot\vert \Iset)$}
\addlegendimage{thick, dashed, green!60!black}
\addlegendentry{$\overline U_{\Fb}(\cdot\vert \Iset)$}
\addlegendimage{thick, cyan!70!black}
\addlegendentry{$\overline L_{\Gb}(\cdot\vert \Iset)$}
\addlegendimage{thick, dashed, cyan!70!black}
\addlegendentry{$\overline U_{\Gb}(\cdot\vert \Iset)$}
\addlegendimage{only marks, mark=*, mark options={draw=red, fill=red}}
\addlegendentry{Information set}

\end{axis}
\end{tikzpicture}
\caption{Bounds on the CCDF values.}
\end{subfigure}
\hfill
% =======================
% (b) Revenue bounds
% =======================
\begin{subfigure}{0.48\textwidth}
\centering
\begin{tikzpicture}
\begin{axis}[
    xmin=0, xmax=1.1,
    xlabel={$p$},
    ylabel={Revenue $p\,\overline{F}(p-)$},
    ylabel style={font=\scriptsize},
    grid=both,
    table/col sep=comma,
    unbounded coords=jump,
    legend style={
    font=\tiny,
    row sep=2pt,
    draw=none,
    fill=white,
    fill opacity=0.8,
    text opacity=1,
},
]

% Regular (green)
\addplot[name path=revUreg, draw=none, forget plot]
  table[x=v, y=rev_upper_reg] {Data/figure1_feasible_bounds.csv};
\addplot[name path=revLreg, draw=none, forget plot]
  table[x=v, y=rev_lower_reg] {Data/figure1_feasible_bounds.csv};
\addplot[fill=green!50!black, fill opacity=0.35, draw=none, forget plot]
  fill between[of=revUreg and revLreg];

% General (blue)
\addplot[name path=revUgen, draw=none, forget plot]
  table[x=v, y=rev_upper_gen] {Data/figure1_feasible_bounds.csv};
\addplot[name path=revLgen, draw=none, forget plot]
  table[x=v, y=rev_lower_gen] {Data/figure1_feasible_bounds.csv};
\addplot[fill=cyan!50, fill opacity=0.25, draw=none, forget plot]
  fill between[of=revUgen and revLgen];

% Information set (red dots)
\addplot[
  only marks,
  mark=*,
  mark size=1.5pt,
  draw=red,
  fill=red,
  forget plot
] table[x=v, y=info_rev] {Data/figure1_feasible_bounds.csv};

% Legend (manual, clean)
\addlegendimage{area legend, fill=green!50!black, fill opacity=0.35, draw=none}
\addlegendentry{Feasible Revenue (Regular)}
\addlegendimage{area legend, fill=cyan!50, fill opacity=0.25, draw=none}
\addlegendentry{Feasible Revenue (General)}
\addlegendimage{only marks, mark=*, mark options={draw=red, fill=red}}
\addlegendentry{Information Set}

\end{axis}
\end{tikzpicture}
\caption{Bounds on the revenue values.}
\end{subfigure}

\caption{ \textbf{Illustration of feasible ccdf and revenue bounds implied by an information set.}
The red dots correspond to the observed conversion rates $(p_i,q_i)$.
Panel (a) shows the set of ccdfs $\overline F$ consistent with the information set.
Panel (b) shows the corresponding set of feasible revenue curves.
The shaded regions indicate the envelopes induced by the regular (green) and general (blue) classes of distributions.}
\label{fig_feasible}
\end{figure}

\Cref{fig_feasible} illustrates how the information set and structural assumptions jointly restrict the set of admissible distributions.
Regularity substantially tightens the feasible region relative to the general class, both at the level of ccdfs and revenue curves.
This characterization is made possible by the single-crossing property established in \Cref{lemma:singlecross}, which allows local bounds to be aggregated into global envelopes.
Importantly, under regularity the resulting feasible sets exhibit a nontrivial geometry, as the envelopes need not be piecewise constant and may interact across adjacent price intervals.
As a consequence, characterizing extremal revenue behavior requires careful construction of envelope distributions rather than pointwise optimization.
The envelope constructions $\overline L_{\Cb}$ and $\overline U_{\Cb}$ therefore characterize the extremal behavior within each class and serve as natural candidates for worst-case distributions.

\subsection{Candidate Worst-case Distributions}
\label{sec:candidate_worst_dist}

Building on this characterization, we now turn to Nature’s Problem.
In particular, the envelope constructions allow us to sharply describe the distributions that are optimal responses to a given price choice.
We begin by formalizing which price–conversion pairs can arise as optimal under distributions consistent with the information set.

\begin{definition}[Feasible optimal price-conversion pair]
\label{def:feasible_pair}
Given an information set $\Iset$, a pair $(\rs,\qs) \in [\lb,\ub]\times[0,1]$ is said to be \emph{feasible} if there exists a distribution $F \in \Cb(\Iset)$ such that $\rs \in \mathcal{O}(F)$ and $\overline F(\rs-)=\qs$. We denote by ${\cal B}(\Iset)$ the set of all feasible price--conversion pairs.
\end{definition}

This notion allows us to decompose Nature’s optimization problem according to the optimal price and its associated conversion rate. Indeed, for any pricing mechanism $\Psi$, Nature’s objective can be rewritten as
\begin{equation}
\label{eq:nature_ratio_decomp}
\inf_{F \in \Cb(\Iset)} 
\frac{\Expect_{\Psi}\!\left[\Rev(p\vert F)\right]}{\opt(F)}
=
\inf_{(\rs,\qs) \in {\cal B}(\Iset)}
\left\{
\frac{1}{\rs \cdot \qs}
\inf_{\substack{F \in \Cb(\Iset):\\ \rs \in \mathcal{O}(F),\; \overline F(\rs-)=\qs}}
\Expect_{\Psi}\!\left[\Rev(p\vert F)\right]
\right\}.
\end{equation}
The denominator $\rs \cdot \qs$ corresponds to the optimal revenue of any distribution whose optimal price is $\rs$ with associated conversion rate $\qs$. The numerator is the smallest revenue that Nature can induce among all distributions consistent with $\Iset$ and this fixed optimal price-conversion pair.
The following proposition shows that this inner minimization problem admits a sharp solution. Fix a feasible pair $(\rs,\qs) \in \mathcal{B}(\Iset)$, i.e., a price-conversion pair that can arise as an optimal solution for some distribution $F \in \Cb(\Iset)$ consistent with the historical data, and consider the extended information set $\Iset \cup \{(\rs,\qs)\}$ obtained by augmenting the original dataset with this new pair. Among all distributions in $\Cb(\Iset)$ that share $(\rs,\qs)$ as an optimal price-conversion pair, the worst-case choice for the expected revenue under mechanism $\Psi$ is given by the lower-envelope distribution $L_{\Cb}(\cdot)$ associated to the extended information set.

\begin{proposition}[Worst-case revenue under a fixed optimal price]
\label{prop:worst_case_revenue}
Assume that $\Cb(\Iset)$ is non-empty. For any feasible pair $(\rs,\qs) \in {\cal B}(\Iset)$ and any pricing mechanism $\Psi \in \MechSet$,
\begin{equation*}
\inf_{\substack{F \in \Cb(\Iset): \\ \rs \in \mathcal{O}(F),\; \overline F(\rs-)=\qs}}
\Expect_{\Psi}\!\left[\Rev(p\vert  F)\right]
=
\Expect_{\Psi}\!\left[\Rev\!\left(p \vert  L_{\Cb}(\cdot \vert  \Iset \cup (\rs,\qs))\right)\right].
\end{equation*}
\end{proposition}

\noindent
Proposition~\ref{prop:worst_case_revenue} implies that the inner minimization problem in \eqref{eq:nature_ratio_decomp} admits an explicit solution. As a result, Nature’s problem reduces to an optimization over the two-dimensional set of feasible price-conversion pairs $(\rs,\qs) \in \mathcal{B}(\Iset)$. In particular, the worst-case distribution associated with any feasible pair is characterized by the lower-envelope construction.

\subsection{Reduction of Nature’s problem and Main Characterization}
\label{sec:main_reduction}

In what follows, we leverage the reduction developed in \Cref{prop:worst_case_revenue} and complete the reduction of Nature’s problem to a one-dimensional optimization problem.

We begin by introducing the final objects needed to state the result. Recall the set $\mathcal{B}(\Iset)$ of feasible price-conversion pairs defined in \Cref{def:feasible_pair}. We define the associated set of feasible prices
\begin{equation*}
\mathcal{S}(\Iset)
=
\left\{
\rs \in [\lb,\ub]
\;\vert  \;
\text{there exists } \qs \in [0,1] \text{ such that } (\rs,\qs) \in \mathcal{B}(\Iset)
\right\},
\end{equation*}
that is, the set of prices that can arise as optimal under at least one value distribution consistent with the information set $\Iset$.

Furthermore, for any $\rs \in \mathcal{S}(\Iset)$, we define the distribution $F_{\Cb}(\cdot \vert  \rs,\Iset)$ by its complementary cumulative distribution function
\begin{equation}
\label{eq:worst_case_F}
\overline{F}_{\Cb}(v \vert  \rs,\Iset)
=
\overline{L}_{\Cb}\!\left(
v \;\vert  \;
\Iset \cup \bigl(\rs,\overline{U}_{\Cb}(\rs- \vert  \Iset)\bigr)
\right).
\end{equation}
Equation~\eqref{eq:worst_case_F} defines a family of distributions parametrized by a candidate optimal price $\rs$. Our main technical result shows that this one–parameter family of extremal distributions is sufficient to evaluate the worst-case ratio defined in \Cref{eq:problem_nature}.

 \begin{theorem}[Reduction for policy evaluation]\label{thm:nature_reduction} 
Assume $\Cb$ is equal to $\Gb$ or $\Fb$. Fix $N \geq 1$ and let $\Iset$ be an information set including $N$ historical prices and conversion rates.  Assume $\Cb(\Iset)$ is non-empty. Then, for any mechanism $\Psi$ in $\MechSet$,
\begin{equation*}
\inf_{F \in \Cb(\Iset)}  \Ratio = \inf_{\rs \in \mathcal{S}(\Iset)}  \Ratio[\Psi,F_{\Cb}( \cdot \vert  \rs, \Iset)]. 
\end{equation*}
\end{theorem}

For any pricing mechanism, \Cref{thm:nature_reduction} reduces Nature’s original infinite-dimensional optimization problem in \eqref{eq:problem_nature} to a one-dimensional search over the set $\mathcal{S}(\Iset)$ of feasible prices. This reduction provides a tractable way to evaluate and compare pricing mechanisms under arbitrary historical information sets, without requiring parametric assumptions on the value distribution. This result significantly generalizes \cite[Theorem 1]{ABBSinglePoint} along several dimensions. First, our reduction applies not only to regular distributions but also to the class of general distributions, allowing us to study settings where the decision-maker does not impose shape constraints. More importantly, our characterization holds for information sets containing an arbitrary number of historical prices $N$, whereas the analysis in \cite{ABBSinglePoint} only applies to the case $N=1$. This extension considerably enlarges the class of information sets $\Iset$ that can be analyzed and, as we will see in later sections, enables new insights and prescriptions for price experimentation. Finally, although the main exposition focuses on the ratio objective, the underlying proof technique extends to regret-based performance measures. In particular, the same reduction applies to the absolute regret objective, as shown in \Cref{thm:lambda_reduction}. 

\paragraph{Characterizing feasible optimal prices.}
We note that the reduction in \Cref{thm:nature_reduction} shows that evaluating the worst-case performance of a pricing mechanism reduces to a one-dimensional search over prices $\rs \in \mathcal{S}(\Iset)$. In order to implement this reduction, it is therefore useful to characterize when a candidate price $\rs$ belongs to the set $\mathcal{S}(\Iset)$.
Intuitively, a price $\rs$ is feasible if there exists a distribution consistent with the information set under which $\rs$ is an optimal price. The next result provides a simple characterization of this condition in terms of the upper-envelope conversion rate $\overline{U}_{\Cb}(\rs-\vert\Iset)$.

\begin{proposition}\label{prop:r_certificate}
Fix $\rs \in [\lb,\ub]$. Then $\rs \in \mathcal{S}(\Iset)$ if and only if, 
$i)$ $\Cb(\Iset \cup(\rs,\overline{U}_{\Cb}(\rs- \vert  \Iset)) )$ is non-empty, and $ii)$ $\rs \cdot \overline{U}_{\Cb}(\rs -\vert  \Iset) \geq \max_{i \in \{0,\ldots,N+1\}} p_i \cdot q_i$.
\end{proposition}

\Cref{prop:r_certificate} provides a practical certificate for testing whether a candidate price $\rs$ belongs to $\mathcal{S}(\Iset)$. In particular, condition $(ii)$ evaluates feasibility under the most favorable demand at $\rs$ compatible with the data, and therefore determines whether $\rs$ can be optimal for some admissible distribution. Leveraging this characterization, we show in \Cref{sec:apx_S} that the set $\mathcal{S}(\Iset)$ admits a simple structure. When $\Cb=\Fb$, the feasible optimal prices form a single interval around the price achieving the largest observed revenue $\max_i p_i \cdot q_i$. In contrast, when $\Cb=\Gb$, the set $\mathcal{S}(\Iset)$ may consist of a union of intervals. In both cases, we provide explicit closed-form expressions for $\mathcal{S}(\Iset)$ that can be computed directly from the information set $\Iset$, and we use these in our numerical procedure.

\section{Near-optimal Policy and Value of Information}\label{sec:value_information}

In this section, we develop a tractable framework for characterizing the intrinsic value of an information set and constructing near-optimal robust pricing policies. We first use \Cref{thm:nature_reduction} to reformulate the maximin problem as an infinite-dimensional linear program with a one-dimensional continuum of constraints indexed by the oracle price. We then exploit the structural properties of least favorable distributions to replace this continuum by a finite family of constraints defined on an admissible grid. This leads to a discretized linear program whose solution provides both a computable lower bound on the intrinsic value and a near-optimal pricing policy with an explicit approximation guarantee.

\subsection{Intrinsic Value and a Reduced Maximin Formulation}
\label{sec:continuous_epi}

Recall that the maximin value
\begin{equation*}
\sup_{\Psi \in \MechSet} \inf_{F \in \Cb(\Iset)} \frac{\Rev(\Psi,F)}{\opt(F)}
\end{equation*}
represents the best achievable worst-case performance ratio across all pricing mechanisms that have access to
the information encoded in $\Iset$. We refer to this quantity as the \emph{intrinsic value} of $\Iset$, as it
captures the maximal performance guarantee attainable under informational constraints alone, independently of
any particular pricing rule.

To analyze this quantity, it is convenient to work with an epigraph formulation based on the notion of $\lambda$-regret. 
For a mechanism $\Psi$ and a distribution $F$, define the $\lambda$-regret as
\begin{equation*}
\LRegret[\lambda][\Psi,F]
=
\lambda \cdot \opt(F)
-
\Rev(\Psi,F).
\end{equation*}
The maximin problem can then be
equivalently written as follows (see \Cref{lem:epi_ratio}).
\begin{subequations}
\label{eq:epigraph}
\begin{alignat}{2}
&\sup_{\Psi(\cdot) \in \MechSet,\; \lambda \in [0,1]} &\qquad& \lambda \\
&\text{subject to} & & \LRegret \le 0 \quad \text{for all } F \in \Cb(\Iset).
\end{alignat}
\end{subequations}
The formulation in \eqref{eq:epigraph} is an infinite-dimensional linear optimization problem. The decision
variable $\Psi(\cdot)$ ranges over an infinite-dimensional space of randomized pricing mechanisms, and the
constraint set is indexed by $\Cb(\Iset)$, an infinite family of distributions. Hence, directly
characterizing optimal mechanisms or the intrinsic value of an information set is not
tractable without further structure.

Applying \Cref{thm:nature_reduction}, we can substantially simplify the constraint set and obtain the
equivalent formulation
\begin{subequations}
\label{eq:epigraph_simplified}
\begin{alignat}{2}
&\sup_{\Psi(\cdot) \in \MechSet,\; \lambda \in [0,1]} &\qquad& \lambda \\
&\text{subject to} 
& & 
\LRegret[\lambda][\Psi, F_{\Cb}(\cdot \vert  \rs,\Iset)] \le 0 
\quad \text{for all } \rs \in \mathcal{S}(\Iset),
\end{alignat}
\end{subequations}
where $F_{\Cb}(\cdot \vert  \rs,\Iset)$ denotes the least favorable distribution associated with oracle price
$\rs$. After the reduction, the constraint takes the form of a continuum of scalar inequalities indexed by
$\rs$, which can be interpreted as candidate oracle prices selected by nature.

While \eqref{eq:epigraph_simplified} represents a considerable simplification relative to
\eqref{eq:epigraph}, it still involves infinitely many constraints and infinitely many decision variables and care is required in constructing finite-dimensional approximations that yield valid
performance guarantees. In particular, we will see in \Cref{sec:discretization} that a naive discretization of \eqref{eq:epigraph_simplified} does not, in
general, yield a certified guarantee for the resulting mechanism.

\subsection{Discrete Mechanisms and Continuum Constraints}
\label{sec:discretization}

We begin by showing that one can restrict attention to randomized mechanisms supported on the set $\mathcal{S}(\Iset)$ of feasible optimal prices (cf. \Cref{prop:support_reduction}). This reduction follows from the fact that prices outside $\mathcal{S}(\Iset)$ are dominated in terms of revenue across all distributions in $\Cb(\Iset)$.
Fix an increasing sequence of $M+1$ real numbers $\IN=\{a_i\}_{i=0}^{M}$ that forms a partition of
$\mathcal{S}(\Iset)$ (we refer the reader to \Cref{def:partition_grid} for a formal definition). 
We restrict attention to pricing mechanisms whose induced distribution over prices is discrete and supported on $\IN$. Formally, define
\begin{equation*}
\MechSet_{\IN}
=
\left\{
\Psi \in \MechSet :
\Psi(x) = \sum_{j=0}^{M} \psi_j \cdot \mathbbm{1}\{x \ge a_j\},
\quad
\psi_j \ge 0,
\quad
\sum_{j=0}^{M} \psi_j = 1
\right\}.
\end{equation*}

For any $\Psi\in\MechSet_{\IN}$ and any $\rs\in\mathcal{S}(\Iset)$, the $\lambda$-regret constraint in
\eqref{eq:epigraph_simplified} can be written as
\begin{equation}\label{eq:constraint_expanded_discrete}
\lambda \cdot \opt \big(F_{\Cb}(\cdot\vert  \rs,\Iset)\big)
\;-\;
\sum_{j=0}^{M}
\psi_j \cdot \Rev\!\big(a_j \vert  F_{\Cb}(\cdot\vert  \rs,\Iset)\big)
\;\le\; 0,
\qquad \forall \rs\in\mathcal{S}(\Iset).
\end{equation}
Discretizing the mechanism therefore yields a continuum of linear inequalities indexed by $\rs$, with
coefficients that depend on $\rs$ through the oracle benchmark
$\opt(F_{\Cb}(\cdot\vert  \rs,\Iset))$ and the revenue terms
$\Rev(a_j \vert  F_{\Cb}(\cdot\vert  \rs,\Iset))$. 
A central difficulty is therefore to replace this continuum by finitely many constraints without violating feasibility of \eqref{eq:constraint_expanded_discrete}.

\noindent \textbf{An illustrative benchmark: the general class $\Cb=\Gb$.}
We first illustrate the idea under the general class $\Cb=\Gb$. In this case, we show that both $\opt(F_{\Gb}(\cdot\vert \rs,\Iset))$ and $\Rev(a_j \vert F_{\Gb}(\cdot\vert \rs,\Iset))$ are non-decreasing as $\rs$ varies over any interval $[p_k,p_{k+1})$ induced by the information set (see  \Cref{lem:cellwise_monotonicity}). Consequently, if there exists $k$ such that
$[a_i,a_{i+1}) \subset [p_k,p_{k+1})$, then the constraint
\begin{equation}
\label{eq:discrete_constraint_naive}
\lambda \cdot \opt \big(F_{\Gb}(\cdot \vert  a_{i+1}^{+},\Iset)\big)
\;-\;
\sum_{j=0}^{M}
\psi_j \cdot \Rev\!\big(a_j \vert  F_{\Gb}(\cdot\vert  a_{i},\Iset)\big)
\;\le\; 0
\end{equation}
implies
\begin{equation*}
\lambda \cdot \opt \big(F_{\Gb}(\cdot\vert  \rs,\Iset)\big)
\;-\;
\sum_{j=0}^{M}
\psi_j \cdot \Rev\!\big(a_j \vert  F_{\Gb}(\cdot\vert  \rs,\Iset)\big)
\;\le\; 0,
\qquad \forall \rs \in [a_i,a_{i+1}).
\end{equation*}
Thus, in this setting, monotonicity ensures that verifying the constraint at appropriately chosen boundary points certifies feasibility throughout the interval. When $N=1$, this property is also true for the class of regular distributions and, this observation underlies the discrete linear programs studied in \cite{ABBSinglePoint}. However, when $N >1$, the situation is more delicate under the regular class $\Cb=\Fb$ as we see next.

\noindent \textbf{Challenges for the regular class $\Cb=\Fb$.}
When $N>1$, the monotonicity property need not hold over $[p_k,p_{k+1})$ for the class of regular distributions. To understand this phenomenon recall the construction of the least favorable distribution
$F_{\Cb}(\cdot\vert  \rs,\Iset)$. For any $\rs\in\mathcal{S}(\Iset)$, the oracle benchmark admits the
representation
\begin{equation*}
\opt \big(F_{\Cb}(\cdot\vert  \rs,\Iset)\big)
=
\rs \cdot \overline{U}_{\Cb}(\rs-\vert \Iset),
\end{equation*}
where $\overline{U}_{\Cb}(\cdot\vert \Iset)$ denotes the upper envelope of feasible tail probabilities implied by
the information set $\Iset$ and the class $\Cb$. On each interval $[p_k,p_{k+1})$, this envelope is constructed as the pointwise minimum of two candidate functions described in \Cref{sec:shape_distribution}. For instance, the two branches induced by the two candidate functions are visible in \Cref{fig_feasible}(a) on the interval $[0.47,0.7]$.

The loss of monotonicity arises because the identity of the candidate function attaining this minimum may change as $\rs$ varies. When the active candidate switches within a grid cell $[a_i,a_{i+1})$, the benchmark and revenue coefficients need not move in a single direction, and boundary constraints no longer guarantee feasibility throughout the interval. This aspects necessitates a more careful discretization than the one proposed in \cite{ABBSinglePoint}, which we develop next.

\subsection{Admissible Grids and Certificate of Near-Optimality}
\label{sec:LP_result}
We now construct a discretization that restores the monotonicity structure locally and yields a certified finite-dimensional program. The key idea is to refine the grid so that, within each cell, the identity of the candidate function defining the envelope remains fixed.

Recall that, on each interval $[p_i,p_{i+1})$, the envelope
$\overline{U}_{\Cb}(\cdot\vert \Iset)$ is obtained as the pointwise minimum of finitely many candidate
functions $\{\overline{G}_{\Cb,\ell}(\cdot)\}$.
For $\Cb=\Fb$, the relevant candidates on $[p_i,p_{i+1})$ are
$\overline{G}_{\Fb,i-1}(\cdot)$ and $\overline{G}_{\Fb,i+1}(\cdot)$, and
\begin{equation*}
\overline{U}_{\Fb}(v\vert \Iset)
=
\min\!\left\{
\overline{G}_{\Fb,i-1}(v),
\overline{G}_{\Fb,i+1}(v)
\right\},
\qquad v\in[p_i,p_{i+1}).
\end{equation*}

A change in the active candidate occurs at any solution $c_i \in [p_i,p_{i+1}]$ of the equation
\begin{equation}
\label{eq:ci_def}
\overline{G}_{\Fb,i-1}(c_i)
=
\overline{G}_{\Fb,i+1}(c_i),
\end{equation}
whenever such a solution exists. If no solution lies in $[p_i,p_{i+1}]$, then the identity of the minimizer is constant over the interval. These switching cutoffs are precisely the points at which the envelope may lose its local monotonicity properties.
This observation motivates the following structural requirement on the grid.

\begin{definition}
An increasing grid $\IN=\{a_k\}_{k=0}^{M}$ is called \emph{admissible} if it partitions
$\mathcal{S}(\Iset)$ (as defined in \Cref{def:partition_grid}) and contains:
$(i)$ all information prices $\{p_i\}\cap\mathcal{S}(\Iset)$, and $(ii)$ all switching cutoffs $\{c_i\}\cap\mathcal{S}(\Iset)$.
\label{def:admissible}
\end{definition}
By construction, if $\IN$ is admissible, then on each cell $[a_i,a_{i+1})$ the same candidate function determines $\overline{U}_{\Fb}$. We note that when $\Cb = \Fb$ and $N=1$ or when $\Cb=\Gb$ admissibility reduces to requiring that $\IN$ contains the information prices.

We show in \Cref{lem:cellwise_monotonicity} that within each such cell both the oracle benchmark
$\rs \mapsto \opt(F_{\Cb}(\cdot\vert \rs,\Iset))$
and each revenue coefficient
$\rs \mapsto \Rev(a_j\vert F_{\Cb}(\cdot\vert \rs,\Iset))$
are monotone functions of $\rs$. Although the direction of monotonicity depends on the active candidate, extremal values are attained at the cell boundaries. This permits a conservative replacement of the continuum constraints by finitely many boundary constraints, as formalized below.

Fix an admissible grid $\IN$. For each cell $(i,j) \in \{0,\ldots,M-1\} \times \{0,\ldots,M\}$ define
\begin{align*}
\widehat{\opt}^{\Cb}_i
&=
\max\!\left\{
\opt(F_{\Cb}(\cdot\vert  a_i,\Iset)),
\opt(F_{\Cb}(\cdot\vert  a_{i+1}^-,\Iset))
\right\},\\
\widehat{\Rev}^{\Cb}_{ij}
&=
\min\!\left\{
\Rev(a_j\vert  F_{\Cb}(\cdot\vert  a_i,\Iset)),
\Rev(a_j\vert  F_{\Cb}(\cdot\vert  a_{i+1}^-,\Iset))
\right\}.
\end{align*}
Consider the linear program defined as
\begin{subequations}
\label{eq:minimax_lp}
\begin{alignat}{2}
\underline{\Lb}(\Cb(\Iset), \IN)
= \;& \sup_{\boldsymbol{\psi},\, \lambda} 
&\quad& \lambda \\[2pt]
& \text{s.t.}
& & \lambda \cdot \widehat{\opt}^{\Cb}_i
-
\sum_{j=0}^{M}
\psi_j \cdot \widehat{\Rev}^{\Cb}_{ij}
\le 0,
\qquad i=0,\ldots,M-1, \\[4pt]
& &
& \sum_{j=0}^{M} \psi_j = 1, \qquad
\psi_j \ge 0,\qquad
0 \le \lambda \le 1.
\end{alignat}
\end{subequations}
The program \eqref{eq:minimax_lp} is finite-dimensional and depends only on the grid $\IN$. 
The following result establishes that this discretized program yields a valid lower bound on the intrinsic value of the information set and quantifies the discretization error.

\begin{theorem} 
\label{thm:lp_discrete_bound} 
Fix an information set $\Iset$, a class $\Cb\in\{\Gb,\Fb\}$, and an admissible grid $\IN=\{a_k\}_{k=0}^{M}$ that partitions $\mathcal{S}(\Iset)$ with mesh size $\Delta(\IN) = \max_k (a_{k+1} - a_k).$ 
Let $(\boldsymbol{\psi}^*,\lambda^*)$ be an optimal solution of the linear program \eqref{eq:minimax_lp}.
Assume that $\max_{i \in \{0,\ldots,N+1\}} p_i \cdot q_i > 0$. 
Then there exists a constant $C > 0$, independent of $\IN$, such that 
\begin{equation*} 
\underline{\Lb}(\Cb(\Iset),\IN) 
\le 
\inf_{F\in\Cb(\Iset)} \frac{\Rev(\Psi_{\IN}^*,F)}{\opt(F)} 
\le 
\sup_{\Psi\in\MechSet} \inf_{F\in\Cb(\Iset)} \frac{\Rev(\Psi,F)}{\opt(F)} 
\leq \underline{\Lb}(\Cb(\Iset),\IN) + C \cdot \Delta(\IN), 
\end{equation*} 
where $\Psi_{\IN}^*$ is the randomized posted-price mechanism supported on $\IN$ with weights $\boldsymbol{\psi}^*$. 
\end{theorem}

\Cref{thm:lp_discrete_bound} formalizes that \eqref{eq:minimax_lp} provides a tractable procedure  to characterize the value of information and to develop near-optimal pricing policies for any information set $\Iset$ and any $N \geq 1$.
Indeed, the mechanism $\Psi_{\IN}^*$ induced by the LP solution, which we refer to as the \emph{discretized maximin policy}, is a feasible robust policy whose worst-case performance over $\Cb(\Iset)$ is certified by the LP objective value. The theorem further provides a quantitative bound on the approximation error: the gap between the optimal maximin value and the LP value is at most $C \cdot \Delta(\IN)$, where $C$ is independent of the grid (its exact expression is provided in \Cref{sec:apx_proof_thm2}). Consequently, the discretized program approximates the benchmark in a controlled manner. In particular, if the grid $\IN$ is chosen uniformly over $\mathcal{S}(\Iset)$ with $M$ interior intervals, then $\Delta(\IN)=O(1/M)$, and the approximation error decreases at rate $O(1/M)$. 
In \Cref{fig:maximin-optimal-policy},  we numerically illustrate the structure of the near-optimal maximin pricing policy returned by the discretized LP under both the regular and general assumptions.

\section{Applications to Price Experimentation}\label{sec:applications}
We described in \Cref{sec:value_information}  a tractable procedure to quantify the value of any historical information set that comes in the form of prices and their associated conversion rates.
In this section, we illustrate how this procedure can be leveraged to derive insights on price experimentation. In \Cref{sec:gradient}, we study the value of \textit{local} experimentation and consider a setting in which the decision-maker cannot substantially vary prices and therefore has only an estimate of the gradient at a given price along with the conversion rate at that price. In \Cref{sec:extra_point}, we investigate the value of \textit{global} experimentation and derive an optimal (freely selected) second price to experiment at  after receiving the conversion rate  from a first price. Finally, in \Cref{sec:ternary},  we illustrate how  our exact analysis could be used to considerably reduce the number of price experiments needed to reach a certain level of performance guarantee.

\subsection{Pricing with gradient information: the value of local experimentation}\label{sec:gradient}
In many practical settings, the seller cannot drastically change the offered price and therefore can only gather ``local'' information. Motivated by this, we consider the setting in which the seller knows the support of the demand function $[\lb, \ub]$ and has observed two demand points, i.e., $\Iset = \{(p_1, q_1), (p_{\epsilon}, q_{\epsilon})\}$. The two prices are such that $p_{\epsilon} = (1 + \epsilon) \cdot p_1$ for $\epsilon >0$.

\noindent \textbf{Value of the numerical gradient.} We propose to understand the value of such observations by leveraging the procedure described in \Cref{sec:value_information} to numerically compute the maximin ratio defined in \eqref{eq:minimax} when the demand function is regular (i.e., $\Cb = \Fb$) and its support is included in $[\lb, \ub] = [0, 100]$. We consider the setting in which the seller experiments at an increment of $\epsilon = 1\%$ of the observed price $p_1$ which we fix at $10$ and we consider different initial observed demand $q_1$ in $\{0.01, 0.1, 0.25, 0.5, 0.75, 0.9, 0.99\}$. 
For each initial $q_1$, we evaluate a lower bound on the maximin ratio for a grid of 200 possible values of $q_{\epsilon}$  and report the results in \Cref{fig:grad_value}. 
For the sake of interpretability, the x-axis of \Cref{fig:grad_value} is parametrized by $g_{\epsilon}$ which is defined as,
\begin{equation*}
g_{\epsilon} = \frac{p_{\epsilon}q_{\epsilon}-p_1  q_1}{p_{\epsilon}-p_1}.
\end{equation*}
We interpret $g_{\epsilon}$ as the gradient of the revenue curve at $p_1$ given that, $\nabla_p \Rev(p_1\vert F)  \approx \frac{p_{\epsilon}q_{\epsilon}-p_1  q_1}{p_{\epsilon}-p_1}$.

\begin{figure}[h!]
\centering
\begin{subfigure}{0.48\textwidth}
\centering
\begin{tikzpicture}
\begin{axis}[
    xmode=log,
    xmin=1e-6,
    xmax=2,
    ymin=25, ymax=100,
    xlabel={$-g_{\epsilon}$},
    ylabel={Worst-case ratio (\%)},
    grid=both,
    table/col sep=comma,
    legend style={font=\footnotesize},
    legend cell align=left,
    legend pos=south west
]

% Manual legend entries
\addlegendimage{thick, blue!70!black, mark=*}
\addlegendentry{$q_1=0.01$}

\addlegendimage{thick, orange!90!black, mark=square*}
\addlegendentry{$q_1=0.10$}

\addlegendimage{thick, teal!80!black, mark=triangle*}
\addlegendentry{$q_1=0.25$}

\addlegendimage{thick, red!75!black, mark=diamond*}
\addlegendentry{$q_1=0.50$}

\addlegendimage{thick, purple!80!black, mark=star}
\addlegendentry{$q_1=0.75$}

\addlegendimage{thick, brown!85!black, mark=o}
\addlegendentry{$q_1=0.90$}

\addlegendimage{thick, black, mark=x}
\addlegendentry{$q_1=0.99$}

% Plots
\addplot[
  thick, blue!70!black,
  mark=*,
  mark indices={20},
  restrict expr to domain={\thisrow{g_eps}}{-1000:0},
  restrict expr to domain={\thisrow{q1}}{0.009:0.011},
] table[x expr={-\thisrow{g_eps}}, y=ratio_pct] {Data/figure2_gradient_value.csv};

\addplot[
  thick, orange!90!black,
  mark=square*,
  mark indices={20},
  restrict expr to domain={\thisrow{g_eps}}{-1000:0},
  restrict expr to domain={\thisrow{q1}}{0.099:0.101},
] table[x expr={-\thisrow{g_eps}}, y=ratio_pct] {Data/figure2_gradient_value.csv};

\addplot[
  thick, teal!80!black,
  mark=triangle*,
  mark indices={20},
  restrict expr to domain={\thisrow{g_eps}}{-1000:0},
  restrict expr to domain={\thisrow{q1}}{0.249:0.251},
] table[x expr={-\thisrow{g_eps}}, y=ratio_pct] {Data/figure2_gradient_value.csv};

\addplot[
  thick, red!75!black,
  mark=diamond*,
  mark indices={20},
  restrict expr to domain={\thisrow{g_eps}}{-1000:0},
  restrict expr to domain={\thisrow{q1}}{0.499:0.501},
] table[x expr={-\thisrow{g_eps}}, y=ratio_pct] {Data/figure2_gradient_value.csv};

\addplot[
  thick, purple!80!black,
  mark=star,
  mark indices={15},
  restrict expr to domain={\thisrow{g_eps}}{-1000:0},
  restrict expr to domain={\thisrow{q1}}{0.749:0.751},
] table[x expr={-\thisrow{g_eps}}, y=ratio_pct] {Data/figure2_gradient_value.csv};

\addplot[
  thick, brown!85!black,
  mark=o,
  mark indices={15},
  restrict expr to domain={\thisrow{g_eps}}{-1000:0},
  restrict expr to domain={\thisrow{q1}}{0.899:0.901},
] table[x expr={-\thisrow{g_eps}}, y=ratio_pct] {Data/figure2_gradient_value.csv};

\addplot[
  thick, black,
  mark=x,
  mark indices={10},
  restrict expr to domain={\thisrow{g_eps}}{-1000:0},
  restrict expr to domain={\thisrow{q1}}{0.989:0.991},
] table[x expr={-\thisrow{g_eps}}, y=ratio_pct] {Data/figure2_gradient_value.csv};

\end{axis}
\end{tikzpicture}
\caption{Negative observed gradient.}
\end{subfigure}
\hfill
\begin{subfigure}{0.48\textwidth}
\centering
\begin{tikzpicture}
\begin{axis}[
    xmode=log,
    xmin=1e-6,
    xmax=2,
    ymin=25, ymax=100,
    xlabel={$g_{\epsilon}$},
    grid=both,
    table/col sep=comma,
    legend style={font=\footnotesize},
    legend cell align=left,
    legend pos=south west
]

% Manual legend entries
\addlegendimage{thick, blue!70!black, mark=*}
\addlegendentry{$q_1=0.01$}

\addlegendimage{thick, orange!90!black, mark=square*}
\addlegendentry{$q_1=0.10$}

\addlegendimage{thick, teal!80!black, mark=triangle*}
\addlegendentry{$q_1=0.25$}

\addlegendimage{thick, red!75!black, mark=diamond*}
\addlegendentry{$q_1=0.50$}

\addlegendimage{thick, purple!80!black, mark=star}
\addlegendentry{$q_1=0.75$}

\addlegendimage{thick, brown!85!black, mark=o}
\addlegendentry{$q_1=0.90$}

\addlegendimage{thick, black, mark=x}
\addlegendentry{$q_1=0.99$}

% Plots
\addplot[
  thick, blue!70!black,
  mark=*,
  mark indices={90},
  restrict expr to domain={\thisrow{g_eps}}{0:1000},
  restrict expr to domain={\thisrow{q1}}{0.009:0.011},
] table[x=g_eps, y=ratio_pct] {Data/figure2_gradient_value.csv};

\addplot[
  thick, orange!90!black,
  mark=square*,
  mark indices={90},
  restrict expr to domain={\thisrow{g_eps}}{0:1000},
  restrict expr to domain={\thisrow{q1}}{0.099:0.101},
] table[x=g_eps, y=ratio_pct] {Data/figure2_gradient_value.csv};

\addplot[
  thick, teal!80!black,
  mark=triangle*,
  mark indices={90},
  restrict expr to domain={\thisrow{g_eps}}{0:1000},
  restrict expr to domain={\thisrow{q1}}{0.249:0.251},
] table[x=g_eps, y=ratio_pct] {Data/figure2_gradient_value.csv};

\addplot[
  thick, red!75!black,
  mark=diamond*,
  mark indices={90},
  restrict expr to domain={\thisrow{g_eps}}{0:1000},
  restrict expr to domain={\thisrow{q1}}{0.499:0.501},
] table[x=g_eps, y=ratio_pct] {Data/figure2_gradient_value.csv};

\addplot[
  thick, purple!80!black,
  mark=star,
  mark indices={90},
  restrict expr to domain={\thisrow{g_eps}}{0:1000},
  restrict expr to domain={\thisrow{q1}}{0.749:0.751},
] table[x=g_eps, y=ratio_pct] {Data/figure2_gradient_value.csv};

\addplot[
  thick, brown!85!black,
  mark=o,
  mark indices={90},
  restrict expr to domain={\thisrow{g_eps}}{0:1000},
  restrict expr to domain={\thisrow{q1}}{0.899:0.901},
] table[x=g_eps, y=ratio_pct] {Data/figure2_gradient_value.csv};

\addplot[
  thick, black,
  mark=x,
  mark indices={95},
  restrict expr to domain={\thisrow{g_eps}}{0:1000},
  restrict expr to domain={\thisrow{q1}}{0.989:0.991},
] table[x=g_eps, y=ratio_pct] {Data/figure2_gradient_value.csv};

\end{axis}
\end{tikzpicture}
\caption{Positive observed gradient.}
\end{subfigure}

\caption{\textbf{Value of a gradient measurement for various $q_1$ ($\Cb = \Fb$).} The maximin ratio is computed by evaluating $\underline{\Lb}(\Cb(\Iset), \IN)$ with a discretization involving $M = 1000$ points.}
\label{fig:grad_value}
\end{figure}

\Cref{fig:grad_value} shows that when the absolute value of the gradient is small, the worst-case ratio is high, as the decision-maker is close to the optimal price. More generally, the results indicate that gradient information is informative even when its magnitude is not negligible. For instance, when $q_1 = 0.5$, the worst-case ratio without gradient information is around $60\%$, while knowing that the absolute value of the gradient is at most $0.1$ increases the ratio to above $80\%$. This illustrates that local information can significantly improve performance guarantees.

\noindent \textbf{Value of the gradient sign.}
In settings where estimating the magnitude of the gradient is noisy and unreliable, the seller may only be able to infer its sign. When demand is regular, the revenue function is unimodal, and the sign of the gradient at $p_1$ determines whether the optimal price $\rs$ satisfies $\rs \geq p_1$ or $\rs \leq p_1$. This corresponds to eliminating half of the feasible set of optimal prices.

We illustrate the value of this information in the same setting as above. As a baseline, we consider the maximin ratio obtained when the information set is $\{(p_1,q_1)\}$. We then evaluate the ratio obtained when the seller additionally observes the sign of the gradient at $p_1$. The results are reported in \Cref{fig:grad_sign}.

\begin{figure}[h!]
\centering
\begin{tikzpicture}
\begin{axis}[
    xmin=0.01, xmax=0.99,
    ymin=20, ymax=100,
    xlabel={$q_1$},
    ylabel={Worst-case ratio (\%)},
    grid=both,
    table/col sep=comma,
    xtick={0.10,0.20,0.30,0.40,0.50,0.60,0.70,0.80,0.90},
    legend style={font=\footnotesize},
    legend cell align=left,
    legend pos=south west,  
]

\addplot[thick, mark=*, black]
  table[x=q1, y=ratio_baseline_pct] {Data/figure3_gradient_sign.csv};
\addlegendentry{Baseline}

\addplot[thick, mark=square*, red!70!black]
  table[x=q1, y=ratio_positive_grad_pct] {Data/figure3_gradient_sign.csv};
\addlegendentry{Positive gradient}

\addplot[thick, mark=triangle*, blue!75!black]
  table[x=q1, y=ratio_negative_grad_pct] {Data/figure3_gradient_sign.csv};
\addlegendentry{Negative gradient}

\end{axis}
\end{tikzpicture}
    \caption{
    \textbf{Value of the gradient-sign for various $q_1$.} The maximin ratio with gradient-sign information is computed by modifying the evaluation of $\underline{\Lb}(\Cb(\{(p_1,q_1)\}), \IN)$ where we only consider constraints for $i$ such that $a_i \leq p_1$ (resp. $a_i \geq p_1$) when the gradient-sign is negative (resp. positive). ($p_1=10$)   
    }
    \label{fig:grad_sign}
\end{figure}

\Cref{fig:grad_sign} shows that the gradient sign can provide substantial information depending on the initial conversion rate. For example, when $q_1 = 0.5$, observing that the gradient is negative increases the worst-case ratio from $60\%$ to $83\%$. In contrast, for values such as $q_1 = 0.25$, the negative sign provides little additional information: the binding baseline worst-case instances already have feasible optimal prices weakly below $p_1$, so the restriction $p^*\le p_1$ does not eliminate them.
Our framework also enables a cost-benefit analysis of local experimentation. For instance, when $q_1 = 0.25$, the negative sign of the gradient does not improve performance, but \Cref{fig:grad_value}-a shows that a more precise estimate of its magnitude can yield significant gains. In that case, knowing that the gradient is equal to $-0.1$ increases the worst-case ratio from approximately $70\%$ to above $80\%$.

These results highlight that the value of local information depends on both the initial conversion rate and the precision of the measurement, and our framework provides a systematic way to quantify the benefit of more accurate, and potentially more costly, local experiments.

\subsection{The value of global experimentation}\label{sec:extra_point}
We now relax the constraint that the decision-maker should explore locally and consider the setting in which the seller has observed a single demand point $(p_1, q_1)$ and has the opportunity to experiment at any second price $p_2$ by collecting the demand level at the latter price, before committing to a final pricing mechanism. 
We model the problem as a game between Nature and the seller, in which the seller commits to a price $p_2$ to experiment at, and Nature counters the seller by choosing the associated feasible conversion rate $q_2$ to yield the worst maximin ratio in the final stage. More specifically, for a fixed $p_2$, Nature is solving the following problem,
\begin{align*}
\inf_{q_2 \in [0,1]}  \sup_{\Psi \in \MechSet} \inf_{F \in \Cb(\{(p_1, p_2), (q_1, q_2)\})}  \Ratio[\Psi, F].
\end{align*}

We illustrate numerically the value of one additional experiment in the setting where the seller knows the demand function is regular and its support is included in $[\lb, \ub] = [0, 100]$. The seller starts with one initial observed demand at the price $p_1 = 10$. For each initial  $q_1$ and candidate experiment price $p_2$ from a fixed grid of $[\lb, \ub]$ of size 200, we numerically evaluate a lower-bound on the above objective using the LP described in \Cref{sec:value_information} (with $M=2500$). We present in \Cref{1p_exp}-a, the improvement in the worst-case ratio allowed with the best experiment price $p_2$ chosen by the decision-maker and show in \Cref{1p_exp}-b what this best price is.
\begin{figure}[h!]
\centering
\begin{subfigure}{0.48\textwidth}
\centering
\begin{tikzpicture}
\begin{axis}[
    xmin=0.01, xmax=0.99,
    ymin=0, ymax=100,
    xlabel={$q_1$},
    ylabel={Worst-case ratio (\%)},
    grid=both,
    table/col sep=comma,
    xtick={0.10,0.20,0.30,0.40,0.50,0.60,0.70,0.80,0.90},
    legend style={font=\footnotesize},
    legend cell align=left,
    legend pos=south east,  
]

\addplot[thick, mark=*, blue!75!black]
  table[x=q1, y=baseline_ratio_pct] {Data/figure4_summary.csv};
\addlegendentry{Baseline $\Iset=(p_1,q_1)$}

\addplot[thick, mark=square*, red!75!black]
  table[x=q1, y=best_additional_exp_ratio_pct] {Data/figure4_summary.csv};
\addlegendentry{$(p_1,q_1)$ + experiment at $p_2$}

\end{axis}
\end{tikzpicture}
\caption{Performance with one additional experiment.}
\end{subfigure}
\hfill
\begin{subfigure}{0.48\textwidth}
\centering
\begin{tikzpicture}
\begin{axis}[
    xmin=0.01, xmax=0.99,
    xlabel={$q_1$},
    ylabel={Best $p_2$},
    grid=both,
    table/col sep=comma,
    xtick={0.10,0.20,0.30,0.40,0.50,0.60,0.70,0.80,0.9},
    legend style={font=\footnotesize},
    legend cell align=left,
    legend pos=south east,  
]

\addplot[thick, mark=*, red!75!black]
  table[x=q1, y=best_p2] {Data/figure4_summary.csv};
\addlegendentry{Best experiment price $p_2$}

\addplot[black, dashed, thick, forget plot]
  coordinates {(0.01,10) (0.99,10)};

\end{axis}
\end{tikzpicture}
\caption{Best candidate experiment price $p_2$.}
\end{subfigure}

\caption{\textbf{Value of one additional global experiment for various $(p_1,q_1)$.}
Panel (a) compares the baseline guarantee to the best robust guarantee achieved after one optimally chosen second price experiment.
Panel (b) reports the corresponding best $p_2$.}
\label{1p_exp}
\end{figure}

\Cref{1p_exp} shows that by allowing the decision-maker to judiciously select a second price at which to collect information, they can considerably improve the performance of the final price selected. Indeed, while a single price already guarantees a ratio of at least $30\%$ for conversion rates between $[0.01,0.99]$, the decision-maker is ensured to obtain more than $50\%$ of the revenue by experimenting correctly at a second price. More generally, across conversion rates at the initial price, the decision-maker can increase the robust performance by approximately $20\%$ through an additional well-selected price experiment.
\Cref{1p_exp}-b also illustrates how our procedure allows us to select a second price to run efficient experimentation. For instance, it provides general guidelines regarding when the price should be deflated or inflated to gather more information: in this case, when the conversion rate at $p_1$ is lower (respectively higher) than approximately $0.3$ the decision-maker should consider experimenting at a lower (resp. higher) price.

\subsection{Stopping criterion for dynamic pricing with deterministic feedback}\label{sec:ternary}

In \Cref{sec:extra_point}, we considered the case of a decision-maker who performs one additional experiment before committing to a pricing decision. We now illustrate how our framework can be used more broadly to guide dynamic pricing procedures. In particular, we show how it can be used to determine when to stop experimenting while guaranteeing a desired level of performance.

Given a target accuracy level $\epsilon$, a valid stopping criterion ensures that the pricing mechanism returned achieves at least a $(1-\epsilon)$ fraction of the optimal revenue. We illustrate how to construct such a stopping rule by leveraging our characterization of the value of information.
To this end, we consider a setting in which the seller dynamically adjusts prices and observes exact conversion rates at each step. We focus on the case where the value distribution is regular, so that the revenue function is unimodal. In this setting, classical approaches rely on procedures such as the ternary search algorithm, which sequentially narrows the interval of candidate optimal prices by comparing revenues at two interior points. We provide a formal description of ternary search in Algorithm~\ref{algo:ts} (see \Cref{sec:apx_algorithms}).

We emphasize that most dynamic pricing algorithms share a similar structure: they consist of an exploration procedure, a stopping criterion, and a final exploitation decision. In what follows, we fix the exploration procedure of a given dynamic pricing method $\theta$ and replace its stopping rule with one derived from our framework.
At each iteration $t$, the decision-maker has access to an information set $\Iset^t$ consisting of observed prices and conversion rates. Using the discretized program introduced in \Cref{sec:value_information}, we compute the quantity $\underline{\Lb}(\Cb(\Iset^t), \IN)$ defined in \eqref{eq:minimax_lp}, which provides a lower bound on the achievable maximin ratio. Therefore, we propose the following stopping rule: stop the exploration phase as soon as
\begin{equation*}
\underline{\Lb}(\Cb(\Iset^t), \IN) \geq 1 - \epsilon.
\end{equation*}
\Cref{thm:lp_discrete_bound} proves that this criterion guarantees that the pricing mechanism returned by solving \eqref{eq:minimax_lp} achieves a worst-case ratio of at least $1-\epsilon$. This provides a valid and data-dependent stopping rule that can be combined with any exploration procedure. Algorithm~\ref{algo:algo_with_stopping} in \Cref{sec:apx_algorithms} formalizes this stopping criterion.

We compare this approach with the standard stopping criterion used in ternary search, which prescribes a fixed number of iterations (see \Cref{lem:ternary_stop} in \Cref{sec:apx_algorithms}) given by
\begin{equation*}
N^{\mathrm{Ternary}}(\epsilon) = 2 \cdot \left \lceil \frac{\log(\frac{\ub-\lb}{\lb \cdot \epsilon})}{\log(3/2)} \right \rceil.
\end{equation*}
This rule is non-adaptive and depends only on the desired accuracy level, not on the observed data.

To evaluate the effectiveness of our approach, we consider the following experimental setup. We fix the exploration procedure to follow the ternary search algorithm, and vary only the stopping criterion. We sample instances from parametric families of regular demand curves, including linear and exponential specifications.

We set $\lb = 0.01$ and $\ub = 1$. For the linear model, we consider demand curves of the form
$\overline{F}(v) = \min\left( \max(a - b \cdot v, 0), 1 \right),$
where $b$ is uniformly sampled from $[1,5]$ and $a$ is uniformly sampled from $[1,b]$. For the exponential model, we consider
$\overline{F}(v) = \min\left( \exp(a - b \cdot v), 1 \right),$
where $b$ is uniformly sampled from $[1,5]$ and $a$ is uniformly sampled from $[-0.2,b]$.

For each instance, we measure the number of price queries required to guarantee a worst-case ratio of at least $99\%$ under three different stopping criteria:
\begin{enumerate}[label=(\roman*)]
    \item the standard ternary search stopping rule,
    \item our stopping rule assuming unimodality of the revenue function,
    \item our stopping rule assuming regularity of the demand distribution.
\end{enumerate}

We report in \Cref{fig:deterministic_queries} the number of price queries needed by each criterion to ensure that the exploitation price will achieve a $99\%$ worst-case ratio. We emphasize that while the stopping criterion changes, the dynamic pricing algorithm used for exploration is always the same and corresponds to the Ternary search approach. 
\begin{figure}[h!]
\centering
\begin{subfigure}{0.48\textwidth}
\centering
\begin{tikzpicture}
\begin{axis}[
    ybar,
    bar width=5pt,
    xmin=4, xmax=50,
    ymin=0, ymax=200,
    xlabel={Number of Price Queries},
    ylabel={Count},
    grid=both,
    table/col sep=comma,
    xtick distance=4,
    legend style={font=\scriptsize},
    legend cell align=left,
    legend pos=north west,  
]

\addplot[fill=red!70, draw=red!70, discard if not={criterion}{queries_regular_stopping}]
  table[x=query_count, y=count] {Data/figure5_linear_histogram.csv};
\addlegendentry{Regular Stopping}

\addplot[fill=blue!60, draw=blue!60, pattern=north east lines,  pattern color=blue!70!black,
 discard if not={criterion}{queries_unimodal_stopping}]
  table[x=query_count, y=count] {Data/figure5_linear_histogram.csv};
\addlegendentry{Unimodal Stopping}

\draw[black, dashed, ultra thick]
    (axis cs:44,0) -- (axis cs:44,210);

\addlegendimage{
  legend image code/.code={
    \draw[black, dashed, ultra thick] (0cm,0cm) -- (0.35cm,0cm);
  }
}
\addlegendentry{$N^{\mathrm{Ternary}}(\varepsilon)$}

\end{axis}
\end{tikzpicture}
\caption{Linear demand.}
\end{subfigure}
\hfill
\begin{subfigure}{0.48\textwidth}
\centering
\begin{tikzpicture}
\begin{axis}[
    ybar,
    bar width=5pt,
    xmin=4, xmax=50,
    ymin=0, ymax=200,
    xlabel={Number of Price Queries},
    ylabel={Count},
    grid=both,
    table/col sep=comma,
    xtick distance=4,
    legend style={font=\scriptsize},
    legend cell align=left,
    legend pos=north west,  
]

\addplot[fill=red!70, draw=red!70, discard if not={criterion}{queries_regular_stopping}]
  table[x=query_count, y=count] {Data/figure5_exponential_histogram.csv};
\addlegendentry{Regular Stopping}

\addplot[fill=blue!60, draw=blue!60,     pattern=north east lines,  pattern color=blue!70!black,
 discard if not={criterion}{queries_unimodal_stopping}]
  table[x=query_count, y=count] {Data/figure5_exponential_histogram.csv};
\addlegendentry{Unimodal Stopping}

\draw[black, dashed, ultra thick]
    (axis cs:44,0) -- (axis cs:44,210);

\addlegendimage{
  legend image code/.code={
    \draw[black, dashed, ultra thick] (0cm,0cm) -- (0.35cm,0cm);
  }
}
\addlegendentry{$N^{\mathrm{Ternary}}(\varepsilon)$}

\end{axis}
\end{tikzpicture}
\caption{Exponential demand.}
\end{subfigure}

\caption{\textbf{Distribution of the number of price queries needed to achieve 99\% worst-case ratio performance.}}
\label{fig:deterministic_queries}
\end{figure}

\Cref{fig:deterministic_queries} shows that using our characterization of the worst-case ratio allows to stop considerably earlier than by using the standard stopping criterion of Ternary search. When using our stopping criterion which only uses the unimodality of the revenue curve, we are already able to divide by two the number of price queries needed on average. This difference can be explained by the data-driven nature of our criterion. Indeed, the Ternary search stopping criterion is non-adaptive and considers the worst-case number of iterations needed across all unimodal curves, whereas our stopping criterion leverages the information gathered at each step to rule out some of the unimodal curves based on the conversion rates observed. 

Furthermore, our stopping criterion can be used to incorporate additional beliefs about the demand distribution. For instance, we can incorporate regularity by setting $\Cb = \Fb$. In that case, we see that our criterion can reduce on average the number of price points at which we experiment by a factor of five compared to the Ternary search criterion. For some of the instances considered our regular stopping criterion proposes to stop after only 5 price queries (as opposed to $44$ with usual stopping rules and ternary search). This significant gain in efficiency considerably changes our understanding of the value of dynamic price experiments.

\section{Handling Noisy Conversion Rates}
\label{sec:noisy_feedback}
In the previous sections, we assumed that conversion rates at observed prices are known exactly. In practice, however, conversion rates must be estimated from finite samples. In this section, we investigate how sampling noise affects the pricing policies computed by our framework.

Fix a value distribution $F$ supported on $[\underline v,\overline v]$ and consider a grid of $K$ prices $(p_i)_{i=1}^K$, equally spaced over $[\underline v,\overline v]$. Let $q_i = \overline{F}(p_i-)$ denote the true conversion rate at price $p_i$. These conversion rates are not directly observed. Instead, at each price $p_i$ the seller observes $T$ independent binary purchase outcomes, yielding an empirical conversion rate $\hat q_i$ defined by
\begin{equation*}
\hat q_i = \frac{1}{T}\sum_{j=1}^T \mathbbm{1}\{v_{ij} \ge p_i\},
\end{equation*}
where $(v_{ij})_{j=1}^T$ are i.i.d.\ draws from $F$. The resulting empirical information set is therefore $\mathcal{I}_{\bm{p},\bm{\hat q}}$. An important issue is that, due to sampling noise, the set $\Cb(\mathcal{I}_{\bm{p},\bm{\hat q}})$ may be empty. For instance, empirical conversion rates may violate monotonicity or fail to satisfy regularity constraints. To address this, we introduce a two-step procedure that combines projection with robust optimization.

\noindent \textbf{Robust policy with noisy observation.}
We propose the following policy.
First, we project the empirical conversion rates onto a structured feasibility set that encodes the desired shape constraints.

In the general case $\Cb=\Gb$, we compute the projected empirical conversion rate $\bm{q^\pi}$ by projecting $\bm{\hat q}$ onto the monotone cone $\{q_0\ge q_1\ge\cdots\ge q_{K+1}\}$, which corresponds to the standard isotonic regression problem \citep{barlow1972statistical, robertson1988order}.

In the regular case $\Cb=\Fb$, the regularity constraints are most conveniently expressed after applying the transformation $\Gamma^{-1}(q)=\frac{1}{q}-1$. Let $\hat z_i=\Gamma^{-1}(\max(\epsilon,\hat q_i))$ for some small\footnote{This $\epsilon$ is needed as $\Gamma^{-1}$ is not defined at $0$. In our implementation we use $\epsilon = 10^{-4}.$} $\epsilon > 0$ and define the projected sequence $\bm{z^\pi}=(z^\pi_i)_{i=0}^{K+1}$ as a solution of the weighted least-squares projection
\begin{equation*}
\bm{z^\pi}\in \arg\min_{\bm{z}\in\mathbb{R}^{K+2}} \sum_{i=0}^{K+1} w_i \cdot (z_i-\hat z_i)^2
\quad \text{s.t.}\quad
z_0 \le z_1 \le \cdots \le z_{K+1},
\quad
\frac{z_{i}-z_{i-1}}{p_{i}-p_{i-1}} \le \frac{z_{i+1}-z_{i}}{p_{i+1}-p_{i}},\ \ i \leq K-1,
\end{equation*}
for some nonnegative weights $(w_i)_{i=0}^{K+1}$ (all equal to $1$ in what follows). The first set of constraints enforces monotone demand, while the second set enforces the discrete convexity condition on $p\mapsto \Gamma^{-1}(q(p))$ that is equivalent to regular feasibility on a finite price grid (see \Cref{lemma:feasible-set}). We then map back to conversion rates by setting
\begin{equation*}
q^\pi_i = \Gamma(z^\pi_i) = \frac{1}{1+z^\pi_i}, \qquad i=0,\ldots,K+1,
\end{equation*}
and obtain a projected information set $\mathcal{I}_{\bm{p},\bm{q^\pi}}$. Importantly, \Cref{lemma:feasible-set} implies that $\Gb(\mathcal{I}_{\bm{p},\bm{q^\pi}})$ and $\Fb(\mathcal{I}_{\bm{p},\bm{q^\pi}})$ are non-empty when using the corresponding projection.

Second, given the projected information set $\mathcal{I}_{\bm{p},\bm{q^\pi}}$, we compute the \emph{discretized maximin policy} by solving the LP described in \Cref{sec:value_information}, which we denote by $\hat\Psi(\bm{q^\pi})$. 

\noindent \textbf{Evaluation and Benchmark.}
To assess how noise affects the \emph{true} robustness of the learned policy, we evaluate its worst-case ratio over the \emph{true} feasible set induced by the true quantiles $\bm{q}$ as,
\begin{equation}
\widehat{\rho}(T,K)  \;=\; \inf_{F\in \Cb(\mathcal{I}_{\bm{p},\bm{q}})} \Ratio[\hat\Psi(\bm{q^\pi}),F].
\label{eq:empirical_perf}
\end{equation}
\Cref{eq:empirical_perf} can be efficiently computed by leveraging the reduction developed in \Cref{thm:nature_reduction}.
We note that this quantity is random because $\bm{q^\pi}$ depends on the sampling noise in $\bm{\hat q}$.

As a benchmark, we also compute the oracle maximin ratio associated with the true quantiles,
\begin{equation*}
\rho^\star(K) \;=\; \sup_{\Psi\in\mathcal{P}} \ \inf_{F\in \Cb(\mathcal{I}_{\bm{p},\bm{q}})} \Ratio,
\end{equation*}
which corresponds to the performance achievable if the decision-maker had access to the exact conversion rates at the $K$ prices. In \Cref{fig:true_ratio_vs_M}, dotted lines report $\rho^\star(K)$, while solid lines report the empirical average of $\widehat{\rho}(T,K)$ across repeated noisy datasets, together with 90\% empirical confidence intervals. The detailed parameters such as prices and conversion rates are presented in \Cref{sec:apx_figure6}.

\begin{figure}[t]
\centering

% =========================
% Left: Regular
% =========================
\begin{subfigure}[t]{0.48\textwidth}
\centering

\def\DATAFILE{Data/figure6_regular_tikz.csv}

\begin{tikzpicture}
\begin{axis}[
    width=\textwidth,
    height=0.8\textwidth,
    xlabel={$T$},
    ylabel={Worst-case ratio},
    grid=both,
    table/col sep=comma,
    ymin=0.2, ymax=1,
    ytick={0,0.2,0.4,0.6,0.8,1.0},
    legend style={font=\footnotesize},
    legend columns=2,
    legend cell align=left,
    legend pos=south east,  
]

\addlegendimage{blue, thick, mark=*}
\addlegendentry{$K=3$}

\addlegendimage{orange, thick, mark=square*}
\addlegendentry{$K=5$}

\addlegendimage{green!60!black, thick, mark=triangle*}
\addlegendentry{$K=9$}

\addlegendimage{dashed, line width = 0.2mm, draw=black, no marks}
\addlegendentry{Known quantile}

\addplot[name path=Kthree_lo, draw=none, forget plot]
    table[x=M, y={P05_K=3_true}] {\DATAFILE};
\addplot[name path=Kthree_hi, draw=none, forget plot]
    table[x=M, y={P95_K=3_true}] {\DATAFILE};
\addplot[blue, fill opacity=0.18, forget plot]
    fill between[of=Kthree_lo and Kthree_hi];

\addplot[blue, thick, mark=*]
    table[x=M, y={avg_K=3_true}] {\DATAFILE};

\addplot[blue, dashed, line width = 0.2mm, forget plot]
    table[x=M, y={known_quantile_ratio_3}] {\DATAFILE};

\addplot[name path=Kfive_lo, draw=none, forget plot]
    table[x=M, y={P05_K=5_true}] {\DATAFILE};
\addplot[name path=Kfive_hi, draw=none, forget plot]
    table[x=M, y={P95_K=5_true}] {\DATAFILE};
\addplot[orange, fill opacity=0.18, forget plot]
    fill between[of=Kfive_lo and Kfive_hi];

\addplot[orange, thick, mark=square*]
    table[x=M, y={avg_K=5_true}] {\DATAFILE};

\addplot[orange, dashed, line width = 0.2mm, forget plot]
    table[x=M, y={known_quantile_ratio_5}] {\DATAFILE};

\addplot[name path=Kten_lo, draw=none, forget plot]
    table[x=M, y={P05_K=9_true}] {\DATAFILE};
\addplot[name path=Kten_hi, draw=none, forget plot]
    table[x=M, y={P95_K=9_true}] {\DATAFILE};
\addplot[green!60!black, fill opacity=0.18, forget plot]
    fill between[of=Kten_lo and Kten_hi];

\addplot[green!60!black, thick, mark=triangle*]
    table[x=M, y={avg_K=9_true}] {\DATAFILE};

\addplot[green!60!black, dashed, line width = 0.2mm, forget plot]
    table[x=M, y={known_quantile_ratio_9}] {\DATAFILE};

\end{axis}
\end{tikzpicture}

\caption{Regular distributions}
\label{fig:true_ratio_vs_M_regular}
\end{subfigure}
\hfill
% =========================
% Right: General
% =========================
\begin{subfigure}[t]{0.48\textwidth}
\centering

\def\DATAFILE{Data/figure6_general_tikz.csv}

\begin{tikzpicture}
\begin{axis}[
    width=\textwidth,
    height=0.8\textwidth,
    xlabel={$T$},
    ylabel={Worst-case ratio},
    grid=both,
    table/col sep=comma,
    ymin=0.2, ymax=1,
    ytick={0,0.2,0.4,0.6,0.8,1.0},
    legend style={font=\footnotesize},
    legend columns=2,
    legend cell align=left,
    legend pos=south east,
]

\addlegendimage{blue, thick, mark=*}
\addlegendentry{$K=3$}

\addlegendimage{orange, thick, mark=square*}
\addlegendentry{$K=5$}

\addlegendimage{green!60!black, thick, mark=triangle*}
\addlegendentry{$K=9$}

\addlegendimage{dashed, line width = 0.2mm, draw=black, no marks}
\addlegendentry{Known quantile}

\addplot[name path=Kthree_lo, draw=none, forget plot]
    table[x=M, y={P05_K=3_true}] {\DATAFILE};
\addplot[name path=Kthree_hi, draw=none, forget plot]
    table[x=M, y={P95_K=3_true}] {\DATAFILE};
\addplot[blue, fill opacity=0.18, forget plot]
    fill between[of=Kthree_lo and Kthree_hi];

\addplot[blue, thick, mark=*]
    table[x=M, y={avg_K=3_true}] {\DATAFILE};

\addplot[blue, dashed, line width = 0.2mm, forget plot]
    table[x=M, y={known_quantile_ratio_3}] {\DATAFILE};

\addplot[name path=Kfive_lo, draw=none, forget plot]
    table[x=M, y={P05_K=5_true}] {\DATAFILE};
\addplot[name path=Kfive_hi, draw=none, forget plot]
    table[x=M, y={P95_K=5_true}] {\DATAFILE};
\addplot[orange, fill opacity=0.18, forget plot]
    fill between[of=Kfive_lo and Kfive_hi];

\addplot[orange, thick, mark=square*]
    table[x=M, y={avg_K=5_true}] {\DATAFILE};

\addplot[orange, dashed, line width = 0.2mm, forget plot]
    table[x=M, y={known_quantile_ratio_5}] {\DATAFILE};

\addplot[name path=Kten_lo, draw=none, forget plot]
    table[x=M, y={P05_K=9_true}] {\DATAFILE};
\addplot[name path=Kten_hi, draw=none, forget plot]
    table[x=M, y={P95_K=9_true}] {\DATAFILE};
\addplot[green!60!black, fill opacity=0.18, forget plot]
    fill between[of=Kten_lo and Kten_hi];

\addplot[green!60!black, thick, mark=triangle*]
    table[x=M, y={avg_K=9_true}] {\DATAFILE};

\addplot[green!60!black, dashed, line width = 0.2mm, forget plot]
    table[x=M, y={known_quantile_ratio_9}] {\DATAFILE};

\end{axis}
\end{tikzpicture}

\caption{General distributions}
\label{fig:true_ratio_vs_M_general}
\end{subfigure}

\caption{Average worst-case ratio with shaded 90\% empirical confidence intervals; dotted lines show the worst-case ratio with exact quantiles.}
\label{fig:true_ratio_vs_M}
\end{figure}

\noindent \textbf{Results.} Recall that Figure~\ref{fig:true_ratio_vs_M} evaluates an end-to-end offline pipeline under noisy feedback: we first project empirical conversion rates onto a shape-restricted feasible set and then compute the maximin policy using these \emph{projected} quantiles. Accordingly, the results reflect the joint effect of sampling noise and the stability of the projection--plus--LP mapping from $\bm{\hat q}$ to a robust pricing policy. With this in mind, we highlight the following observations.

While our theoretical analysis applies to the setting in which conversion rates are known exactly, the figure shows that, when combined with the projection step, the resulting robust policies retain strong performance guarantees in the noisy setting. In particular, the expected performance achieved with noisy observations approaches the benchmark suggested by the exact-quantile analysis as the number of samples per price increases.
This convergence occurs relatively quickly. For instance, to be within $10\%$ of the optimal performance with $K=9$, one requires approximately $100$ samples per price under regularity and about $50$ under general distributions, whereas for $K=3$ this level of performance is already achieved with roughly $25$ samples in both cases. These observations suggest that two factors contribute to faster convergence: smaller values of $K$, which reduce the dimensionality of the projection problem, and weaker structural assumptions, which make the projection step less sensitive to estimation error.

All in all, these results indicate that within the projection--plus--maximin framework studied here, the robust pricing policies computed from noisy conversion data using our approach recover near-optimal performance once $T$  is moderately large.

\section{Conclusion}
We present in this work a framework to precisely quantify the robust value of data obtained through price experiments in the form of past offered prices and observed conversion rates. Our main methodological contribution is a reduction of an optimization problem whose value quantifies the worst-case revenue gap incurred by a decision-maker who only has access to certain conversion rates compared to an oracle knowing the full value distribution. Building on this reduction, we develop a tractable procedure to compute near-optimal robust pricing policies and quantify the intrinsic value of information. We leverage these results to derive insights on the value of diverse types of historical data and to prescribe efficient ways of running price experiments, and show that our framework enables rigorous cost-benefit analyses of different types of information.

We also illustrate how our framework can be used to improve price experiments. We propose a procedure to select the next price at which one should experiment before committing to a final price. Furthermore, we show that our framework allows to significantly reduce the length of price experiments, even when committing to standard dynamic algorithms such as Ternary search, by providing a data-driven stopping criterion adaptive to the information collected. Finally, our numerical experiments illustrate that our robust policies can be modified to achieve near-optimal performance even when conversion rates must be estimated from noisy observations. We believe that this work opens up promising directions for more efficient price experimentation, including the design of dynamic pricing algorithms that better trade off exploration and exploitation.

\subsection*{Acknowledgment}
The authors would like to express their gratitude to Amine Allouah for his invaluable comments during the early stages of this work. The authors are also grateful to Victor Araman, Lin Fan and Michael Hamilton for their helpful discussions and valuable insights.

{
\setstretch{1}
\bibliographystyle{agsm}
\bibliography{ref}
}

\newpage

\appendix

\renewcommand{\thefigure}{\thesection-\arabic{figure}}
\renewcommand{\theequation}{\thesection-\arabic{equation}}
\renewcommand{\theproposition}{\thesection-\arabic{proposition}}
\renewcommand{\thelemma}{\thesection-\arabic{lemma}}
\renewcommand{\thetheorem}{\thesection-\arabic{theorem}}
\renewcommand{\thedefinition}{\thesection-\arabic{definition}}
\pagenumbering{arabic}
\renewcommand{\thepage}{App-\arabic{page}}

\setcounter{equation}{0}
\setcounter{proposition}{0}
\setcounter{definition}{0}
\setcounter{lemma}{0}
\setcounter{theorem}{0}
\setcounter{figure}{0}

\section{Proof of Results in \Cref{sec:Noptpricing}}
\label{sec:apx_Proof_Thm1}

\begin{proof}[\textbf{Proof of \Cref{lemma:singlecross}}] 
When $\Cb = \Fb$, the result follows directly  from  \cite[Lemma 2]{{ABBSamples}}, which generalizes \cite[Chapter 4, Theorem 2.18]{barlow1975statistical}, with the following additional identity:
\begin{align*}
q \cdot \Gad{\Gainv{\frac{q'}{q}}\frac{v-s}{s'-s}} &= \frac{1}{\frac{1}{q}\left(1+ \left(\frac{q}{q'}-1\right)\frac{v-s}{s'-s}\right)}\\
&= \frac{1}{\frac{1}{q} + \left(\frac{1}{q'}-\frac{1}{q}\right)\frac{v-s}{s'-s}} 
= \Gad{\Gainv{q} + \frac{\Gainv{q'}-\Gainv{q}}{s'-s}(v-s)} 
\end{align*}
When $\Cb = \Gb$, the result follows directly from the non-increasing property of complementary cumulative distribution functions.
\end{proof}

\subsection{Proof of \Cref{prop:worst_case_revenue}}

In what follows, we prove \Cref{prop:worst_case_revenue}. 
Our first lemma shows that our candidate worst-case distribution provides a lower bound on the revenue associated with any value distribution $F \in \Cb(\Iset)$.
\begin{lemma}
\label{lem:lower_revenue}
Consider an information set $\Iset$, then for any pricing mechanism $\Psi \in \MechSet$ and any value distribution $F \in \Cb(\Iset)$, we have that,
$\Expect_{\Psi}[  \Rev\left(p\vert F\right)] \geq \Expect_{\Psi}[ \Rev\left(p\vert L_{\Cb}( \cdot \vert  \Iset)  \right)]$,
where the ccdf of $L_{\Cb}$ is defined in \eqref{eq:L}.
\end{lemma}

Next, we show that our candidate worst-case distribution belongs to $\Cb(\Iset)$.
\begin{lemma}
\label{lem:L_belongs}
If $\Cb(\Iset)$ is non-empty, then $L_{\Cb}(\cdot \vert \Iset) \in \Cb(\Iset)$.
\end{lemma}

We note that from \Cref{lem:lower_revenue} and \Cref{lem:L_belongs} we can conclude that $L_{\Cb}(\cdot \vert \Iset)$ is the distribution which minimizes the absolute revenue among all distributions in $\Cb(\Iset)$. To establish \Cref{prop:worst_case_revenue} we also need to characterize the optimal price of this family of distributions to satisfy the constraint on the optimal revenue and its associated conversion rate.

The next structural result shows that the optimal revenue under our candidate distribution  $L_{\Cb}(\cdot \vert \Iset)$ is equal to the largest revenue obtained with points in $\Iset$.
\begin{lemma}
\label{lem:max_L}
If $\Cb(\Iset)$ is non-empty, then $\sup_{p \in [\lb,\ub]} \Rev\left(p \vert L_{\Cb}( \cdot \vert \Iset) \right) = \max_{i \in \{0,\ldots,N+1\}} p_i \cdot q_i.$
\end{lemma}

We next combine these lemmas to formally prove \Cref{prop:worst_case_revenue}.

\begin{proof}[\textbf{Proof of \Cref{prop:worst_case_revenue}}]
Consider $(\rs,\qs) \in {\cal B}(\Iset)$ and consider the extended information set $\IsetPlus = \Iset \cup \{(\rs,\qs)\}$. We note that $\Cb(\IsetPlus)$ is the subset of $\Cb(\Iset)$ where distributions must also satisfy $\bF(\rs-) = \qs$. We also let $\mathcal{A} = \{ F \in \Cb(\Iset) \text{ s.t. } \rs \in \mathcal{O}(F) \text{ and } \bF(\rs-) = \qs \}.$
We first show that,
\begin{equation}
\label{eq:prop1_geq}
\inf_{F \in \mathcal{A}} \Expect_{\Psi}[  \Rev\left(p\vert F\right)] \geq \Expect_{\Psi}[  \Rev\left(p\vert L_{\Cb}( \cdot \vert  \IsetPlus )  \right)]. 
\end{equation}
Indeed by definition, we have that $\mathcal{A} \subset \Cb(\IsetPlus)$. Therefore, \eqref{eq:prop1_geq} follows from \Cref{lem:lower_revenue} applied to the information set $\IsetPlus$.
Next, we show that 
\begin{equation}
\label{eq:prop1_leq}
\inf_{F \in \mathcal{A}} \Expect_{\Psi}[  \Rev\left(p\vert F\right)] \leq \Expect_{\Psi}[  \Rev\left(p\vert L_{\Cb}( \cdot \vert  \IsetPlus )  \right)]. 
\end{equation}
To do so, we will show that $L_{\Cb}( \cdot \vert  \IsetPlus ) \in \mathcal{A}$ by proving that $L_{\Cb}( \cdot \vert  \IsetPlus ) \in \Cb(\Iset)$ and that $(\rs,\qs)$ are optimal prices and associated conversion rates for $L_{\Cb}( \cdot \vert  \IsetPlus )$.

We first show that $\Cb(\IsetPlus)$ is non-empty. Note that since $(\rs,\qs) \in {\cal B}(\Iset)$ there exists $F \in \Cb(\Iset)$ such that $\bF(\rs-) = \qs$. Such a distribution belongs to $\mathcal{A}$ and therefore belongs to $\Cb(\IsetPlus)$. Therefore, this set is not empty.
It then follows from \Cref{lem:L_belongs} that $L_{\Cb}( \cdot \vert  \IsetPlus ) \in \Cb(\IsetPlus) \subset \Cb(\Iset)$. 

In addition, \Cref{lem:max_L} applied to the information set  $\IsetPlus$ implies that, 
\begin{equation*}
\sup_{p \in [\lb,\ub]} \Rev\left(p \vert L_{\Cb}( \cdot \vert \IsetPlus) \right) = \max\left( \max_{i \in \{0,\ldots,N+1\}} p_i \cdot q_i , \rs \cdot \qs \right).
\end{equation*}
Furthermore, we argue that $\rs \cdot \qs \geq \max_{i \in \{0,\ldots,N+1\}} p_i \cdot q_i$. This inequality holds because $\mathcal{A}$ is non-empty and for any distribution $F \in  \mathcal{A}$, we have for every $i \in \{0,\ldots,N+1\}$ that
$\rs \cdot \qs = \rs \cdot \bF(\rs-)  \stackrel{(a)}{\geq} p_i \cdot  \bF(p_i-) = p_i \cdot q_i,$
where $(a)$ holds because $F \in \mathcal{A}$ and hence $\rs$ is an optimal price for this distribution. This implies that $\sup_{p \in [\lb,\ub]} \Rev\left(p \vert L_{\Cb}( \cdot \vert \IsetPlus) \right) = \rs \cdot \qs$ and given that $\overline{L}_{\Cb}( \rs- \vert \IsetPlus) = \qs$, it concludes that $\rs$ is an optimal price for $L_{\Cb}( \cdot \vert \IsetPlus)$ with associated conversion rate $\qs$.

Therefore, $L_{\Cb}( \cdot \vert  \IsetPlus ) \in \mathcal{A}$ and  \eqref{eq:prop1_leq} holds. Finally, \eqref{eq:prop1_geq} and \eqref{eq:prop1_leq} imply \Cref{prop:worst_case_revenue}.
\end{proof}

We finally prove our auxiliary lemmas.

\begin{proof}[\textbf{Proof of \Cref{lem:lower_revenue}}]
Fix $F \in \Cb(\Iset)$ and consider $v \in [\lb,\ub]$. 
We first show that, 
$\bF(v) \geq \overline{L}_{\Cb}(v \vert  \Iset).$
Consider  $i \in \{0,\ldots,N\}$ such that, $v \in [p_i,p_{i+1})$. 
We remark that $\overline{F}(p_i-) = q_i$ and $\overline{F}_0(p_{i+1}-) = q_{i+1}$ because $F \in \Cb(\Iset)$. Therefore \Cref{lemma:singlecross} implies that,
    \begin{equation*}
        \overline{F}(v) \geq \overline{G}_{\Cb}(v \vert  (p_{i},q_{i}),(p_{i+1},q_{i+1})) = \overline{L}_{\Cb}(v \vert  \Iset), 
   \end{equation*}
 where the last equality holds by definition of $\overline{L}_{\Cb}(v \vert  \Iset)$. Hence, $\bF(v) \geq \overline{L}_{\Cb}(v \vert  \Iset)$.
Furthermore, for any pricing mechanism $\Psi \in \MechSet$, we have that,
\begin{equation*}
\Expect_{\Psi}[  \Rev\left(p\vert F\right)] = \int_0^{\infty} p \cdot \bF(p-) \, d\Psi(p) \stackrel{(a)}{\geq} \int_0^{\infty} p \cdot \overline{L}_{\Cb}(p- \vert \Iset)  \, d\Psi(p) =  \Expect_{\Psi}[  \Rev\left(p \vert L_{\Cb}( \cdot \vert \Iset) \right)],
\end{equation*}
where $(a)$ holds because $\bF(v) \geq \overline{L}_{\Cb}(v \vert  \Iset)$ for every $v \in [\lb,\ub]$ and therefore by taking a left limit $\bF(v-) \geq \overline{L}_{\Cb}(v- \vert  \Iset)$ for every $v \in [\lb,\ub]$. This concludes the proof.
\end{proof}

\begin{proof}[\textbf{Proof of \Cref{lem:L_belongs}}]
First, one can verify that for every $i \in \{0,\ldots, N+1\}$, $\overline{L}_{\Cb}(p_i- \vert \Iset) = q_i.$ Therefore, we need to show that, $L_{\Cb}(\cdot \vert \Iset) \in \Cb.$

When $\Cb = \Gb$, we note that $\overline{L}_{\Gb}(\cdot \vert \Iset)$ is non-decreasing because $(q_i)_{i \leq N+1}$ is non-decreasing. Indeed, $(q_i)_{i \leq N+1}$ is non-decreasing because the set $\Gb(\Iset)$ is non-empty. Hence, $L_{\Gb}(\cdot \vert \Iset) \in \Gb(\Iset)$.

When $\Cb = \Fb$, we show in \Cref{lemma:feasible-set} (Case $\Cb = \Fb:$ 
$\Longrightarrow$)  that if $\Cb(\Iset)$ is non-empty then $ \left(\frac{\Gainv{q_{i+1}}-\Gainv{q_{i}}}{p_{i+1}-p_{i}}\right)_{0 \leq i \leq N-1}$ is non-decreasing. We next show that this implies that 
$\overline{L}_{\Fb}(\cdot \vert \Iset)$ is regular. \Cref{lem:const} implies that the associated virtual value function satisfies
\begin{equation*}
\phi_{L_{\Fb}(\cdot \vert  \Iset)}(v) = \pii - (\pip-\pii)\frac{1+\gqi}{\gqip-\gqi}, \quad \mbox{if } v \in [p_{i},p_{i+1}), \mbox{ for } i = 0 \cdots N
\end{equation*}
Hence, the virtual value function is piece-wise constant. Now we need to show that $\phi_L(v)$ is non-decreasing. Fix $1 \leq i \leq N-1$. we evaluate the difference between the two consecutive constant values that the virtual value function is and note that it is equal to,
\begin{equation*}
 \pars{1+\gqi}\acols{\frac{\pii-\pim}{\gqi-\gqim}-\frac{\pip-\pii}{\gqip-\gqi}} \stackrel{(a)}{\geq} 0, 
\end{equation*}
where inequality $(a)$ holds because the sequence $ \left(\frac{\Gainv{q_{i+1}}-\Gainv{q_{i}}}{p_{i+1}-p_{i}}\right)_{0 \leq i \leq N-1}$ is non-decreasing.

Hence, the virtual value function of $\overline{L}_{\Fb}(\cdot \vert  \Iset)$ is non-decreasing and  $L_{\Fb}(\cdot \vert  \Iset) \in \Fb(\Iset)$.
\end{proof}

\begin{proof}[\textbf{Proof of \Cref{lem:max_L}}]
We first note that \Cref{lem:L_belongs} implies that $L_{\Cb}( \cdot \vert \Iset)$ belongs to $\Cb(\Iset)$. Therefore, for every $i \in \{1,\ldots,N+1\}$, we have $\Rev\left(p_i \vert L_{\Cb}( \cdot \vert \Iset) \right) =  p_i \cdot \overline{L}_{\Cb}(p_i- \vert \Iset)  = p_i \cdot q_i$. Hence,  $\sup_{p \in [\lb,\ub]} \Rev\left(p \vert L_{\Cb}( \cdot \vert \Iset) \right) \geq \max_{i \in \{0,\ldots,N+1\}} p_i \cdot q_i.$

To prove the reverse inequality, it suffices to show that the revenue function $p \mapsto \Rev\left(p \vert L_{\Cb}( \cdot \vert \Iset) \right)$ is monotonic on every interval $[p_i,p_{i+1})$ for $i \in \{0,\ldots,N\}$. By the definition in \eqref{eq:L}, we need to analyze the monotonicity of the function $\kappa: p \mapsto p \cdot \overline{G}_{\Cb,i}(p)$.

When $\Cb = \Gb$, we have that $\overline{G}_{\Cb,i}(p)$ is constant on $[p_i,p_{i+1})$. Hence, $\kappa$ is a linear function and is  monotonic on $[p_i,p_{i+1})$.

When $\Cb = \Fb$, we define for every $q \in  [\overline{G}_{\Cb,i}(p_i),\overline{G}_{\Cb,i}(p_{i+1}))$ the function $\tilde{\kappa}(q) = \overline{G}_{\Cb,i}^{-1}(q) \cdot q$. We note that for every $p \in [p_i,p_{i+1})$ we have $ \kappa(p) = \tilde{\kappa}(\overline{G}_{\Cb,i}(p))$ and,  by composition, if $\tilde{\kappa}$ is monotonic then $\kappa$ is monotonic. 
To show that $\tilde{\kappa}$ is monotonic we note that it corresponds to the revenue function in the quantile space. Its derivative satisfies,  $\frac{d}{dq}\tilde{\kappa}(q) = \phi_{G_{\Cb,i}}(\overline{G}_{\Cb,i}^{-1}(q))$, where $\phi_F$ is the virtual value function of a distribution $F$.  To conclude the proof, we show in \Cref{lem:const} that the virtual value function of $G_{\Cb,i}$ is constant. Therefore, $\tilde{\kappa}$ is linear and monotonic on $[q_i,q_{i+1})$. 
\end{proof}

\subsection{Proof of \Cref{thm:nature_reduction}}

We prove a stronger version of the result stated in terms of a
$\lambda$-regret criterion. The reduction for the worst-case ratio in
\Cref{thm:nature_reduction} will then follow from a simple epigraph
argument.

For any $\lambda \in [0,1]$, define the $\lambda$-regret of a mechanism
$\Psi \in \MechSet$ under a distribution $F$ as
\begin{equation*}
\LRegret[\lambda][\Psi,F]
=
\lambda \cdot \opt(F) - \Rev(\Psi,F).
\end{equation*}
The next result establishes a statement similar to \Cref{thm:nature_reduction} but applies to the $\lambda$-regret for every $\lambda$ above the policy-specific worst-case ratio.

\begin{theorem}[Reduction for $\lambda$-regret]
\label{thm:lambda_reduction}
Assume $\Cb$ is equal to $\Gb$ or $\Fb$. Fix $N \ge 1$ and let $\Iset$
be an information set including $N$ historical prices and conversion
rates. Assume $\Cb(\Iset)$ is non-empty.

Then for any $\Psi \in \MechSet$ and any
$\lambda
\in
\left[
\inf_{F \in \Cb(\Iset)} \Ratio[\Psi,F],
1
\right],$
we have
\begin{equation*}
\sup_{F \in \Cb(\Iset)} \LRegret[\lambda][\Psi,F]
=
\sup_{\rs \in \mathcal{S}(\Iset)}
\LRegret[\lambda][\Psi,F_{\Cb}(\cdot \vert \rs,\Iset)].
\end{equation*}
\end{theorem}
\Cref{thm:lambda_reduction} allows us to establish our reduction of Nature's problem for both the absolute regret criterion, by taking $\lambda =1$, and for the ratio. We next formally show how \Cref{thm:lambda_reduction} implies \Cref{thm:nature_reduction}. The key observation is that the worst-case ratio admits an epigraph representation in terms of the $\lambda$-regret family.

\begin{lemma}[Epigraph formulation]
\label{lem:epi_ratio}
Fix a mechanism $\Psi \in \MechSet$ and an information set $\Iset$
such that $\Cb(\Iset)$ is non-empty. Then
\begin{equation*}
\inf_{F \in \Cb(\Iset)} \Ratio[\Psi,F]
=
\sup\left\{
\lambda \in [0,1] :
\sup_{F \in \Cb(\Iset)} \LRegret[\lambda][\Psi,F] \le 0
\right\}.
\end{equation*}
\end{lemma}

\begin{proof}[\textbf{Proof of \Cref{lem:epi_ratio}}]
For any $F \in \Cb(\Iset)$ we have
\begin{equation*}
\LRegret[\lambda][\Psi,F]
=
\lambda \cdot \opt(F) - \Rev(\Psi,F).
\end{equation*}
Thus
\begin{equation*}
\LRegret[\lambda][\Psi,F] \le 0
\quad \Longleftrightarrow \quad
\frac{\Rev(\Psi,F)}{\opt(F)} \ge \lambda.
\end{equation*}
Taking the supremum over $F$ yields
\begin{equation*}
\sup_{F \in \Cb(\Iset)} \LRegret[\lambda][\Psi,F] \le 0
\quad \Longleftrightarrow \quad
\inf_{F \in \Cb(\Iset)} \Ratio[\Psi,F] \ge \lambda,
\end{equation*}
which proves the claim.
\end{proof}

\begin{proof}[\textbf{Proof of \Cref{thm:nature_reduction}}]
Fix $\Psi \in \MechSet$ and let
$\eta = \inf_{F \in \Cb(\Iset)} \Ratio[\Psi,F].$

Since $F_{\Cb}(\cdot \vert \rs,\Iset) \in \Cb(\Iset)$ for every $\rs \in \mathcal{S}(\Iset)$ by \Cref{lem:largest_element}, we have
\begin{equation*}
\eta
\le
\inf_{\rs \in \mathcal{S}(\Iset)}
\Ratio[\Psi,F_{\Cb}(\cdot \vert \rs,\Iset)].
\end{equation*}
It remains to prove the reverse inequality.

If $\eta=1$, then the reverse inequality follows because $\Ratio[\Psi,F]\le 1$ for every distribution $F$, and hence
\begin{equation*}
\inf_{\rs \in \mathcal{S}(\Iset)}
\Ratio[\Psi,F_{\Cb}(\cdot \vert \rs,\Iset)]
\le 1
=
\eta.
\end{equation*}
Assume now that $\eta<1$. Fix $\lambda \in (\eta,1]$. \Cref{lem:epi_ratio} implies that,
$\sup_{F \in \Cb(\Iset)}
\LRegret[\lambda][\Psi,F]
>
0.$
Because $\lambda \in [\eta,1]$, we can apply \Cref{thm:lambda_reduction} to obtain
\begin{equation*}
\sup_{\rs \in \mathcal{S}(\Iset)}
\LRegret[\lambda][\Psi,F_{\Cb}(\cdot \vert \rs,\Iset)]
>
0.
\end{equation*}
Therefore, there exists $\rs \in \mathcal{S}(\Iset)$ such that
$\LRegret[\lambda][\Psi,F_{\Cb}(\cdot \vert \rs,\Iset)]
>
0.$ This implies
\begin{equation*}
\Ratio[\Psi,F_{\Cb}(\cdot \vert \rs,\Iset)]
<
\lambda.
\end{equation*}
Thus,
$\inf_{\rs \in \mathcal{S}(\Iset)}
\Ratio[\Psi,F_{\Cb}(\cdot \vert \rs,\Iset)]
\le
\lambda.$
Since this holds for every $\lambda \in (\eta,1]$, we conclude that
\begin{equation*}
\inf_{\rs \in \mathcal{S}(\Iset)}
\Ratio[\Psi,F_{\Cb}(\cdot \vert \rs,\Iset)]
\le
\eta.
\end{equation*}
Combining the two inequalities gives
\begin{equation*}
\inf_{F \in \Cb(\Iset)} \Ratio[\Psi,F]
=
\inf_{\rs \in \mathcal{S}(\Iset)}
\Ratio[\Psi,F_{\Cb}(\cdot \vert \rs,\Iset)].
\end{equation*}
This proves \Cref{thm:nature_reduction}.
\end{proof}

In what follows, we prove \Cref{thm:lambda_reduction}.

For any $\rs$, we define
${\cal B}_{\rs}(\Iset)
=
\{q^* \in [0,1] \text{ s.t. } (\rs,q^*) \in {\cal B}(\Iset)\}$
as the set of conversion rates $q^*$ such that there exists a distribution which has $(\rs,q^*)$ as an optimal couple of price and conversion rate.
Our next result establishes a monotonicity property that is needed to optimize over the possible conversion rates ${\cal B}_{\rs}(\Iset)$.

\begin{lemma}
\label{lem:monotonic_rho}
Fix $\rs$ such that ${\cal B}_{\rs}(\Iset)$ is non-empty. Then, for every $p \in [\lb,\ub]$, the mapping
\begin{equation*}
\qs
\mapsto
\frac{1}{\qs}
\Rev\left(p\vert L_{\Cb}(\cdot \vert \Iset \cup (\rs,\qs))\right)
\end{equation*}
is non-increasing on ${\cal B}_{\rs}(\Iset)$.
\end{lemma}
We also show that the largest feasible conversion rate is the one of the upper envelope distribution.
\begin{lemma}
\label{lem:largest_element}
For any $\rs$ such that ${\cal B}_{\rs}(\Iset)$ is non-empty,
\begin{equation*}
F_{\Cb}(\cdot \vert \rs,\Iset) \in \Cb(\Iset),
\qquad
\overline{U}_{\Cb}(\rs-\vert \Iset) \in {\cal B}_{\rs}(\Iset),
\qquad \text{and} \qquad
\sup {\cal B}_{\rs}(\Iset)=\overline{U}_{\Cb}(\rs-\vert \Iset).
\end{equation*}
\end{lemma}

We now formally complete the proof of our main technical result.

\begin{proof}[\textbf{Proof of \Cref{thm:lambda_reduction}.}]
Fix a randomized pricing mechanism $\Psi \in \MechSet$. For compactness, let
$\eta
=
\inf_{F \in \Cb(\Iset)} \Ratio[\Psi,F].$ 
Let $\lambda
\in
\left[
\inf_{F \in \Cb(\Iset)} \Ratio[\Psi,F],
1
\right].$

For every $\rs \in \mathcal{S}(\Iset)$, let
$u_{\rs}
=
\overline{U}_{\Cb}(\rs-\vert \Iset)$ and, for every $\qs \in {\cal B}_{\rs}(\Iset)$, define
\begin{equation*}
A_{\rs}(\qs)
=
\Expect_{p \sim \Psi}
\left[
\Rev\left(p\vert L_{\Cb}(\cdot \vert \Iset \cup (\rs,\qs))\right)
\right],
\end{equation*}
and
$h_{\rs}(\qs)
=
\frac{A_{\rs}(\qs)}{\rs \cdot \qs}.$ Finally, define
\begin{equation*}
\rho_{\rs,\lambda}(\qs)
=
\lambda \cdot \rs \cdot \qs
-
A_{\rs}(\qs)
=
\rs \cdot \qs \cdot
\left(
\lambda-h_{\rs}(\qs)
\right).
\end{equation*}
This is the $\lambda$-regret associated with the candidate worst-case distribution corresponding to the pair $(\rs,\qs)$.

\noindent \textit{Step 1:} We first show that
\begin{equation}
\label{eq:Mlambda_nonnegative}
\sup_{\tilde{p} \in \mathcal{S}(\Iset)}
\rho_{\tilde{p},\lambda}(u_{\tilde{p}})
\ge
0.
\end{equation}
If $\lambda>\eta$, then by definition of $\eta$ there exists $F \in \Cb(\Iset)$ such that
$\Ratio[\Psi,F]<\lambda.$ Fix such $F$ and let  $\tilde{p} \in \mathcal{O}(F)$ and $\tilde q=\overline F(\tilde{p}-)$. Then $(\tilde{p},\tilde q)\in{\cal B}(\Iset)$ and $\opt(F)=\tilde{p} \cdot \tilde q$. 
We have,
\begin{equation*}
h_{\tilde{p}}(\tilde q)
=
\frac{A_{\tilde{p}}(\tilde q)}{\tilde{p} \cdot \tilde q}
\stackrel{(a)}{\leq}
\frac{\Rev(\Psi,F)}{\opt(F)}
=
\Ratio[\Psi,F]
<
\lambda,
\end{equation*}
where $(a)$ follows from \Cref{prop:worst_case_revenue}.

Since $h_{\tilde{p}}$ is non-increasing on ${\cal B}_{\rs}(\Iset)$ (\Cref{lem:monotonic_rho}) and $\tilde q \le u_{\tilde{p}}$ by \Cref{lem:largest_element}, we have
\begin{equation*}
h_{\tilde{p}}(u_{\tilde{p}})
\le
h_{\tilde{p}}(\tilde q)
<
\lambda.
\end{equation*}
 Consequently,
\begin{equation*}
\rho_{\tilde{p},\lambda}(u_{\tilde{p}})
=
\tilde{p} \cdot u_{\tilde{p}} \cdot
\left(
\lambda-h_{\tilde{p}}(u_{\tilde{p}})
\right)
>
0,
\end{equation*}
which proves \eqref{eq:Mlambda_nonnegative} in the case $\lambda>\eta$.

If $\lambda=\eta$, take a sequence $(F_n)_{n \ge 1}$ in $\Cb(\Iset)$ such that $\Ratio[\Psi,F_n] \to \eta.$
For every $n$, let $\tilde{p}_n \in \mathcal{O}(F_n)$ and $\tilde q_n=\overline F_n(\tilde{p}_n-)$. The argument above gives
\begin{equation*}
h_{\tilde{p}_n}(u_{\tilde{p}_n})
\le
h_{\tilde{p}_n}(\tilde q_n)
\le
\Ratio[\Psi,F_n].
\end{equation*}
Moreover, by \Cref{lem:largest_element}, $F_{\Cb}(\cdot \vert \tilde{p}_n,\Iset)\in\Cb(\Iset)$, and hence
\begin{equation*}
h_{\tilde{p}_n}(u_{\tilde{p}_n})
=
\Ratio[\Psi,F_{\Cb}(\cdot \vert \tilde{p}_n,\Iset)]
\ge
\eta.
\end{equation*}
Consequently, we have that,
$h_{\tilde{p}_n}(u_{\tilde{p}_n}) \to \eta.$
Moreover, given that, $0\le \tilde{p}_n \cdot u_{\tilde{p}_n}\le \ub$, we conclude that,
\begin{equation*}
\rho_{\tilde{p}_n,\eta}(u_{\tilde{p}_n})
=
\tilde{p}_n \cdot u_{\tilde{p}_n} \cdot 
\left(
\eta-h_{\tilde{p}_n}(u_{\tilde{p}_n})
\right)
\to 0,
\end{equation*}
which proves \eqref{eq:Mlambda_nonnegative} in the case $\lambda=\eta$ as well.

\noindent \textit{Step 2:} We next show that for every $\rs \in \mathcal{S}(\Iset)$ and every $\qs \in {\cal B}_{\rs}(\Iset)$,
\begin{equation}
\label{eq:rho_endpoint_bound}
\rho_{\rs,\lambda}(\qs)
\le
\sup_{\tilde{p} \in \mathcal{S}(\Iset)} \rho_{\tilde{p},\lambda}(u_{\tilde{p}}).
\end{equation}
Fix $\rs \in \mathcal{S}(\Iset)$ and $\qs \in {\cal B}_{\rs}(\Iset)$. Since $u_{\rs}=\sup{\cal B}_{\rs}(\Iset)$ by \Cref{lem:largest_element}, we have $\qs \le u_{\rs}$.

If $h_{\rs}(u_{\rs}) \le \lambda$, then the monotonicity of $h_{\rs}$ (\Cref{lem:monotonic_rho}) gives
$h_{\rs}(\qs) \ge h_{\rs}(u_{\rs}).$
Thus,
\begin{align*}
\rho_{\rs,\lambda}(\qs)
=
\rs \cdot \qs \cdot
\left(
\lambda-h_{\rs}(\qs)
\right) 
&\le
\rs \cdot \qs \cdot
\left(
\lambda-h_{\rs}(u_{\rs})
\right) \\
&\le
\rs \cdot u_{\rs} \cdot
\left(
\lambda-h_{\rs}(u_{\rs})
\right)
=
\rho_{\rs,\lambda}(u_{\rs})
\le
\sup_{\tilde{p} \in \mathcal{S}(\Iset)} \rho_{\tilde{p},\lambda}(u_{\tilde{p}}).
\end{align*}
If $h_{\rs}(u_{\rs})>\lambda$, then the monotonicity of $h_{\rs}$ implies that
$h_{\rs}(\qs) \ge h_{\rs}(u_{\rs})>\lambda.$
Therefore,
\begin{equation*}
\rho_{\rs,\lambda}(\qs)<0 \stackrel{(a)}{\le} \sup_{\tilde{p} \in \mathcal{S}(\Iset)} \rho_{\tilde{p},\lambda}(u_{\tilde{p}}),
\end{equation*}
where $(a)$ follows from \eqref{eq:Mlambda_nonnegative}.
This proves \eqref{eq:rho_endpoint_bound}.

\noindent \textit{Step 3:} We now complete the reduction. By \Cref{prop:worst_case_revenue},
\begin{equation*}
\begin{aligned}
\sup_{F \in \Cb(\Iset)} \LRegret[\lambda][\Psi,F]
&=
\sup_{(\rs,\qs) \in {\cal B}(\Iset)}
\left\{
\lambda \cdot \rs \cdot \qs
-
\Expect_{p \sim \Psi}
\left[
\Rev\left(p\vert L_{\Cb}(\cdot \vert \Iset \cup (\rs,\qs))\right)
\right]
\right\} \\
&=
\sup_{\rs \in \mathcal{S}(\Iset)}
\sup_{\qs \in {\cal B}_{\rs}(\Iset)}
\rho_{\rs,\lambda}(\qs).
\end{aligned}
\end{equation*}
The bound \eqref{eq:rho_endpoint_bound} shows that
\begin{equation*}
\sup_{\rs \in \mathcal{S}(\Iset)}
\sup_{\qs \in {\cal B}_{\rs}(\Iset)}
\rho_{\rs,\lambda}(\qs)
\le
\sup_{\rs \in \mathcal{S}(\Iset)}
\rho_{\rs,\lambda}(u_{\rs}).
\end{equation*}
The reverse inequality follows from \Cref{lem:largest_element}, which implies that $u_{\rs}\in{\cal B}_{\rs}(\Iset)$ for every $\rs \in \mathcal{S}(\Iset)$. Hence,
\begin{equation*}
\sup_{\rs \in \mathcal{S}(\Iset)}
\sup_{\qs \in {\cal B}_{\rs}(\Iset)}
\rho_{\rs,\lambda}(\qs)
=
\sup_{\rs \in \mathcal{S}(\Iset)}
\rho_{\rs,\lambda}(u_{\rs}).
\end{equation*}
Finally, by the definition of $F_{\Cb}(\cdot \vert \rs,\Iset)$ and $u_{\rs}=\overline{U}_{\Cb}(\rs-\vert \Iset)$, we have that
\begin{equation*}
\rho_{\rs,\lambda}(u_{\rs})
=
\LRegret[\lambda][\Psi,F_{\Cb}(\cdot \vert \rs,\Iset)].
\end{equation*}
Hence, we have established that,
\begin{equation*}
\sup_{F \in \Cb(\Iset)} \LRegret[\lambda][\Psi,F]
=
\sup_{\rs \in \mathcal{S}(\Iset)}
\LRegret[\lambda][\Psi,F_{\Cb}(\cdot \vert \rs,\Iset)],
\end{equation*}
which proves the result.
\end{proof}

\begin{proof}[\textbf{Proof of \Cref{lem:monotonic_rho}}]
Fix $p \in [\lb,\ub]$ and define
\begin{equation*}
\tilde{\rho}(\qs)
=
\frac{1}{\qs}
\Rev\left(p\vert L_{\Cb}(\cdot \vert \Iset \cup (\rs,\qs))\right).
\end{equation*}
We show that $\tilde{\rho}$ is non-increasing in $\qs$.

When $\Cb=\Gb$, the lower envelope $L_{\Gb}(\cdot \vert \Iset \cup (\rs,\qs))$ is piecewise constant. For fixed $p$, either the relevant value of the lower envelope is $\qs$, in which case $\tilde{\rho}(\qs)$ is constant, or the relevant value is an element of the original information set and does not depend on $\qs$, in which case $\tilde{\rho}(\qs)$ is proportional to $1/\qs$. In both cases, $\tilde{\rho}$ is non-increasing.

When $\Cb=\Fb$, let $\IsetPlus = \Iset \cup (\rs,\qs)$ be the extended information set and reindex the elements such that
$\IsetPlus= \{(\hat{p}_i,\hat{q}_i); \; i = 0,\ldots,N+2\}.$
Fix $p$ and let $i$ be such that $p \in [\hat{p}_{i},\hat{p}_{i+1})$. By definition of $L_{\Cb}$,
\begin{equation*}
\Rev\left(p\vert L_{\Cb}(\cdot \vert \IsetPlus) \right)
=
p \cdot
\overline{G}_{\Cb}
\left(
p - \vert
(\hat{p}_{i},\hat{q}_{i}),
(\hat{p}_{i+1},\hat{q}_{i+1})
\right).
\end{equation*}
Hence, $\Rev\left(p\vert L_{\Cb}(\cdot \vert \IsetPlus) \right)$ only varies as a function of $\qs$ if $\rs \in \{\hat{p}_{i},\hat{p}_{i+1}\}$. In all other cases, the numerator in $\tilde{\rho}$ is a nonnegative constant that does not depend on $\qs$, and therefore $\tilde{\rho}$ is non-increasing. We next show that when either $\hat{p}_{i}=\rs$ or $\hat{p}_{i+1}=\rs$, $\tilde{\rho}$ is also non-increasing.

\noindent \textbf{Case 1: $\hat{p}_{i}=\rs$.}
In this case, using the alternative form of $\overline{G}_{\Cb}$ derived in the proof of \Cref{lemma:singlecross}, we obtain
\begin{equation*}
\tilde{\rho}(\qs)
=
\frac{1}{\qs}
\cdot
p \cdot \qs
\cdot
\Gad{
\Gainv{\frac{\hat{q}_{i+1}}{\qs}}
\frac{p-\rs}{\hat{p}_{i+1}-\rs}
} 
=
p \cdot
\Gad{
\Gainv{\frac{\hat{q}_{i+1}}{\qs}}
\frac{p-\rs}{\hat{p}_{i+1}-\rs}
}.
\end{equation*}
By composition, $\tilde{\rho}$ is non-increasing in $\qs$.

\noindent \textbf{Case 2: $\hat{p}_{i+1}=\rs$.}
Using again one of the alternative forms of $\overline{G}_{\Cb}$ derived in the proof of \Cref{lemma:singlecross}, we obtain
\begin{equation*}
\tilde{\rho}(\qs)
=
\frac{1}{\qs}
\cdot
\frac{p}{
\frac{1}{\hat{q}_i}
+
\left(
\frac{1}{\qs}
-
\frac{1}{\hat{q}_i}
\right)
\frac{p-\hat{p}_i}{\rs-\hat{p}_i}
} 
=
\frac{p}{
\frac{\qs}{\hat{q}_i}
\cdot
\left(
1-\frac{p-\hat{p}_i}{\rs-\hat{p}_i}
\right)
+
\frac{p-\hat{p}_i}{\rs-\hat{p}_i}
}.
\end{equation*}
The denominator is non-decreasing in $\qs$, and therefore $\tilde{\rho}$ is non-increasing in $\qs$.
\end{proof}

\begin{proof}[\textbf{Proof of \Cref{lem:largest_element}}]
\textit{Step 1:} We first show that
\begin{equation*}
\sup {\cal B}_{\rs}(\Iset)
\le
\overline{U}_{\Cb}(\rs-\vert \Iset).
\end{equation*}
Let $q \in {\cal B}_{\rs}(\Iset)$ and let $F \in \Cb(\Iset)$ be such that $\rs$ is an optimal price for $F$ with associated conversion rate $q$. By the definition of the upper envelope and \Cref{lemma:singlecross},
\begin{equation*}
q
=
\overline{F}(\rs-)
\le
\overline{U}_{\Cb}(\rs-\vert \Iset).
\end{equation*}
Taking the supremum over $q \in {\cal B}_{\rs}(\Iset)$ proves the desired upper bound.

\textit{Step 2:} Next, we show that
\begin{equation*}
\overline{U}_{\Cb}(\rs-\vert \Iset)
\in
{\cal B}_{\rs}(\Iset).
\end{equation*}
Let $u_{\rs}
=
\overline{U}_{\Cb}(\rs-\vert \Iset)$
and consider the extended information set
$\IsetPlus
=
\Iset \cup (\rs,u_{\rs}).$

By construction of the upper envelope, adding the point $(\rs,u_{\rs})$ preserves feasibility of the information set. Indeed, when $\Cb=\Gb$, this follows from monotonicity of the ccdf, and when $\Cb=\Fb$, this follows from the definition of $\overline{U}_{\Fb}$ as the largest value that preserves the monotonicity of the discrete slopes in the $\Gamma^{-1}$ scale. Therefore, $\Cb(\IsetPlus)$ is non-empty.

By \Cref{lem:L_belongs}, this implies that
$ L_{\Cb}(\cdot \vert \IsetPlus) \in \Cb(\IsetPlus) \subset \Cb(\Iset).$
By definition of $F_{\Cb}(\cdot \vert \rs,\Iset)$, we have
$F_{\Cb}(\cdot \vert \rs,\Iset)
=
L_{\Cb}(\cdot \vert \IsetPlus),$
and therefore
\begin{equation*}
F_{\Cb}(\cdot \vert \rs,\Iset) \in \Cb(\Iset)
\qquad \text{and} \qquad
\overline{F}_{\Cb}(\rs-\vert \rs,\Iset)=u_{\rs}.
\end{equation*}

It remains to show that $\rs$ is an optimal price for $F_{\Cb}(\cdot \vert \rs,\Iset)$. By \Cref{lem:max_L}, it suffices to prove that
$\rs \cdot u_{\rs}
\ge
\max_i p_i \cdot q_i.$
Since ${\cal B}_{\rs}(\Iset)$ is non-empty, take $q \in {\cal B}_{\rs}(\Iset)$. Then there exists $F \in \Cb(\Iset)$ such that $\overline F(\rs-)=q$ and $\rs$ is an optimal price for $F$. Hence, for every $i \in \{0,\ldots,N+1\}$,
\begin{equation*}
\rs \cdot u_{\rs}
\ge
\rs \cdot q
=
\rs \cdot \overline F(\rs-)
\ge
p_i \cdot \overline F(p_i-)
=
p_i \cdot q_i,
\end{equation*}
where the first inequality follows from Step 1. Therefore, $\rs$ is an optimal price for $F_{\Cb}(\cdot \vert \rs,\Iset)$ with associated conversion rate $u_{\rs}$. Thus,
\begin{equation*}
u_{\rs}
=
\overline{U}_{\Cb}(\rs-\vert \Iset)
\in
{\cal B}_{\rs}(\Iset).
\end{equation*}
Together with Step 1, this proves
\begin{equation*}
\sup {\cal B}_{\rs}(\Iset)
=
\overline{U}_{\Cb}(\rs-\vert \Iset),
\end{equation*}
and concludes the proof.
\end{proof}

\subsection{Proof of \Cref{prop:r_certificate} and Additional Properties on $\mathcal{S}(\Iset)$} \label{sec:apx_S}

\begin{proof}[\textbf{Proof of \Cref{prop:r_certificate}}]
First, if $\rs \in \mathcal{S}(\Iset)$, then by definition ${\cal B}_{\rs}(\Iset)$ is non-empty. Hence, \Cref{lem:largest_element} implies that $\overline{U}_{\Cb}(\rs-\vert \Iset) \in {\cal B}_{\rs}(\Iset)$. By definition of ${\cal B}_{\rs}(\Iset)$, this implies that $i)$ and $ii)$ are satisfied.

Reciprocally, if $\Cb(\Iset \cup (\rs,\overline{U}_{\Cb}(\rs-\vert \Iset)))$ is non-empty, then \Cref{lem:L_belongs} implies that the distribution $\overline{L}_{\Cb}\left(v \vert \Iset \cup (\rs,\overline{U}_{\Cb}(\rs-\vert \Iset))\right) \in \Cb(\Iset \cup (\rs,\overline{U}_{\Cb}(\rs-\vert \Iset))) \subset \Cb(\Iset)$. Furthermore, \Cref{lem:max_L}, together with the condition that $\rs \cdot \overline{U}_{\Cb}(\rs-\vert \Iset) \geq \max_{i \in \{0,\ldots,N+1\}} p_i \cdot q_i$, implies that $\rs$ is an optimal price for $\overline{L}_{\Cb}\left(v \vert \Iset \cup (\rs,\overline{U}_{\Cb}(\rs-\vert \Iset))\right)$ with associated conversion rate $\overline{U}_{\Cb}(\rs-\vert \Iset)$. This proves that $\overline{U}_{\Cb}(\rs-\vert \Iset) \in {\cal B}_{\rs}(\Iset)$. Hence, ${\cal B}_{\rs}(\Iset)$ is non-empty, and therefore $\rs \in \mathcal{S}(\Iset)$.
\end{proof}

In what follows, we derive explicit expressions for the set $\mathcal{S}(\Iset)$ in the two classes $\Cb=\Gb$ and $\Cb=\Fb$.

For the regular case, it is convenient to introduce the inverse revenue map associated with one of the
building-block distributions $\overline G_{\Fb}$. For any $v>0$ and any two points $(s,q)$ and $(s',q')$
such that $0 \le s < s'$ and $1 \ge q \ge q' > 0$, define
\begin{equation}
\label{eq:R_inv}
\mathcal{R}^{-1}\!\left(v \vert (s,q),(s',q')\right)
=
\frac{1+\Gamma^{-1}(q)- \frac{\Gamma^{-1}(q')-\Gamma^{-1}(q)}{s'-s} s}{\frac{1}{v}-\frac{\Gamma^{-1}(q')-\Gamma^{-1}(q)}{s'-s}}.
\end{equation}
A direct calculation shows that $\mathcal{R}^{-1}\!\left(v \vert (s,q),(s',q')\right)$ is the unique solution
$r$ to
\begin{equation*}
r \cdot \overline G_{\Fb}\!\left(r \vert (s,q),(s',q')\right)=v.
\end{equation*}

\begin{proposition}[Structure of $\mathcal{S}(\Iset)$ for $\Cb=\Gb$]
\label{prop:S_general}
Assume that $\Gb(\Iset)$ is non-empty and that $\max_{i \in \{0,\ldots,N+1\}} p_i \cdot q_i > 0$.
Then
\begin{equation*}
\mathcal{S}(\Iset)
=
\bigcup_{i=0}^{N}
\left[
\max\left\{
p_i,\,
\frac{\max_{j \in \{0,\ldots,N+1\}} p_j q_j}{q_i}
\right\},
\,
p_{i+1}
\right),
\end{equation*}
where empty intervals are omitted.
\end{proposition}

\begin{proof}[\textbf{Proof of \Cref{prop:S_general}}]
Fix $i \in \{0,\ldots,N\}$ and $\rs \in [p_i,p_{i+1})$. By definition of $\overline U_{\Gb}$,
we have that, $\overline U_{\Gb}(\rs-\vert \Iset)=q_i.$

Furthermore, \Cref{lemma:feasible-set} implies that $(q_i)_{0 \le i \le N+1}$ is non-increasing.
Therefore, adding the point $(\rs,q_i)$ preserves feasibility in the general class, and hence
\begin{equation*}
\Gb\bigl(\Iset \cup (\rs,\overline U_{\Gb}(\rs-\vert \Iset))\bigr)
=
\Gb\bigl(\Iset \cup (\rs,q_i)\bigr)
\neq \emptyset.
\end{equation*}
Applying \Cref{prop:r_certificate}, we conclude that $\rs \in \mathcal{S}(\Iset)$ if and only if
$\rs \cdot q_i \ge \max_{j \in \{0,\ldots,N+1\}} p_j \cdot q_j,$ or equivalently, if and only if
\begin{equation*}
\rs \ge
\frac{\max_{j \in \{0,\ldots,N+1\}} p_j \cdot q_j}{q_i}.
\end{equation*}
Intersecting this condition with the interval $[p_i,p_{i+1})$ yields exactly
\begin{equation*}
\mathcal{S}(\Iset)\cap [p_i,p_{i+1})
=
\left[
\max\left\{
p_i,\,
\frac{\max_{j \in \{0,\ldots,N+1\}} p_j q_j}{q_i}
\right\},
\,
p_{i+1}
\right),
\end{equation*}
with the understanding that the interval is empty when the left endpoint exceeds $p_{i+1}$.
Taking the union over $i \in \{0,\ldots,N\}$ proves the result.
\end{proof}

\begin{proposition}[Structure of $\mathcal{S}(\Iset)$ for $\Cb=\Fb$]
\label{prop:S_regular}
Assume that $\Fb(\Iset)$ is non-empty and that $\max_{j \in \{0,\ldots,N+1\}} p_j \cdot q_j > 0$.
Define
\begin{equation*}
i_\ell = \min \argmax_{j \in \{0,\ldots,N+1\}} p_j \cdot q_j,
\qquad
i_r = \max \argmax_{j \in \{0,\ldots,N+1\}} p_j \cdot q_j.
\end{equation*}
Then the set $\mathcal{S}(\Iset)$ is a single interval
$\mathcal{S}(\Iset)= \left[\underline r_{\Fb}(\Iset), \overline r_{\Fb}(\Iset)\right] \cap [\lb,\ub),$
where
\begin{equation*}
\underline r_{\Fb}(\Iset)
=
\begin{cases}
\lb,
& \text{if } i_\ell = 0,\\[4pt]
p_1 \cdot q_1,
& \text{if } i_\ell = 1,\\[4pt]
\mathcal{R}^{-1}\!\left(
p_{i_\ell} \cdot q_{i_\ell}
\vert
(p_{i_\ell -2},q_{i_\ell -2}),
(p_{i_\ell -1},q_{i_\ell -1})
\right),
& \text{if } i_\ell \ge 2,
\end{cases}
\end{equation*}
and
\begin{equation*}
\overline r_{\Fb}(\Iset)
=
\begin{cases}
\mathcal{R}^{-1}\!\left(
p_{i_r} \cdot q_{i_r}
\vert
(p_{i_r +1},q_{i_r +1}),
(p_{i_r +2},q_{i_r +2})
\right),
& \text{if } i_r \le N-2,\\[4pt]
p_N,
& \text{if } i_r = N-1,\\[4pt]
\ub,
& \text{if } i_r = N.
\end{cases}
\end{equation*}
\end{proposition}

\begin{proof}[\textbf{Proof of \Cref{prop:S_regular}}]
Let
$\rho = \max_{i \in \{0,\ldots,N+1\}} p_i \cdot q_i.$
We first note that, for every $\rs \in [\lb,\ub)$, the set
$\Fb(\Iset \cup (\rs,\overline U_{\Fb}(\rs-\vert \Iset)))$ is non-empty. Indeed, in the
$\Gamma^{-1}$ scale, the quantity $\overline U_{\Fb}(\rs-\vert \Iset)$ is constructed precisely so
as to preserve the monotonicity of the discrete slopes in \Cref{lemma:feasible-set}. Therefore,
\Cref{prop:r_certificate} implies that for every $\rs \in [\lb,\ub)$
\begin{equation}
\label{eq:S_regular_reduction_new}
\rs \in \mathcal{S}(\Iset)
\qquad \Longleftrightarrow \qquad
\rs \cdot \overline U_{\Fb}(\rs-\vert \Iset) \ge \rho.
\end{equation}

Next, since $\Fb(\Iset)$ is non-empty, there exists $F \in \Fb(\Iset)$. The revenue curve
$p \mapsto p \cdot \overline F(p-)$ is unimodal for every regular distribution. Hence, the sequence
$(p_i \cdot q_i)_{i=0}^{N+1}$ is unimodal as it corresponds to the sequence $( \Rev(p_i \vert  F) )_{i=0}^{N+1} $ for some $F \in \Fb(\Iset)$. In particular, the maximizing indices form a contiguous block, namely
\begin{equation*}
\argmax_{j \in \{0,\ldots,N+1\}} p_j \cdot q_j = \{i_\ell,\ldots,i_r\}.
\end{equation*}

For every $k \in \{0,\ldots,N-1\}$, define
\begin{equation*}
R_k(r)
=
r \cdot \overline G_{\Fb}\!\left(r \vert (p_k,q_k),(p_{k+1},q_{k+1})\right).
\end{equation*}
By \Cref{lem:const}, the function $R_k$ is monotone. Moreover,
\begin{equation*}
R_k(p_k)=p_k \cdot q_k,
\qquad
R_k(p_{k+1})=p_{k+1} \cdot q_{k+1},
\end{equation*}
so $R_k$ is non-decreasing when $p_k \cdot q_k \le p_{k+1} \cdot q_{k+1}$ and non-increasing otherwise.

We next make explicit the form of the quantity $r \cdot \overline U_{\Fb}(r-\vert \Iset)$ on each interval.
By definition of $\overline U_{\Fb}$,
\begin{equation}
\label{eq:explicit_upper_revenue_regular}
r \cdot \overline U_{\Fb}(r-\vert \Iset)
=
\begin{cases}
\min\{r, R_1(r)\},
& \text{if } r \in [p_0,p_1),\\[4pt]
\min\{R_{i-1}(r),R_{i+1}(r)\},
& \text{if } r \in [p_i,p_{i+1}), \text{ for } i \in \{1,\ldots,N-2\},\\[4pt]
R_{N-2}(r),
& \text{if } r \in [p_{N-1},p_N),\\[4pt]
R_{N-1}(r),
& \text{if } r \in [p_N,p_{N+1}).
\end{cases}
\end{equation}
We will refer to the functions appearing on the right-hand side of \eqref{eq:explicit_upper_revenue_regular}
as the candidate revenue curves on the corresponding interval.

We now identify the set of prices satisfying \eqref{eq:S_regular_reduction_new}.

\smallskip
\noindent \textit{Step 1:} We first show that $\mathcal{S}(\Iset) \subset [p_{i_\ell-1},p_{i_r+1}).$

Consider an interval $[p_i,p_{i+1})$ with $i \le i_\ell -2$.
If $i=0$, then the candidate revenue curves on $[p_0,p_1)$ are $r$ and $R_1(r)$. Since $i_\ell \ge 2$ in that case, we have
$p_1 \cdot q_1 \le p_2 \cdot q_2$, and therefore $R_1$ is non-decreasing.
If $i \in \{1,\ldots,i_\ell-2\}$, then the candidate revenue curves on $[p_i,p_{i+1})$ are
$R_{i-1}$ and $R_{i+1}$. Since the sequence $(p_j \cdot q_j)_j$ is non-decreasing up to index $i_\ell$,
both of these functions are non-decreasing.
Hence, in all cases, the mapping
$r \mapsto r \cdot \overline U_{\Fb}(r-\vert \Iset)$ is non-decreasing on $[p_i,p_{i+1})$.
Moreover, by \eqref{eq:explicit_upper_revenue_regular},
\begin{equation*}
r \cdot \overline U_{\Fb}(r-\vert \Iset)
\le p_{i+1} \cdot q_{i+1}
< \rho
\qquad \text{for all } r \in [p_i,p_{i+1}).
\end{equation*}

Similarly, if $i \ge i_r +1$, then every candidate revenue curve entering
\eqref{eq:explicit_upper_revenue_regular} on $[p_i,p_{i+1})$ is non-increasing, and thus
$r \mapsto r \cdot \overline U_{\Fb}(r-\vert \Iset)$ is non-increasing on $[p_i,p_{i+1})$.
Again by \eqref{eq:explicit_upper_revenue_regular},
\begin{equation*}
r \cdot \overline U_{\Fb}(r-\vert \Iset)
\le p_i \cdot q_i
< \rho
\qquad \text{for all } r \in [p_i,p_{i+1}).
\end{equation*}
Therefore, \eqref{eq:S_regular_reduction_new} implies that $\mathcal{S}(\Iset) \subset [p_{i_\ell-1},p_{i_r+1})$ 
with the convention that $p_{i_\ell-1} = \lb$ when $i_\ell=0$.

\smallskip
\noindent \textit{Step 2:} We next show that $[p_{i_\ell},p_{i_r}] \subseteq \mathcal{S}(\Iset).$

If $i_\ell = i_r$, we note that $p_{i_\ell} \in \mathcal{S}(\Iset)$ by
\Cref{prop:r_certificate}, since
\begin{equation*}
\Fb(\Iset \cup (p_{i_\ell},q_{i_\ell})) = \Fb(\Iset)
\qquad \text{and} \qquad
p_{i_\ell} \cdot q_{i_\ell} = \rho.
\end{equation*}
This implies that $[p_{i_\ell},p_{i_r}] \subseteq \mathcal{S}(\Iset)$ in that case.

Next, assume that $i_\ell < i_r$ and, fix $i \in \{i_\ell,\ldots,i_r-1\}$. By \eqref{eq:S_regular_reduction_new} it is sufficient to prove that 
\begin{equation}
\label{eq:condition_over_i}
r \cdot \overline U_{\Fb}(r-\vert \Iset) \ge \rho
\qquad \text{for all } r \in [p_i,p_{i+1}).
\end{equation}

If $i=0$, then the candidate revenue curves on $[p_0,p_1)$ are $r$ and $R_1(r)$.
Since $p_0 \cdot q_0 = p_1 \cdot q_1 = \rho$, the function $r \mapsto r$ is non-decreasing and equals $\rho$ at $p_0$, while
$R_1$ is non-increasing and equals $\rho$ at $p_1$. Hence both $r \geq \rho$ and $R_1(r) \geq \rho$ for every $r \in [p_0,p_{1})$ and the definition \eqref{eq:explicit_upper_revenue_regular} implies that \eqref{eq:condition_over_i} holds in that case.

If $i \in \{1,\ldots,i_r-1\}$, then the candidate revenue curves on $[p_i,p_{i+1})$ are
$R_{i-1}$ and $R_{i+1}$.
Because $i \ge i_\ell$, the function $R_{i-1}$ is non-decreasing and satisfies
$R_{i-1}(p_i)=p_i \cdot q_i=\rho$.
Because $i+1 \le i_r$, the function $R_{i+1}$ is non-increasing and satisfies
$R_{i+1}(p_{i+1})=p_{i+1} \cdot q_{i+1}=\rho$.
Hence, in all cases, both candidate revenue curves are at least $\rho$ throughout the interval, and  \eqref{eq:explicit_upper_revenue_regular} implies that \eqref{eq:condition_over_i} holds in that case.

Hence, we conclude from \eqref{eq:S_regular_reduction_new} that
\begin{equation*}
[p_{i_\ell},p_{i_r}] \subseteq \mathcal{S}(\Iset).
\end{equation*}

Next, define $\underline r_{\Fb}(\Iset) = \inf \mathcal{S}(\Iset)$ and $\overline r_{\Fb}(\Iset) \sup \mathcal{S}(\Iset)$. We note that these quantities are well-defined and finite because $\mathcal{S}(\Iset)$ is non-empty and bounded. 

Combining Steps 1 and 2, we have established that
\begin{equation*}
[p_{i_\ell},p_{i_r}] \subseteq \mathcal{S}(\Iset) \subseteq [p_{i_\ell-1},p_{i_r+1}).
\end{equation*}
Consequently,
\begin{equation}
\label{eq:range_boundary}
\underline r_{\Fb}(\Iset) \in [p_{i_\ell-1},p_{i_\ell}]
\qquad \mbox{and} \qquad 
\overline r_{\Fb}(\Iset) \in [p_{i_r},p_{i_r+1}).
\end{equation}
In what follows, it remains to characterize these two boundary values and prove that $\mathcal{S}(\Iset)$ is an interval.

\smallskip
\noindent \textit{Step 3:} We now characterize $\underline r_{\Fb}(\Iset)$.

\underline{Case (a):} If $i_\ell = 0$, then Step 1 already implies that $\underline r_{\Fb}(\Iset)=p_0$.

\underline{Case (b):} Assume $i_\ell = 1$. On $[p_0,p_1)$, we have
\begin{equation*}
r \cdot \overline U_{\Fb}(r-\vert \Iset) = \min\{r,R_1(r)\}.
\end{equation*}
Since $p_1 \cdot q_1=\rho \ge p_2 \cdot q_2$, the function $R_1$ is non-increasing, and therefore
\begin{equation*}
R_1(r) \ge R_1(p_1)=\rho
\qquad \text{for all } r \in [p_0,p_1).
\end{equation*}
Hence, on this interval,
\begin{equation*}
r \cdot \overline U_{\Fb}(r-\vert \Iset) \ge \rho
\qquad \Longleftrightarrow \qquad
r \ge \rho.
\end{equation*}
This implies that $\underline r_{\Fb}(\Iset) = \rho = p_1 \cdot q_1$ and that,
\begin{equation}
\label{eq:interval_S_left_1}
\mathcal{S}(\Iset) \cap [p_{0},p_{1})
=
[\underline r_{\Fb}(\Iset),p_{1}).
\end{equation}

\underline{Case (c):} Assume now that $i_\ell \ge 2$. On the interval $[p_{i_\ell-1},p_{i_\ell})$, we have
\begin{equation*}
r \cdot \overline U_{\Fb}(r-\vert \Iset)
=
\min\{R_{i_\ell-2}(r),R_{i_\ell}(r)\}.
\end{equation*}
Because $i_\ell$ is the first maximizing index, the function $R_{i_\ell-2}$ is non-decreasing, while
$R_{i_\ell}$ is non-increasing. Moreover, $R_{i_\ell}(p_{i_\ell}) = p_{i_\ell} \cdot q_{i_\ell} = \rho$,
hence for ever $r \in [p_{i_\ell-1},p_{i_\ell})$, we have that $R_{i_\ell}(r)\ge \rho$.
Therefore,
\begin{equation*}
r \cdot \overline U_{\Fb}(r-\vert \Iset) \ge \rho
\qquad \Longleftrightarrow \qquad
R_{i_\ell-2}(r)\ge \rho.
\end{equation*}
Next, remark that $R_{i_\ell-2}$ is continuous and monotonic on $[p_{i_\ell-1},p_{i_\ell}]$ and it satisfies, that
$R_{i_\ell-2}(p_{i_\ell-1}) = p_{i_\ell-1} \cdot q_{i_\ell-1} < \rho$
and $R_{i_\ell-2}(p_{i_\ell}) \ge \rho.$
Consequently, the intermediate value theorem implies that there exists a unique threshold
$\underline r_{\Fb}(\Iset) \in [p_{i_\ell-1},p_{i_\ell}]$ such that
\begin{equation*}
R_{i_\ell-2}(r)\ge \rho
\qquad \Longleftrightarrow \qquad
r \ge \underline r_{\Fb}(\Iset),
\end{equation*}
and, by definition of $\mathcal{R}^{-1}$ (see \eqref{eq:R_inv}), such a threshold satisfies
\begin{equation*}
\underline r_{\Fb}(\Iset)
=
\mathcal{R}^{-1}\!\left(
\rho
\vert
(p_{i_\ell -2},q_{i_\ell -2}),
(p_{i_\ell -1},q_{i_\ell -1})
\right).
\end{equation*}
In particular, we have also established that
\begin{equation}
\label{eq:interval_S_left_2}
\mathcal{S}(\Iset) \cap [p_{i_\ell-1},p_{i_\ell})
=
[\underline r_{\Fb}(\Iset),p_{i_\ell}).
\end{equation}

\smallskip
\noindent \textit{Step 4:}  We now characterize $\overline r_{\Fb}(\Iset)$.

\underline{Case (a):} If $i_r = N$, then on $[p_N,\ub)$ we have that $r \cdot \overline U_{\Fb}(r-\vert \Iset)=R_{N-1}(r).$
Since $p_{N-1} \cdot q_{N-1} \le p_N \cdot q_N=\rho$, the function $R_{N-1}$ is non-decreasing, and
$R_{N-1}(p_N)=\rho.$
Therefore,
\begin{equation*}
r \cdot \overline U_{\Fb}(r-\vert \Iset)\ge \rho
\qquad \text{for all } r \in [p_N,\ub),
\end{equation*}
which proves that $\overline r_{\Fb}(\Iset)=\ub.$

\underline{Case (b):} If $i_r = N-1$, then on $[p_{N-1},p_N)$ we have that
$r \cdot \overline U_{\Fb}(r-\vert \Iset)=R_{N-2}(r).$

Since $p_{N-2} \cdot q_{N-2} \le p_{N-1} \cdot q_{N-1}=\rho$, the function $R_{N-2}$ is non-decreasing, and
$R_{N-2}(p_{N-1})=\rho.$
Hence,
\begin{equation*}
[p_{N-1},p_N) \subseteq \mathcal{S}(\Iset).
\end{equation*}
Together with \eqref{eq:range_boundary}, we conclude that $r_{\Fb}(\Iset)=p_N$.

\underline{Case (c):} Assume finally that $i_r \le N-2$. On the interval $[p_{i_r},p_{i_r+1})$, we have
\begin{equation*}
r \cdot \overline U_{\Fb}(r-\vert \Iset)
=
\min\{R_{i_r-1}(r),R_{i_r+1}(r)\}.
\end{equation*}
Because $i_r$ is the last maximizing index, the function $R_{i_r-1}$ is non-decreasing, while $R_{i_r+1}$ is non-increasing. Moreover, $R_{i_r-1}(p_{i_r})=p_{i_r} \cdot q_{i_r}=\rho,$ which implies that for every $r \in [p_{i_r},p_{i_r+1})$, we have that $R_{i_r-1}(r)\ge \rho.$

Therefore,
\begin{equation*}
r \cdot \overline U_{\Fb}(r-\vert \Iset) \ge \rho
\qquad \Longleftrightarrow \qquad
R_{i_r+1}(r)\ge \rho.
\end{equation*}

Next, remark that $R_{i_r+1}$ is continuous and monotonic on $[p_{i_r},p_{i_r+1})$, and it satisfies, $R_{i_r+1}(p_{i_r}) \ge \rho$ and $R_{i_r+1}(p_{i_r+1})=p_{i_r+1} \cdot q_{i_r+1}<\rho$. Consequently, the intermediate value theorem implies  that there exists a unique threshold
$\overline r_{\Fb}(\Iset) \in [p_{i_r},p_{i_r+1})$ such that
\begin{equation*}
R_{i_r+1}(r)\ge \rho
\qquad \Longleftrightarrow \qquad
r \le \overline r_{\Fb}(\Iset)
\end{equation*}
and, by definition of $\mathcal{R}^{-1}$ (see \eqref{eq:R_inv}), such a threshold satisfies
\begin{equation*}
\overline r_{\Fb}(\Iset)
=
\mathcal{R}^{-1}\!\left(
\rho
\vert
(p_{i_r +1},q_{i_r +1}),
(p_{i_r +2},q_{i_r +2})
\right).
\end{equation*}
In particular, we have also established that
\begin{equation}
\label{eq:interval_S_right}
\mathcal{S}(\Iset) \cap [p_{i_r},p_{i_r+1})
=
[p_{i_r},\overline r_{\Fb}(\Iset)] \cap [\lb,\ub).
\end{equation}

\textit{Step 5:}  We finally establish that $\mathcal{S}(\Iset)$ is an interval.
Combining Steps 2, with \eqref{eq:interval_S_left_1}, \eqref{eq:interval_S_left_2} and \eqref{eq:interval_S_right}, we obtain that
\begin{equation*}
\mathcal{S}(\Iset)
=
[\underline r_{\Fb}(\Iset),p_{i_\ell})
\cup
[p_{i_\ell},p_{i_r}]
\cup
[p_{i_r},\overline r_{\Fb}(\Iset)] \cap [\lb,\ub)
=
[\underline r_{\Fb}(\Iset),\overline r_{\Fb}(\Iset)] \cap [\lb,\ub).
\end{equation*}
This concludes the proof.
\end{proof}

\subsection{Auxiliary Results and their Proofs}

\begin{proposition}[Feasible information sets]\label{lemma:feasible-set}
Let $\Iset = \{(p_i, q_i): i = 0 \ldots, N+1\}$ be a set of sorted increasing prices and associated ccdf values with $(p_0, q_0)=(\lb,1)$ and $(p_{N+1}, q_{N+1}) = (\ub, 0)$. Then:
\begin{itemize}
    \item $\Gb(\Iset)$ is non-empty if and only if the sequence $(q_i)_{0 \leq i \leq N+1}$ is non-increasing.
    \item  $\Fb(\Iset)$ is non-empty if and only if the sequence $(q_i)_{0 \leq i \leq N+1}$ is non-increasing and the sequence of slopes $ \left(\frac{\Gainv{q_{i+1}}-\Gainv{q_{i}}}{p_{i+1}-p_{i}}\right)_{0 \leq i \leq N-1}$ is non-decreasing.
\end{itemize}
\end{proposition}

\begin{proof}[\textbf{\underline{Proof of \Cref{lemma:feasible-set}}}]
Let $\Iset = \{(p_i, q_i): i = 0 \ldots, N+1\}$ a set of sorted increasing prices and associated ccdf values with $(p_0, q_0)=(\lb,1)$ and $(p_{N+1}, q_{N+1}) = (\ub, 0)$.\\
\underline{Case $\Cb = \Fb:$} \\
$\Longrightarrow$) If $\Iset$ is $\Fb$-feasible then there exists a distribution characterized by its cdf $F$ in $\aDistSetqa{\Iset}$. Then by definition of a cdf, the sequence $(q_i = \bF(p_i))_{0 \leq i \leq N+1}$ is non-increasing. Fix $1 \leq i \leq N-1$, Using \Cref{lemma:singlecross} on $\bF$ with $(s,q) = (p_{i-1}, q_{i-1})$ and $(s',q') = (p_{i+1},q_{i+1})$, we have
\begin{equation*}
\bF(p_i) = q_i \geq \Gad{\Gainv{q_{i-1}} + \frac{\Gainv{q_{i+1}}-\Gainv{q_{i-1}}}{p_{i+1}-p_{i-1}} (p_{i}-p_{i-1})}
\end{equation*}
which implies that
$\frac{\gqi-\gqim}{\pii-\pim} \leq \frac{\gqip-\gqim}{\pip -\pim}.$ Furthermore, note that
\begin{align*}
\frac{\gqip-\gqim}{\pip -\pim} = \frac{\Gamma^{-1}(\bF(p_{i+1})) -\Gamma^{-1}(\bF(p_{i-1}))}{\pip -\pim} \leq \frac{\Gamma^{-1}(\bF(p_{i+1})) -\Gamma^{-1}(F(p_{i}))}{\pip -p_{i}} = \frac{\gqip-\gqi}{\pip -\pii},
\end{align*}
where the inequality holds because $\Gamma^{-1}(\bF(\cdot))$ is convex (see proof of Lemma 2 in \cite{ABBSamples}) and the slope function of a convex function is non-decreasing.

Therefore, $\frac{\gqi-\gqim}{\pii-\pim} \leq \frac{\gqip-\gqi}{\pip -\pii}$, which proves the implication.

$\Longleftarrow$) Assume that the sequence $(q_i)_{0 \leq i \leq N+1}$ is non-increasing and $\left(\frac{\Gainv{q_{i+1}}-\Gainv{q_{i}}}{p_{i+1}-p_{i}}\right)_{0 \leq i \leq N-1}$ is non-decreasing.
Hence, \Cref{lem:L_belongs} implies that $\overline{L}_{\Fb}(\cdot \vert  \Iset)$ defined in \eqref{eq:L} belongs to $\Fb(\Iset)$. \\
\underline{Case $\Cb = \Gb:$} \\
$\Longrightarrow$) If $\Iset$ is $\Gb$-feasible then there exists a distribution characterized by its cdf $F$ in $\Gb(\Iset)$. Then by definition of a cdf, the sequence $(q_i = \bF(p_i))_{0 \leq i \leq N+1}$ is non-increasing. 

$\Longleftarrow$) It follows from \Cref{lem:L_belongs} that $\overline{L}_{\Gb}(\cdot \vert  \Iset)$ defined in \eqref{eq:L} belongs to $\Gb(\Iset)$
\end{proof}

\begin{lemma}\label{lem:const} 
Fix three scalars $\alpha, \beta > 0$ and $s\ge0$ and a cumulative distribution function $F$ such that $\overline{F} = \Gad{\alpha + \beta(v-s)}$  for $v \ge s$. Then $F$ admits a constant virtual value  function equal to  $s - \frac{1+ \alpha}{\beta}.$
\end{lemma}

\begin{proof}[\textbf{\underline{Proof of \Cref{lem:const}}}]
Let us first explicitly compute the virtual value function.  The derivative of $\Gad{\alpha + \beta(v-s)}$, for any $v \ge s$, is equal to $- \beta \Gad{\alpha + \beta(v-s)}^2.$
Therefore, the virtual value function evaluated at $v \ge s$ is equal to,
\begin{equation*}
 v - \frac{ \Gad{\alpha + \beta(v-s)}}{ \beta \Gad{\alpha + \beta(v-s)}^2 } 
=  v - \frac{1}{\beta \Gad{\alpha + \beta(v-s)}}  
= v - \frac{1+ \alpha + \beta (v-s)}{\beta} 
= s - \frac{1+ \alpha}{\beta}.\qedhere
\end{equation*}	
\end{proof}

\section{Proof of Results in \Cref{sec:value_information}}\label{sec:apx_value_information}

In this appendix, we prove \Cref{thm:lp_discrete_bound}. We proceed in five steps. First, we quantify the loss incurred by restricting attention to  mechanisms supported on a grid. Second, we analyze the structure of the upper envelope and show that, on an admissible grid, the continuum constraints can be reduced to finitely many endpoint constraints. Third, we prove that the resulting linear program provides a certified lower bound. Fourth, we derive a complementary upper bound and quantify the discretization error. Finally, we combine these results to complete the proof.

\subsection{Discretization of the Mechanism Space}
We first establish a structural result on the support of the optimal randomized pricing policy.

\begin{proposition}[Reduction to mechanisms supported on $\mathcal{S}(\Iset)$]
\label{prop:support_reduction}
Fix an information set $\Iset=\{(p_i,q_i)\}_{i=0}^{N+1}$ and a class $\Cb\in\{\Gb,\Fb\}$, and assume that $\Cb(\Iset)$ is non-empty. Then for every mechanism $\Psi \in \MechSet$, there exists a mechanism $\widetilde{\Psi} \in \MechSet$ supported on $\mathcal{S}(\Iset)$ such that
\begin{equation*}
\inf_{F\in\Cb(\Iset)} \frac{\Rev(\widetilde{\Psi},F)}{\opt(F)}
\;\ge\;
\inf_{F\in\Cb(\Iset)} \frac{\Rev(\Psi,F)}{\opt(F)}.
\end{equation*}
\end{proposition}

\begin{proof}[\textbf{Proof of \Cref{prop:support_reduction}}]
Let $i^{\mathrm{LB}} \in \arg\max_{0 \le i \le N+1} p_i \cdot q_i$ and define $p^{\mathrm{LB}} = p_{i^{\mathrm{LB}}}$. By Proposition~\ref{prop:r_certificate}, we have $p^{\mathrm{LB}} \in \mathcal{S}(\Iset)$.
We first show that for every $p \notin \mathcal{S}(\Iset)$ and every $F \in \Cb(\Iset)$,
\begin{equation}\label{eq:outside_S_domination}
\Rev(p\vert  F)\le p^{\mathrm{LB}} \cdot q_{i^{\mathrm{LB}}}.
\end{equation}
Assume for the sake of  contradiction that there exist $p\notin\mathcal{S}(\Iset)$ and $F\in\Cb(\Iset)$ such that
$\Rev(p\vert  F) > p^{\mathrm{LB}} \cdot q_{i^{\mathrm{LB}}}$.
Let $q=\overline{F}(p-)$ and consider the extended information set $\IsetPlus = \Iset \cup (p,q)$.
Since $F\in\Cb(\IsetPlus)$, the set $\Cb(\IsetPlus)$ is non-empty.
Therefore, by \Cref{lem:L_belongs}, we have $L_{\Cb}(\cdot\vert  \IsetPlus)\in\Cb(\IsetPlus)\subset\Cb(\Iset)$.
Moreover, applying \Cref{lem:max_L} to $\IsetPlus$ yields
\begin{equation*}
\opt\!\left(L_{\Cb}(\cdot\vert  \IsetPlus)\right)
=
\max\!\left\{ p^{\mathrm{LB}} \cdot q_{i^{\mathrm{LB}}},\, p \cdot q\right\}
=
p \cdot q,
\end{equation*}
where the last equality holds because $p \cdot q=\Rev(p\vert  F)> p^{\mathrm{LB}} \cdot q_{i^{\mathrm{LB}}}$.

Therefore, we have
$\Rev\!\left(p\vert  L_{\Cb}(\cdot\vert  \IsetPlus)\right)=p \cdot q=\opt\!\left(L_{\Cb}(\cdot\vert  \IsetPlus)\right)$,
so $p$ is an optimal price for a distribution in $\Cb(\Iset)$, which implies $p\in\mathcal{S}(\Iset)$.
This contradicts $p\notin\mathcal{S}(\Iset)$, establishing \eqref{eq:outside_S_domination}.
\begin{equation*}
\Rev(p \vert F) \le p^{\mathrm{LB}} \cdot q_{i^{\mathrm{LB}}}.
\end{equation*}

Next, define the mapping $\pi:[\lb,\ub]\to \mathcal{S}(\Iset)$ by
\begin{equation*}
\pi(p) =
\begin{cases}
p, & \text{if } p \in \mathcal{S}(\Iset),\\
p^{\mathrm{LB}}, & \text{otherwise.}
\end{cases}
\end{equation*}
Let $\widetilde{\Psi}$ be the pushforward of $\Psi$ through $\pi$, i.e., $\tilde{\Psi}(B)=\Psi(\pi^{-1}(B))$ for all Borel sets $B\subseteq[\lb,\ub]$.

Fix any $F \in \Cb(\Iset)$. If $p \in \mathcal{S}(\Iset)$, then $\pi(p)=p$. If $p \notin \mathcal{S}(\Iset)$, then
\begin{equation*}
\Rev(\pi(p)\vert F) = \Rev(p^{\mathrm{LB}}\vert F) \ge \Rev(p\vert F).
\end{equation*}
Furthermore, we have
\begin{equation*}
\Rev(\widetilde{\Psi},F)
=
\int_{\lb}^{\ub} \Rev(p\vert F)\, d\widetilde{\Psi}(p)
=
\int_{\lb}^{\ub} \Rev(\pi(p)\vert F)\, d\Psi(p),
\end{equation*}
where the second equality follows from the definition of the pushforward measure.
Therefore,
\begin{equation*}
\Rev(\widetilde{\Psi},F)
=
\int_{\lb}^{\ub} \Rev(\pi(p)\vert F)\, d\Psi(p)
\ge
\int_{\lb}^{\ub} \Rev(p\vert F)\, d\Psi(p)
=
\Rev(\Psi,F).
\end{equation*}
Dividing by $\opt(F)$ and taking the infimum over $F \in \Cb(\Iset)$ yields the result.
\end{proof}

Next, we formally define the notion of partition grids for the set $\mathcal{S}(\Iset)$. We now from \Cref{prop:S_general} and \Cref{prop:S_regular} that this set is a union of intervals. We therefore define a partition as follows.

\begin{definition}[Partition grid of $\mathcal{S}(\Iset)$]
\label{def:partition_grid}
Assume that $\mathcal{S}(\Iset)$ admits a decomposition into finitely many pairwise disjoint intervals
\begin{equation*}
\mathcal{S}(\Iset)=\bigcup_{k=1}^{K} I_k,
\qquad
\sup I_k < \inf I_{k+1}
\quad \text{for all } k\in\{1,\ldots,K-1\}.
\end{equation*}
We say that a finite set $\IN$ is a \emph{partition grid} of $\mathcal{S}(\Iset)$ if, for every $k\in\{1,\ldots,K\}$, there exist points
\begin{equation*}
\inf I_k = a_{k,0} < a_{k,1} < \cdots < a_{k,M_k} = \sup I_k
\end{equation*}
such that $\IN=\bigcup_{k=1}^{K}\{a_{k,0},\ldots,a_{k,M_k}\}.$

% The associated cells of the partition are
% $I_k\cap [a_{k,m},a_{k,m+1}),$
% for  $m \in \{0,\ldots,M_k-1\},$
% together with the terminal cells
% $I_k\cap [a_{k,M_k},\sup I_k].$
\end{definition}

\begin{proposition}[Discretization gap for admissible grids]
\label{prop:mech_discretization_gap_admissible}
Fix an information set $\Iset=\{(p_i,q_i)\}_{i=0}^{N+1}$ and a class $\Cb\in\{\Gb,\Fb\}$, and assume that $\Cb(\Iset)$ is non-empty.
Let $\IN=\{a_j\}_{j=0}^{M}$ be an admissible grid (as in \Cref{def:admissible}) that partitions $\mathcal{S}(\Iset)$, and define its mesh size as $\Delta(\IN)=\max_{0\le j\le M-1}\big(a_{j+1}-a_j\big).$
Assume that
\begin{equation*}
R_{\Iset} =\max_{0\le i\le N+1} p_i \cdot q_i > 0.
\end{equation*}
Then for every mechanism $\Psi\in\MechSet$, there exists a grid-supported mechanism
$\Psi_{\IN}\in\MechSet_{\IN}$ such that
\begin{equation*}
\inf_{F\in\Cb(\Iset)}
\frac{\Rev(\Psi_{\IN},F)}{\opt(F)}
\ \ge\
\inf_{F\in\Cb(\Iset)}
\frac{\Rev(\Psi,F)}{\opt(F)}
\ -\
\frac{\Delta(\IN)}{R_{\Iset}}.
\end{equation*}
Consequently,
\begin{equation*}
\sup_{\Psi\in\MechSet}
\inf_{F\in\Cb(\Iset)}
\frac{\Rev(\Psi,F)}{\opt(F)}
\ \le\
\sup_{\Psi\in\MechSet_{\IN}}
\inf_{F\in\Cb(\Iset)}
\frac{\Rev(\Psi,F)}{\opt(F)}
\ +\
\frac{\Delta(\IN)}{R_{\Iset}}.
\end{equation*}
\end{proposition}

\begin{proof}[\textbf{Proof of \Cref{prop:mech_discretization_gap_admissible}}]

Fix $\Psi \in \MechSet$. By Proposition~\ref{prop:support_reduction}, there exists a mechanism $\widetilde{\Psi}$ supported on $\mathcal{S}(\Iset)$ such that
\begin{equation}
\label{eq:reduction_to_S}
\inf_{F \in \Cb(\Iset)} \frac{\Rev(\widetilde{\Psi},F)}{\opt(F)}
\;\ge\;
\inf_{F \in \Cb(\Iset)} \frac{\Rev(\Psi,F)}{\opt(F)}.
\end{equation}
It therefore suffices to construct a discretized mechanism from $\widetilde{\Psi}$.

Let $\IN=\{a_j\}_{j=0}^{M}$ be an admissible grid that partitions $\mathcal{S}(\Iset)$, and define the rounding map $\pi:\mathcal{S}(\Iset)\to \IN$ by
\begin{equation*}
\pi(p) = \max\{a_j \in \IN \;\text{s.t.}\; a_j \le p\}.
\end{equation*}
Let $\Psi_{\IN}$ be the pushforward of $\widetilde{\Psi}$ through $\pi$, i.e.,
$\Psi_{\IN}(B)=\tilde{\Psi}(\pi^{-1}(B))$ for all Borel sets $B\subseteq[\lb,\ub]$.
By construction, $\Psi_{\IN}\in\MechSet_{\IN}$.

Fix any $F \in \Cb(\Iset)$. Since $\widetilde{\Psi}$ is supported on $\mathcal{S}(\Iset)$, we have
\begin{equation*}
\Rev(\Psi_{\IN},F)
=
\int_{\mathcal{S}(\Iset)} \Rev(\pi(p)\vert F)\, d\widetilde{\Psi}(p).
\end{equation*}
For every $p \in \mathcal{S}(\Iset)$, we have $\pi(p) \le p$ and $\overline{F}(\cdot)$ is non-increasing, hence
\begin{equation*}
\Rev(p\vert F) - \Rev(\pi(p)\vert F)
=
p \cdot \overline{F}(p-) - \pi(p) \cdot \overline{F}(\pi(p)-)
\le
(p - \pi(p)) \cdot \overline{F}(p-)
\le
p - \pi(p).
\end{equation*}
By definition of the mesh size, $p - \pi(p) \le \Delta(\IN)$, which implies
\begin{equation*}
\Rev(\pi(p)\vert F) \ge \Rev(p\vert F) - \Delta(\IN).
\end{equation*}
Integrating yields
\begin{equation*}
\Rev(\Psi_{\IN},F)
\ge
\Rev(\widetilde{\Psi},F) - \Delta(\IN).
\end{equation*}

Finally, since $\opt(F) \ge R_{\Iset}$ for all $F \in \Cb(\Iset)$, we obtain
\begin{equation*}
\frac{\Rev(\Psi_{\IN},F)}{\opt(F)}
\ge
\frac{\Rev(\widetilde{\Psi},F)}{\opt(F)}
-
\frac{\Delta(\IN)}{R_{\Iset}}.
\end{equation*}
Taking the infimum over $F \in \Cb(\Iset)$ yields
\begin{equation*}
\inf_{F \in \Cb(\Iset)}
\frac{\Rev(\Psi_{\IN},F)}{\opt(F)}
\ge
\inf_{F \in \Cb(\Iset)}
\frac{\Rev(\widetilde{\Psi},F)}{\opt(F)}
-
\frac{\Delta(\IN)}{R_{\Iset}}.
\end{equation*}
Combining with \eqref{eq:reduction_to_S} completes the proof.
\end{proof}

\subsection{Structural Results for Constraints Discretization}\label{app:envelope_switching}

The discretization in \Cref{sec:value_information} relies on the fact that, once we refine the grid to include
all envelope switching points, the identity of the active branch of the upper envelope does not change within
a grid cell. We record this formally next.

\begin{lemma}\label{lem:unique_switching_cutoff}
Assume $\Fb(\Iset)\neq\emptyset$ and let $\Iset=\{(p_i,q_i)\}_{i=0}^{N+1}$ be sorted by increasing prices.
Fix any $i\in\{1,\ldots,N-2\}$ and consider the two candidate upper bounds
$\overline{G}_{\Fb,i-1}(\cdot)$ and $\overline{G}_{\Fb,i+1}(\cdot)$ in \eqref{eq:GF}. Then the
equation
\begin{equation*}
\overline{G}_{\Fb,i-1}(v)=\overline{G}_{\Fb,i+1}(v)
\end{equation*}
has either no solution in $[p_i,p_{i+1}]$, or exactly one solution in $[p_i,p_{i+1}]$, or holds for all
$v\in[p_i,p_{i+1}]$. Consequently, there exists a cutoff $c_i\in[p_i,p_{i+1}]$ such that the identity of the
minimizer in
\begin{equation*}
\overline{U}_{\Fb}(v\vert \Iset)=\min\!\left\{\overline{G}_{\Fb,i-1}(v),\overline{G}_{\Fb,i+1}(v)\right\},
\qquad v\in[p_i,p_{i+1}),
\end{equation*}
does not switch on each of the subintervals $[p_i,c_i)$ and $[c_i,p_{i+1})$.
\end{lemma}

\begin{proof}[\textbf{Proof of \Cref{lem:unique_switching_cutoff}}]
Using $\Gamma^{-1}(q)=\frac{1}{q}-1$, the definition of $\overline{G}_{\Fb}$ in \eqref{eq:GF}  can be rewritten
in the reciprocal-linear form
\begin{equation*}
\overline{G}_{\Fb}(v\vert (s,q),(s',q'))
=
\left(
\frac{s'-v}{s'-s}\cdot \frac{1}{q}
+
\frac{v-s}{s'-s}\cdot \frac{1}{q'}
\right)^{-1},
\end{equation*}
hence $\frac{1}{\overline{G}_{\Fb}(v\vert (s,q),(s',q'))}$ is affine in $v$. In particular, both mappings
$v\mapsto \frac{1}{\overline{G}_{\Fb,i-1}(v)}$ and $v\mapsto \frac{1}{\overline{G}_{\Fb,i+1}(v)}$ are affine
on $[p_i,p_{i+1}]$. Therefore, the difference
\begin{equation*}
v\;\mapsto\;\frac{1}{\overline{G}_{\Fb,i-1}(v)}-\frac{1}{\overline{G}_{\Fb,i+1}(v)}
\end{equation*}
is affine on $[p_i,p_{i+1}]$, and its zero set in $[p_i,p_{i+1}]$ is either empty, a singleton, or the entire
interval. Since both candidates are strictly positive, the equality
$\overline{G}_{\Fb,i-1}(v)=\overline{G}_{\Fb,i+1}(v)$ is equivalent to equality of reciprocals, so it has the
same solution set. If the two candidates coincide on all of $[p_i,p_{i+1}]$, then the minimizer in the upper
envelope does not switch. Otherwise there is at most one intersection point; taking $c_i$ equal to that point
(if it exists) yields the stated two-region structure.
\end{proof}

\begin{lemma}\label{lem:admissible_fixed_branch}
Fix an information set $\Iset$, a class $\Cb\in\{\Gb,\Fb\}$, and an admissible grid
$\IN=\{a_k\}_{k=0}^{M}$ as in \Cref{sec:value_information}. Then for every cell index $k\in\{0,\ldots,M\}$
there exists an index $m(k)$ such that
\begin{equation*}
\overline{U}_{\Cb}(\rs-\vert  \Iset)=\overline{G}_{\Cb,m(k)}(\rs)
\qquad \text{for all } \rs\in[a_k,a_{k+1}).
\end{equation*}
\end{lemma}

\begin{proof}[\textbf{Proof of \Cref{lem:admissible_fixed_branch}}]
Fix a cell $[a_k,a_{k+1})$. Since $\IN$ contains all information prices $\{p_i\}\cap\mathcal{S}(\Iset)$, there
exists $t\in\{0,\ldots,N\}$ such that $[a_k,a_{k+1})\subseteq[p_t,p_{t+1})$.

If $\Cb=\Gb$, then by the explicit formula for $\overline{U}_{\Gb}(\cdot\vert \Iset)$ in \Cref{sec:shape_distribution}
the upper envelope is constant on $[p_t,p_{t+1})$, hence also on $[a_k,a_{k+1})$, and can be written as a
single candidate $\overline{G}_{\Gb,m(k)}$ on that cell.

If $\Cb=\Fb$, then for interior information intervals the envelope on $[p_t,p_{t+1})$ is the pointwise minimum
of two candidates (again \Cref{sec:shape_distribution}), and by \Cref{lem:unique_switching_cutoff} there is at most
one switching point $c_t$ in $[p_t,p_{t+1}]$. Since $\IN$ is admissible, it contains every such switching point
that lies in $\mathcal{S}(\Iset)$. Therefore the cell $[a_k,a_{k+1})$ cannot contain a switching point in its
interior, and the identity of the minimizing candidate is constant on the cell. Denote that candidate by
$\overline{G}_{\Fb,m(k)}$.
\end{proof}

\begin{lemma}[Cellwise monotonicity]\label{lem:cellwise_monotonicity}
Fix $\Iset$, $\Cb\in\{\Gb,\Fb\}$, and a grid $\IN=\{a_i\}_{i=0}^{M}$ that partitions
$\mathcal{S}(\Iset)$. Fix a cell index $i\in\{0,\ldots,M\}$ and suppose that
$\overline{U}_{\Cb}(\rs-\vert  \Iset)=\overline{G}_{\Cb,m(i)}(\rs)$ for all
$\rs\in[a_i,a_{i+1})$. Then:
\begin{enumerate}
\item The mapping $\rs \mapsto \opt(F_{\Cb}(\cdot\vert  \rs,\Iset))$ is monotone on $[a_i,a_{i+1})$, and
\begin{equation*}
\sup_{\rs\in[a_i,a_{i+1})}\opt(F_{\Cb}(\cdot\vert  \rs,\Iset))
=
\max\!\left\{
\opt(F_{\Cb}(\cdot\vert  a_i,\Iset)),
\opt(F_{\Cb}(\cdot\vert  a_{i+1}^-,\Iset))
\right\}.
\end{equation*}

\item For every fixed grid price $a_j\in\IN$, the mapping
$\rs \mapsto \Rev(a_j\vert  F_{\Cb}(\cdot\vert  \rs,\Iset))$ is monotone on $[a_i,a_{i+1})$, and
\begin{equation*}
\inf_{\rs\in[a_i,a_{i+1})}\Rev(a_j\vert  F_{\Cb}(\cdot\vert  \rs,\Iset))
=
\min\!\left\{
\Rev(a_j\vert  F_{\Cb}(\cdot\vert  a_i,\Iset)),
\Rev(a_j\vert  F_{\Cb}(\cdot\vert  a_{i+1}^-,\Iset))
\right\}.
\end{equation*}
\end{enumerate}
\end{lemma}

\begin{proof}[\textbf{Proof of \Cref{lem:cellwise_monotonicity}}]
Fix $i$ and assume $\overline{U}_{\Cb}(\rs-\vert  \Iset)=\overline{G}_{\Cb,m(i)}(\rs)$ for all
$\rs\in[a_i,a_{i+1})$.

\smallskip
\noindent\emph{Step 1: the benchmark term.}
For every $\rs\in\mathcal{S}(\Iset)$, the definition of $F_{\Cb}(\cdot\vert  \rs,\Iset)$ in
\eqref{eq:worst_case_F} together with \Cref{prop:r_certificate} implies that $\rs$ is an optimal posted price
under $F_{\Cb}(\cdot\vert  \rs,\Iset)$, with conversion rate
$\overline{F}_{\Cb}(\rs-\vert  \rs,\Iset)=\overline{U}_{\Cb}(\rs-\vert  \Iset)$. Hence,
\begin{equation*}
\opt(F_{\Cb}(\cdot\vert  \rs,\Iset))
=
\rs\cdot \overline{U}_{\Cb}(\rs-\vert  \Iset)
=
\rs\cdot \overline{G}_{\Cb,m(i)}(\rs),
\qquad \rs\in[a_i,a_{i+1}).
\end{equation*}

If $\Cb=\Gb$, then $\overline{G}_{\Gb,m(i)}(\rs)$ is constant in $\rs$ on the cell, so the benchmark
is affine in $\rs$ and therefore monotone.

If $\Cb=\Fb$, then rewriting $\overline{G}_{\Fb}$ using $\Gamma^{-1}(q)=\frac{1}{q}-1$ yields the
reciprocal-linear representation
\begin{equation*}
\overline{G}_{\Fb}(v\vert (s,q),(s',q'))
=
\left(
\frac{s'-v}{s'-s}\cdot \frac{1}{q}
+
\frac{v-s}{s'-s}\cdot \frac{1}{q'}
\right)^{-1},
\end{equation*}
so that $\frac{1}{\overline{G}_{\Fb,m(i)}(\rs)}=\alpha_i+\beta_i \rs$ for some constants $\alpha_i>0$ and
$\beta_i\ge 0$, and therefore
\begin{equation*}
\opt(F_{\Fb}(\cdot\vert  \rs,\Iset))=\rs\cdot \overline{G}_{\Fb,m(i)}(\rs)=\frac{\rs}{\alpha_i+\beta_i \cdot \rs},
\end{equation*}
which is monotone on $[a_i,a_{i+1})$. In both cases, the supremum over the cell is attained at an endpoint,
which yields the desired endpoint reduction.

\smallskip
\noindent\emph{Step 2: the revenue coefficients.}
Fix a grid price $a_j\in\IN$ and define
\begin{equation*}
R_j(\rs)
=
\Rev(a_j\vert  F_{\Cb}(\cdot\vert  \rs,\Iset))
=
a_j\cdot \overline{F}_{\Cb}(a_j-\vert  \rs,\Iset).
\end{equation*}
Let $k$ be such that $[a_i,a_{i+1})\subseteq[p_k,p_{k+1})$, and write
$q_\rs=\overline{U}_{\Cb}(\rs-\vert \Iset)$. Since $\rs\in[p_k,p_{k+1})$, the extended information set
$\Iset\cup(\rs,q_\rs)$ inserts a single anchor point between $(p_k,q_k)$ and $(p_{k+1},q_{k+1})$. Therefore,
by the definition of the lower-envelope operator in \eqref{eq:L},
\begin{equation*}
\overline{F}_{\Cb}(v-\vert  \rs,\Iset)
=
\overline{L}_{\Cb}\!\left(v-\ \vert  \ \Iset\cup(\rs,q_\rs)\right)
=
\begin{cases}
\overline{G}_{\Cb}(v\vert (p_k,q_k),(\rs,q_\rs)), & v\in[p_k,\rs),\\
\overline{G}_{\Cb}(v\vert (\rs,q_\rs),(p_{k+1},q_{k+1})), & v\in[\rs,p_{k+1}),
\end{cases}
\end{equation*}
while $\overline{F}_{\Cb}(v-\vert  \rs,\Iset)=\overline{L}_{\Cb}(v-\vert \Iset)$ for $v\notin[p_k,p_{k+1})$.
Hence, if $a_j\notin[p_k,p_{k+1})$, then $\overline{F}_{\Cb}(a_j-\vert  \rs,\Iset)$ is independent of $\rs$,
and $R_j(\rs)$ is constant on $[a_i,a_{i+1})$.

Assume now that $a_j\in[p_k,p_{k+1})$. Since $a_j$ is a grid point and $\rs$ ranges over the open cell
$[a_i,a_{i+1})$, we have $a_j\notin(a_i,a_{i+1})$. Therefore, on the entire cell either
$a_j\le a_i\le \rs$ for all $\rs\in[a_i,a_{i+1})$, or $a_j\ge a_{i+1}>\rs$ for all $\rs\in[a_i,a_{i+1})$.
In the first case,
\begin{equation*}
\overline{F}_{\Cb}(a_j-\vert  \rs,\Iset)=\overline{G}_{\Cb}(a_j\vert (p_k,q_k),(\rs,q_\rs)),
\qquad \rs\in[a_i,a_{i+1}),
\end{equation*}
and in the second case,
\begin{equation*}
\overline{F}_{\Cb}(a_j-\vert  \rs,\Iset)=\overline{G}_{\Cb}(a_j\vert (\rs,q_\rs),(p_{k+1},q_{k+1})),
\qquad \rs\in[a_i,a_{i+1}).
\end{equation*}
Thus the functional form of $\rs\mapsto \overline{F}_{\Cb}(a_j-\vert  \rs,\Iset)$ does not change within the
cell.

If $\Cb=\Gb$, then $\overline{U}_{\Gb}(v\vert \Iset)$ is constant on $[p_k,p_{k+1})$ (see \Cref{sec:shape_distribution}),
so $q_\rs=q_k$ for all $\rs\in[a_i,a_{i+1})$. Using \eqref{eq:GG}, in the first case we obtain
$\overline{F}_{\Gb}(a_j-\vert  \rs,\Iset)=q_\rs=q_k$, while in the second case we obtain
$\overline{F}_{\Gb}(a_j-\vert  \rs,\Iset)=q_{k+1}$. In either case $R_j(\rs)$ is constant, and thus monotone, on
$[a_i,a_{i+1})$.

If $\Cb=\Fb$, then by \eqref{eq:GF} and the reciprocal-linear representation, in the first case we have
\begin{equation*}
\frac{1}{\overline{F}_{\Fb}(a_j-\vert  \rs,\Iset)}
=
\frac{\rs-a_j}{\rs-p_k}\cdot \frac{1}{q_k}
+
\frac{a_j-p_k}{\rs-p_k}\cdot \frac{1}{q_\rs},
\qquad \rs\in[a_i,a_{i+1}),
\end{equation*}
and in the second case,
\begin{equation*}
\frac{1}{\overline{F}_{\Fb}(a_j-\vert  \rs,\Iset)}
=
\frac{p_{k+1}-a_j}{p_{k+1}-\rs}\cdot \frac{1}{q_\rs}
+
\frac{a_j-\rs}{p_{k+1}-\rs}\cdot \frac{1}{q_{k+1}},
\qquad \rs\in[a_i,a_{i+1}).
\end{equation*}
Furthermore we note that, $q_\rs=\overline{U}_{\Fb}(\rs-\vert \Iset)=\overline{G}_{\Fb,m(i)}(\rs)$ where the second equality follows from the assumption in the Lemma. Hence
$\frac{1}{q_\rs}=\alpha_i+\beta_i\rs$ for some constants $\alpha_i>0$ and $\beta_i\ge 0$ as in Step~1.
Substituting $\frac{1}{q_\rs}=\alpha_i+\beta_i\rs$ and collecting terms yields
\begin{equation*}
\frac{1}{\overline{F}_{\Fb}(a_j-\vert  \rs,\Iset)}
=
\frac{A_{ij}+B_{ij}\rs}{\rs-p_k}
\quad\text{or}\quad
\frac{1}{\overline{F}_{\Fb}(a_j-\vert  \rs,\Iset)}
=
\frac{A'_{ij}+B'_{ij}\rs}{p_{k+1}-\rs},
\end{equation*}
where the coefficients $A_{ij},B_{ij},A'_{ij},B'_{ij}$ are constant on the cell. Differentiating shows that
these ratios have derivatives of constant sign on $[a_i,a_{i+1})$:
\begin{equation*}
\frac{\mathrm{d}}{\mathrm{d}\rs}\left(\frac{A_{ij}+B_{ij}\rs}{\rs-p_k}\right)
=
-\frac{A_{ij}+B_{ij}p_k}{(\rs-p_k)^2},
\qquad
\frac{\mathrm{d}}{\mathrm{d}\rs}\left(\frac{A'_{ij}+B'_{ij}\rs}{p_{k+1}-\rs}\right)
=
\frac{A'_{ij}+B'_{ij}p_{k+1}}{(p_{k+1}-\rs)^2}.
\end{equation*}
It follows that
$\rs\mapsto \overline{F}_{\Fb}(a_j-\vert  \rs,\Iset)$ is monotone on the cell, and therefore so is
$\rs\mapsto R_j(\rs)$.

In all cases, $R_j(\rs)$ is monotone on $[a_i,a_{i+1})$, hence its infimum over the cell is attained at one of
the endpoints. This yields the endpoint reduction in part~(ii).
\end{proof}

\subsection{Lower bound Guarantee of the LP in \eqref{eq:minimax_lp}}

\begin{lemma}[Certified lower bound]
\label{lem:lp_lower_bound}
Fix an information set $\Iset$, a class $\Cb\in\{\Gb,\Fb\}$, and an admissible grid 
$\IN=\{a_k\}_{k=0}^{M}$ that partitions $\mathcal{S}(\Iset)$. 
Let $(\boldsymbol{\psi},\lambda)$ be any feasible solution of the linear program 
\eqref{eq:minimax_lp}, and let $\Psi_{\IN}$ denote the randomized posted-price 
mechanism supported on $\IN$ with weights $\boldsymbol{\psi}$.

Then,
\begin{equation*}
\inf_{F\in\Cb(\Iset)}
\frac{\Rev(\Psi_{\IN},F)}{\opt(F)}
\ge
\lambda.
\end{equation*}

In particular, if $(\boldsymbol{\psi}^*,\lambda^*)$ is an optimal solution of 
\eqref{eq:minimax_lp}, then
\begin{equation*}
\underline{\Lb}(\Cb(\Iset),\IN)
=
\lambda^*
\le
\inf_{F\in\Cb(\Iset)}
\frac{\Rev(\Psi_{\IN}^*,F)}{\opt(F)} \leq  \sup_{\Psi \in \MechSet } \inf_{F\in\Cb(\Iset)}
\frac{\Rev(\Psi,F)}{\opt(F)} 
\end{equation*}
\end{lemma}

\begin{proof}[\textbf{Proof of \Cref{lem:lp_lower_bound}}]
The second inequality is immediate because $\Psi_{\IN}^*\in\MechSet$.

We prove the first inequality. Let $(\boldsymbol{\psi},\lambda)$ be any feasible solution of
\eqref{eq:minimax_lp} and let $\Psi_{\IN}\in\MechSet_{\IN}$ be the corresponding discrete mechanism.
By \Cref{thm:nature_reduction},
\begin{equation*}
\inf_{F\in\Cb(\Iset)}\frac{\Rev(\Psi_{\IN},F)}{\opt(F)}
=
\inf_{\rs\in\mathcal{S}(\Iset)}
\frac{\Rev(\Psi_{\IN},F_{\Cb}(\cdot\vert  \rs,\Iset))}
{\opt(F_{\Cb}(\cdot\vert  \rs,\Iset))}.
\end{equation*}
Fix $\rs\in\mathcal{S}(\Iset)$ and choose $i\in\{0,\ldots,M\}$ such that $\rs\in[a_i,a_{i+1})$.
Since $\IN$ is admissible, \Cref{lem:admissible_fixed_branch} ensures that the hypothesis of
\Cref{lem:cellwise_monotonicity} holds on the cell, so \Cref{lem:cellwise_monotonicity} applies. Using the
definitions of $\widehat{\opt}^{\Cb}_i$ and $\widehat{\Rev}^{\Cb}_{ij}$ from \Cref{sec:value_information}, we
obtain
\begin{equation*}
\opt(F_{\Cb}(\cdot\vert  \rs,\Iset)) \le \widehat{\opt}^{\Cb}_i,
\qquad
\Rev(a_j\vert  F_{\Cb}(\cdot\vert  \rs,\Iset)) \ge \widehat{\Rev}^{\Cb}_{ij}
\quad \text{for all } j \in \{0,\ldots,M\}.
\end{equation*}
Taking expectations under $\Psi_{\IN}$ and using linearity,
\begin{equation*}
\Rev(\Psi_{\IN},F_{\Cb}(\cdot\vert  \rs,\Iset))
=
\sum_{j=0}^{M}\psi_j \cdot \Rev(a_j\vert  F_{\Cb}(\cdot\vert  \rs,\Iset))
\ge
\sum_{j=0}^{M}\psi_j \cdot \widehat{\Rev}^{\Cb}_{ij}.
\end{equation*}
Furthermore, feasibility of $(\boldsymbol{\psi},\lambda)$ in \eqref{eq:minimax_lp} implies
that $\sum_{j=0}^{M}\psi_j \cdot \widehat{\Rev}^{\Cb}_{ij}\ge \lambda \cdot \widehat{\opt}^{\Cb}_i$, hence
\begin{equation*}
\Rev(\Psi_{\IN},F_{\Cb}(\cdot\vert  \rs,\Iset))
\ge
\lambda \cdot \widehat{\opt}^{\Cb}_i
\ge
\lambda \cdot \opt(F_{\Cb}(\cdot\vert  \rs,\Iset)).
\end{equation*}
Therefore,
\begin{equation*}
\frac{\Rev(\Psi_{\IN},F_{\Cb}(\cdot\vert  \rs,\Iset))}
{\opt(F_{\Cb}(\cdot\vert  \rs,\Iset))}
\ge \lambda
\qquad \text{for all } \rs\in\mathcal{S}(\Iset).
\end{equation*}
Taking the infimum over $\rs$ yields
$\inf_{F\in\Cb(\Iset)}\frac{\Rev(\Psi_{\IN},F)}{\opt(F)}\ge \lambda$.
Applying this inequality to an optimal solution $(\boldsymbol{\psi}^*,\lambda^*)$ proves
\begin{equation*}
\inf_{F\in\Cb(\Iset)}\frac{\Rev(\Psi_{\IN}^*,F)}{\opt(F)}
\ge
\lambda^*
=
\underline{\Lb}(\Cb(\Iset),\IN),
\end{equation*}
which is the desired certified lower bound.
\end{proof}

\subsection{Gap Analysis}
\label{app:lp_upper_bound}

The LP in \Cref{sec:value_information} yields a certified lower bound by replacing, within each cell,
the true benchmark $\opt(F_{\Cb}(\cdot\vert \rs,\Iset))$ by its cellwise supremum and each revenue coefficient
$\Rev(a_j\vert F_{\Cb}(\cdot\vert \rs,\Iset))$ by its cellwise infimum. In this subsection, we prove the
corresponding upper approximation guarantee. The argument differs across the two distribution
classes. For the general class $\Gb$, the certified LP is exact for mechanisms supported on the grid.
For the regular class $\Fb$, we compare the certified LP with a complementary optimistic LP and
bound the gap between the two programs.

\subsubsection{Exactness for the general class}

We first show that, under the general class, no additional approximation is introduced by replacing
the continuum of nature constraints with the certified cell constraints. Thus, once the mechanism is
restricted to be supported on the grid, the LP in \eqref{eq:minimax_lp} exactly characterizes the
grid-supported maximin value.

\begin{lemma}[Exactness on the grid for the general class]
\label{lem:general_grid_exactness}
Fix an information set $\Iset$ and an admissible grid $\IN=\{a_k\}_{k=0}^{M}$ that partitions
$\mathcal{S}(\Iset)$. Assume that $\max_{0\le \ell\le N+1}p_\ell \cdot q_\ell>0.$
Then
\begin{equation*}
\underline{\Lb}(\Gb(\Iset),\IN)
=
\sup_{\Psi\in\MechSet_{\IN}}
\inf_{F\in\Gb(\Iset)}
\frac{\Rev(\Psi,F)}{\opt(F)}.
\end{equation*}
\end{lemma}

\begin{proof}[\textbf{Proof of \Cref{lem:general_grid_exactness}}]
The inequality
\begin{equation*}
\underline{\Lb}(\Gb(\Iset),\IN)
\le
\sup_{\Psi\in\MechSet_{\IN}}
\inf_{F\in\Gb(\Iset)}
\frac{\Rev(\Psi,F)}{\opt(F)}
\end{equation*}
follows from \Cref{lem:lp_lower_bound}. We prove the reverse inequality.

Fix a mechanism $\Psi_{\IN}\in\MechSet_{\IN}$ with weights
$\boldsymbol{\psi}=(\psi_0,\ldots,\psi_M)$, and define
\begin{equation*}
\rho(\Psi_{\IN})
=
\inf_{F\in\Gb(\Iset)}
\frac{\Rev(\Psi_{\IN},F)}{\opt(F)}.
\end{equation*}
We show that $(\boldsymbol{\psi},\rho(\Psi_{\IN}))$ is feasible for the certified LP
\eqref{eq:minimax_lp} with $\Cb=\Gb$.

Fix a cell $[a_i,a_{i+1})$. Since the grid is admissible, there exists $k$ such that
$[a_i,a_{i+1})\subseteq[p_k,p_{k+1})$. For every $\rs\in[a_i,a_{i+1})$, we have $\overline U_{\Gb}(\rs-\vert\Iset)=q_k,
$ and  $\opt(F_{\Gb}(\cdot\vert \rs,\Iset))=q_k \cdot \rs.
$

Moreover, the proof of \Cref{lem:cellwise_monotonicity}, specialized to $\Cb=\Gb$, shows that
for every grid price $a_j\in\IN$, the mapping 
$\rs \mapsto \Rev(a_j\vert F_{\Gb}(\cdot\vert \rs,\Iset))
$
is constant on the cell $[a_i,a_{i+1})$. Hence, by the definition of the certified revenue
coefficient,
\begin{equation*}
\Rev(a_j\vert F_{\Gb}(\cdot\vert \rs,\Iset))
=
\widehat{\Rev}^{\Gb}_{ij},
\qquad
\forall \rs\in[a_i,a_{i+1}),\ \forall j.
\end{equation*}

By \Cref{lem:largest_element}, $F_{\Gb}(\cdot\vert \rs,\Iset)\in\Gb(\Iset)$ for every
$\rs\in\mathcal{S}(\Iset)$ and, by definition of $\rho(\Psi_{\IN})$,
\begin{equation*}
\sum_{j=0}^{M}\psi_j
\cdot \Rev(a_j\vert F_{\Gb}(\cdot\vert \rs,\Iset))
\ge
\rho(\Psi_{\IN}) \cdot 
\opt(F_{\Gb}(\cdot\vert \rs,\Iset)),
\qquad
\forall \rs\in[a_i,a_{i+1}).
\end{equation*}
Using the two identities above, this becomes
\begin{equation*}
\sum_{j=0}^{M}\psi_j \cdot \widehat{\Rev}^{\Gb}_{ij}
\ge
\rho(\Psi_{\IN}) \cdot  q_k \cdot \rs,
\qquad
\forall \rs\in[a_i,a_{i+1}).
\end{equation*}
Taking the supremum over $\rs\in[a_i,a_{i+1})$ yields
\begin{equation*}
\sum_{j=0}^{M}\psi_j \cdot \widehat{\Rev}^{\Gb}_{ij}
\ge
\rho(\Psi_{\IN}) \cdot  \widehat{\opt}^{\Gb}_i,
\end{equation*}
Thus all constraints of \eqref{eq:minimax_lp} are satisfied by
$(\boldsymbol{\psi},\rho(\Psi_{\IN}))$, and so
\begin{equation*}
\rho(\Psi_{\IN})
\le
\underline{\Lb}(\Gb(\Iset),\IN).
\end{equation*}
Taking the supremum over $\Psi_{\IN}\in\MechSet_{\IN}$ gives
\begin{equation*}
\sup_{\Psi\in\MechSet_{\IN}}
\inf_{F\in\Gb(\Iset)}
\frac{\Rev(\Psi,F)}{\opt(F)}
\le
\underline{\Lb}(\Gb(\Iset),\IN).
\end{equation*}
Combining the two inequalities proves the claim.
\end{proof}

\subsubsection{Definition of the LP upper bound for the regular class}

We now introduce an optimistic LP for the regular class. This LP is used only as an analytical
device to upper bound the best robust guarantee achievable by mechanisms supported on the grid.

Fix an admissible grid $\IN=\{a_i\}_{i=0}^{M}$. For each cell $i\in\{0,\ldots,M-1\}$ and each grid
price $a_j\in\IN$, define the optimistic coefficients
\begin{equation*}
\check{\opt}^{\Fb}_i
=
\inf_{\rs\in[a_i,a_{i+1})}
\opt(F_{\Fb}(\cdot\vert \rs,\Iset)),
\qquad
\check{\Rev}^{\Fb}_{ij}
=
\sup_{\rs\in[a_i,a_{i+1})}
\Rev(a_j\vert F_{\Fb}(\cdot\vert \rs,\Iset)).
\end{equation*}
By \Cref{lem:cellwise_monotonicity}, both extrema are attained at cell endpoints, hence
\begin{align*}
\check{\opt}^{\Fb}_i
&=
\min\!\left\{
\opt(F_{\Fb}(\cdot\vert a_i,\Iset)),
\opt(F_{\Fb}(\cdot\vert a_{i+1}^-,\Iset))
\right\},\\
\check{\Rev}^{\Fb}_{ij}
&=
\max\!\left\{
\Rev(a_j\vert F_{\Fb}(\cdot\vert a_i,\Iset)),
\Rev(a_j\vert F_{\Fb}(\cdot\vert a_{i+1}^-,\Iset))
\right\}.
\end{align*}

Next, consider mechanisms supported on $\IN$, represented by probability vectors
$\boldsymbol{\psi}=(\psi_0,\ldots,\psi_M)$ with $\psi_j\ge 0$ and $\sum_{j=0}^{M}\psi_j=1$.
Define the optimistic LP value for the regular class as
\begin{subequations}
\begin{alignat}{2}
\overline{\Lb}(\Fb(\Iset), \IN)
= \;& \sup_{\boldsymbol{\psi},\, \lambda}
&\quad& \lambda  \\[2pt]
& \text{s.t.}
& & \lambda \cdot \check{\opt}^{\Fb}_i
-
\sum_{j=0}^{M}
\psi_j \cdot \check{\Rev}^{\Fb}_{ij}
\le 0,
\qquad i=0,\ldots,M-1, \\[4pt]
& &
& \sum_{j=0}^{M} \psi_j = 1, \qquad
\psi_j \ge 0,\qquad
0 \le \lambda \le 1.
\end{alignat}
\label{eq:optimitic_LP}
\end{subequations}

\begin{lemma}[Optimistic upper bound for the regular class]
\label{prop:optimistic_ub_discrete}
Fix an information set $\Iset=\{(p_i,q_i)\}_{i=0}^{N+1}$ and an admissible grid
$\IN=\{a_i\}_{i=0}^{M}$ that partitions $\mathcal{S}(\Iset)$ as in \Cref{def:admissible}. Then
\begin{equation*}
\sup_{\Psi\in \MechSet_{\IN}}
\inf_{F\in\Fb(\Iset)}
\frac{\Rev(\Psi,F)}{\opt(F)}
\le
\overline{\Lb}(\Fb(\Iset),\IN).
\end{equation*}
\end{lemma}

\begin{proof}[\textbf{Proof of \Cref{prop:optimistic_ub_discrete}}]
Within a fixed cell $[a_i,a_{i+1})$, the definitions above and
\Cref{lem:cellwise_monotonicity} ensure the pointwise comparisons
\begin{equation*}
\opt(F_{\Fb}(\cdot\vert \rs,\Iset))
\ge
\check{\opt}^{\Fb}_i,
\qquad
\Rev(a_j\vert F_{\Fb}(\cdot\vert \rs,\Iset))
\le
\check{\Rev}^{\Fb}_{ij},
\qquad
\forall \rs\in[a_i,a_{i+1}),\ \forall j.
\end{equation*}
Therefore, if a discrete mechanism $\Psi_{\IN}$, with associated weights
$\boldsymbol{\psi}$, achieves a worst-case ratio $\lambda$, meaning that
\begin{equation*}
\Rev(\Psi_{\IN},F_{\Fb}(\cdot\vert \rs,\Iset))
\ge
\lambda \cdot \opt(F_{\Fb}(\cdot\vert \rs,\Iset))
\qquad
\forall \rs\in\mathcal{S}(\Iset),
\end{equation*}
then, for each cell $i$ and each $\rs\in[a_i,a_{i+1})$,
\begin{equation*}
\sum_{j=0}^{M}\psi_j \cdot \check{\Rev}^{\Fb}_{ij}
\ge
\sum_{j=0}^{M}\psi_j \cdot
\Rev(a_j\vert F_{\Fb}(\cdot\vert \rs,\Iset))
\ge
\lambda \cdot \opt(F_{\Fb}(\cdot\vert \rs,\Iset))
\ge
\lambda \cdot \check{\opt}^{\Fb}_i.
\end{equation*}
Hence every feasible pair $(\boldsymbol{\psi},\lambda)$ for the maximin ratio problem over
grid-supported mechanisms is also feasible for the optimistic LP. This inclusion shows that the
optimistic LP is a relaxation of the robust design problem over $\MechSet_{\IN}$, and its optimal
value upper bounds the true robust guarantee on the grid:
\begin{equation*}
\sup_{\Psi\in \MechSet_{\IN}}
\inf_{F\in\Fb(\Iset)}
\frac{\Rev(\Psi,F)}{\opt(F)}
\le
\overline{\Lb}(\Fb(\Iset),\IN).
\end{equation*}
\end{proof}

\subsubsection{Quantifying the approximation gap for the regular class}
\label{app:regular_lp_gap}

We now quantify the gap between the optimistic and certified LPs for the regular class. The
argument relies on two grid-independent moduli. The first controls the variation of the oracle
benchmark, while the second controls the variation of revenue coefficients only in the configurations
that can arise inside a grid cell.

Fix an information set $\Iset=\{(p_\ell,q_\ell)\}_{\ell=0}^{N+1}$ and assume
\begin{equation*}
R_{\Iset}
=
\max_{0\le \ell\le N+1}p_\ell \cdot q_\ell
>
0.
\end{equation*}
Let $\mathcal J^{\Fb}(\Iset)$ denote the finite collection of nonempty intervals obtained as follows.
On each information interval $[p_\ell,p_{\ell+1})$, split at the envelope-switching cutoff, if such a
cutoff lies in the interval; if no switch occurs, keep the whole interval. Intersect the resulting
intervals with $\mathcal S(\Iset)$ and discard empty intervals. By construction, on each
$J\in\mathcal J^{\Fb}(\Iset)$ the identity of the candidate function defining
$\overline U_{\Fb}(\cdot-\vert\Iset)$ is fixed.

For each $J\in\mathcal J^{\Fb}(\Iset)$, define the local benchmark modulus
\begin{equation*}
\mathrm{Lip}_{\opt}^{\Fb}(\Iset;J)
=
\sup_{\substack{\rs,\rs'\in J\\ \rs\ne \rs'}}
\frac{
|\opt(F_{\Fb}(\cdot\vert \rs,\Iset))-\opt(F_{\Fb}(\cdot\vert \rs',\Iset))|
}{
|\rs-\rs'|
},
\end{equation*}
and set
\begin{equation*}
L_{\opt}^{\Fb}(\Iset)
=
\max_{J\in\mathcal J^{\Fb}(\Iset)}
\mathrm{Lip}_{\opt}^{\Fb}(\Iset;J).
\end{equation*}

Next, for $J\in\mathcal J^{\Fb}(\Iset)$, define the set of cell-compatible triples
\begin{equation*}
\mathcal D_{\Rev}^{\Fb}(\Iset;J)
=
\left\{
(p,\rs,\rs')\in[\lb,\ub]\times J\times J:
\rs\ne\rs',
\;
p\le \min\{\rs,\rs'\}
\text{ or }
p\ge \max\{\rs,\rs'\}
\right\}.
\end{equation*}
These are precisely the relative positions that can occur when $\rs$ and $\rs'$ belong to the same
grid cell and $p$ is a grid price. Define
\begin{equation*}
\mathrm{Lip}_{\Rev}^{\Fb}(\Iset;J)
=
\sup_{(p,\rs,\rs')\in\mathcal D_{\Rev}^{\Fb}(\Iset;J)}
\frac{
\left|
\Rev(p\vert F_{\Fb}(\cdot\vert \rs,\Iset))
-
\Rev(p\vert F_{\Fb}(\cdot\vert \rs',\Iset))
\right|
}{
|\rs-\rs'|
},
\end{equation*}
and set
\begin{equation*}
L_{\Rev}^{\Fb}(\Iset)
=
\max_{J\in\mathcal J^{\Fb}(\Iset)}
\mathrm{Lip}_{\Rev}^{\Fb}(\Iset;J).
\end{equation*}
Both constants depend only on the information set and on the coarse decomposition induced by the
information prices and switching cutoffs. In particular, they do not depend on the  grid
$\IN$.

\begin{lemma}[Cellwise variation bounds for the regular class]
\label{lem:regular_cellwise_variation_bounds}
Fix an information set $\Iset$ such that $R_{\Iset}>0$, and let $\IN=\{a_i\}_{i=0}^{M}$ be an
admissible grid that partitions $\mathcal S(\Iset)$. Then, for every cell $[a_i,a_{i+1})$,
\begin{equation*}
0
\le
\widehat{\opt}^{\Fb}_i
-
\check{\opt}^{\Fb}_i
\le
L_{\opt}^{\Fb}(\Iset) \cdot \Delta(\IN).
\end{equation*}
Moreover, for every grid price $a_j\in\IN$,
\begin{equation*}
0
\le
\check{\Rev}^{\Fb}_{ij}
-
\widehat{\Rev}^{\Fb}_{ij}
\le
L_{\Rev}^{\Fb}(\Iset) \cdot \Delta(\IN).
\end{equation*}
\end{lemma}

\begin{proof}[\textbf{Proof of \Cref{lem:regular_cellwise_variation_bounds}}]
We first note that the two moduli are finite. Fix $J\in\mathcal J^{\Fb}(\Iset)$. Since the active
candidate defining $\overline U_{\Fb}(\cdot-\vert\Iset)$ is fixed on $J$, there exist constants
$\alpha_J,\beta_J$ such that
\begin{equation*}
\frac{1}{\overline U_{\Fb}(\rs-\vert\Iset)}
=
\alpha_J+\beta_J \cdot \rs,
\qquad
\rs\in J.
\end{equation*}
Therefore,
\begin{equation*}
O(\rs)
=
\opt(F_{\Fb}(\cdot\vert \rs,\Iset))
=
\rs \cdot \overline U_{\Fb}(\rs-\vert\Iset)
=
\frac{\rs}{\alpha_J+\beta_J \cdot \rs},
\qquad
\rs\in J.
\end{equation*}
The denominator is strictly positive on the relevant set of feasible oracle prices, so
$O(\cdot)$ is Lipschitz on $J$.

We next consider the revenue coefficients. Let $[p_k,p_{k+1})$ be the information interval
containing $J$, and write
\begin{equation*}
z_k=\frac{1}{q_k},
\qquad
z_{k+1}=\frac{1}{q_{k+1}},
\qquad
z(\rs)=\frac{1}{\overline U_{\Fb}(\rs-\vert\Iset)}.
\end{equation*}
If $p\notin[p_k,p_{k+1})$, then the value of
$\Rev(p\vert F_{\Fb}(\cdot\vert \rs,\Iset))$ is constant in $\rs$ on $J$. If
$p\in[p_k,p_{k+1})$, the lower-envelope construction gives the following expressions. When
$p\le \rs$,
\begin{equation*}
\overline F_{\Fb}(p-\vert \rs,\Iset)
=
\frac{\rs-p_k}
{
z_k \cdot (\rs-p)+z(\rs) \cdot (p-p_k)
}.
\end{equation*}
When $p\ge \rs$,
\begin{equation*}
\overline F_{\Fb}(p-\vert \rs,\Iset)
=
\frac{p_{k+1}-\rs}
{
z(\rs)\cdot (p_{k+1}-p)+z_{k+1}\cdot(p-\rs)
}.
\end{equation*}

It remains to verify that these expressions have uniformly bounded derivatives with respect to
$\rs$ on the cell-compatible domains. Consider first the region $p\le \rs$. Differentiating the
first expression gives
\begin{equation*}
\frac{\partial}{\partial \rs}
\overline F_{\Fb}(p-\vert \rs,\Iset)
=
\frac{
(p-p_k)\bigl(z(p_k)-z_k\bigr)
}{
\left[z_k(\rs-p)+z(\rs)(p-p_k)\right]^2
},
\end{equation*}
where $z(p_k)$ denotes the value at $p_k$ of the affine branch $z(\cdot)$ active on $J$. If $J$ is
bounded away from $p_k$, then $\rs-p_k$ is uniformly bounded away from zero. Since $z_k>0$ and
$z(\rs)>0$ on $J$, the denominator is uniformly bounded away from zero over all
$p\in[p_k,\rs]$ and $\rs\in J$. Hence the derivative is uniformly bounded in this case.
It remains to consider the case in which the closure of $J$ contains $p_k$. By the construction of
the coarse partition, this can happen only on the leftmost subinterval of $[p_k,p_{k+1})$, where
the active branch is the one passing through $(p_k,q_k)$. Thus $z(p_k)=z_k$, and the numerator
in the derivative above is identically zero. The derivative therefore admits a bounded continuous
extension to the boundary $p=\rs=p_k$.

The same argument applies on the region $p\ge \rs$. 
Thus, on each $J\in\mathcal J^{\Fb}(\Iset)$, the derivatives of the sale-probability coefficients
with respect to $\rs$ are uniformly bounded on all cell-compatible configurations. Moreover, since posted
prices are bounded by $\ub$, the same conclusion holds for the revenue coefficients
$p\,\overline F_{\Fb}(p-\vert \rs,\Iset)$. Therefore $L_{\Rev}^{\Fb}(\Iset)<\infty$.

Now fix a grid cell $[a_i,a_{i+1})$. Since $\IN$ is admissible, this cell is contained in some
$J\in\mathcal J^{\Fb}(\Iset)$. By the definition of $L_{\opt}^{\Fb}(\Iset)$,
\begin{equation*}
\widehat{\opt}^{\Fb}_i-\check{\opt}^{\Fb}_i
=
\sup_{\rs\in[a_i,a_{i+1})}O(\rs)
-
\inf_{\rs\in[a_i,a_{i+1})}O(\rs)
\le
L_{\opt}^{\Fb}(\Iset) \cdot (a_{i+1}-a_i)
\le
L_{\opt}^{\Fb}(\Iset) \cdot \Delta(\IN).
\end{equation*}
This proves the first inequality.

For the revenue coefficients, fix a grid price $a_j\in\IN$. For any
$\rs,\rs'\in[a_i,a_{i+1})$, the point $a_j$ cannot lie strictly between $\rs$ and $\rs'$, because
$a_j$ is a grid point and $\rs,\rs'$ lie in the same grid cell. Therefore
$(a_j,\rs,\rs')\in\mathcal D_{\Rev}^{\Fb}(\Iset;J)$. By the definition of
$L_{\Rev}^{\Fb}(\Iset)$,
\begin{align*}
\check{\Rev}^{\Fb}_{ij}
-
\widehat{\Rev}^{\Fb}_{ij}
&=
\sup_{\rs\in[a_i,a_{i+1})}
\Rev(a_j\vert F_{\Fb}(\cdot\vert \rs,\Iset))
-
\inf_{\rs\in[a_i,a_{i+1})}
\Rev(a_j\vert F_{\Fb}(\cdot\vert \rs,\Iset))\\
&\le
L_{\Rev}^{\Fb}(\Iset) \cdot (a_{i+1}-a_i)
\le
L_{\Rev}^{\Fb}(\Iset) \cdot \Delta(\IN).
\end{align*}
This proves the second inequality.
\end{proof}

\begin{theorem}[Quantified gap between the optimistic and certified LPs for the regular class]
\label{thm:regular_lp_gap}
Fix an information set $\Iset=\{(p_\ell,q_\ell)\}_{\ell=0}^{N+1}$ and an admissible grid
$\IN=\{a_i\}_{i=0}^{M}$ that partitions $\mathcal S(\Iset)$. Assume
\begin{equation*}
R_{\Iset}
=
\max_{0\le \ell\le N+1}p_\ell \cdot q_\ell
>
0.
\end{equation*}
Then
\begin{equation*}
0
\le
\overline{\Lb}(\Fb(\Iset),\IN)
-
\underline{\Lb}(\Fb(\Iset),\IN)
\le
\frac{
L_{\Rev}^{\Fb}(\Iset)+L_{\opt}^{\Fb}(\Iset)
}{
R_{\Iset}
}
 \cdot \Delta(\IN).
\end{equation*}
\end{theorem}

\begin{proof}[\textbf{Proof of \Cref{thm:regular_lp_gap}}]
The first inequality follows because the optimistic LP is a relaxation of the certified LP:
$\check{\opt}^{\Fb}_i\le \widehat{\opt}^{\Fb}_i$ and
$\check{\Rev}^{\Fb}_{ij}\ge \widehat{\Rev}^{\Fb}_{ij}$ for all $i,j$.

We prove the upper bound. Let $(\overline{\boldsymbol{\psi}},\overline{\lambda})$ be an optimal
solution of the optimistic LP \eqref{eq:optimitic_LP}. Define
\begin{equation*}
\delta
=
\frac{
L_{\Rev}^{\Fb}(\Iset)+L_{\opt}^{\Fb}(\Iset)
}{
R_{\Iset}
}
\Delta(\IN).
\end{equation*}
If $\overline{\lambda}\le \delta$, then
\begin{equation*}
\overline{\Lb}(\Fb(\Iset),\IN)
-
\underline{\Lb}(\Fb(\Iset),\IN)
\le
\overline{\lambda}
\le
\delta,
\end{equation*}
because the certified LP value is nonnegative. Hence assume $\overline{\lambda}>\delta$, and set $\lambda^-=\overline{\lambda}-\delta$.
We show that $(\overline{\boldsymbol{\psi}},\lambda^-)$ is feasible for the certified LP.

Fix a cell $i$. Since $F_{\Fb}(\cdot\vert \rs,\Iset)\in\Fb(\Iset)$ for every
$\rs\in\mathcal S(\Iset)$, the benchmark is uniformly bounded below by the best observed revenue:
\begin{equation*}
\opt(F_{\Fb}(\cdot\vert \rs,\Iset))
\ge
R_{\Iset},
\qquad
\forall \rs\in\mathcal S(\Iset).
\end{equation*}
In particular, this implies that  $\widehat{\opt}^{\Fb}_i\ge R_{\Iset}.
$ Furthermore, by \Cref{lem:regular_cellwise_variation_bounds}, we have 
\begin{equation*} 
\check{\opt}^{\Fb}_i
\ge
\widehat{\opt}^{\Fb}_i
-
L_{\opt}^{\Fb}(\Iset) \cdot \Delta(\IN) \quad \mbox{and} \quad \widehat{\Rev}^{\Fb}_{ij}
\ge
\check{\Rev}^{\Fb}_{ij}
-
L_{\Rev}^{\Fb}(\Iset) \cdot \Delta(\IN).
\end{equation*}
Using $\sum_{j=0}^{M}\overline{\psi}_j=1$, we obtain
\begin{align*}
\sum_{j=0}^{M}\overline{\psi}_j \cdot \widehat{\Rev}^{\Fb}_{ij}
&\ge
\sum_{j=0}^{M}\overline{\psi}_j \cdot \check{\Rev}^{\Fb}_{ij}
-
L_{\Rev}^{\Fb}(\Iset) \cdot \Delta(\IN)\\
&\ge
\overline{\lambda} \cdot \check{\opt}^{\Fb}_i
-
L_{\Rev}^{\Fb}(\Iset) \cdot \Delta(\IN)\\
&\ge
\overline{\lambda} \cdot \widehat{\opt}^{\Fb}_i
-
\left(
L_{\opt}^{\Fb}(\Iset)+L_{\Rev}^{\Fb}(\Iset)
\right) \cdot \Delta(\IN),
\end{align*}
where the second inequality uses feasibility of
$(\overline{\boldsymbol{\psi}},\overline{\lambda})$ for the optimistic LP, and the last inequality uses
$\overline{\lambda}\le 1$. Since $\widehat{\opt}^{\Fb}_i\ge R_{\Iset}$, the definition of $\delta$ implies
\begin{equation*}
\left(
L_{\opt}^{\Fb}(\Iset)+L_{\Rev}^{\Fb}(\Iset)
\right) \cdot \Delta(\IN)
\le
\delta \cdot \widehat{\opt}^{\Fb}_i.
\end{equation*}
Therefore,
\begin{equation*}
\sum_{j=0}^{M}\overline{\psi}_j \cdot \widehat{\Rev}^{\Fb}_{ij}
\ge
(\overline{\lambda}-\delta) \cdot \widehat{\opt}^{\Fb}_i
=
\lambda^- \cdot \widehat{\opt}^{\Fb}_i.
\end{equation*}
This is exactly the certified LP constraint for cell $i$. Since $i$ was arbitrary,
$(\overline{\boldsymbol{\psi}},\lambda^-)$ is feasible for the certified LP, and hence
\begin{equation*}
\underline{\Lb}(\Fb(\Iset),\IN)
\ge
\lambda^-
=
\overline{\Lb}(\Fb(\Iset),\IN)-\delta.
\end{equation*}
Rearranging proves the result.
\end{proof}

\subsection{Proof of \Cref{thm:lp_discrete_bound}}
\label{sec:apx_proof_thm2}

\begin{proof}[\textbf{Proof of \Cref{thm:lp_discrete_bound}}]
The first two inequalities in the statement follow directly from \Cref{lem:lp_lower_bound}. Indeed,
if $(\boldsymbol{\psi}^*,\lambda^*)$ is an optimal solution of \eqref{eq:minimax_lp}, and
$\Psi_{\IN}^*$ denotes the corresponding mechanism supported on $\IN$, then
\begin{equation*}
\underline{\Lb}(\Cb(\Iset),\IN)
=
\lambda^*
\le
\inf_{F\in\Cb(\Iset)}
\frac{\Rev(\Psi_{\IN}^*,F)}{\opt(F)}
\le
\sup_{\Psi\in\MechSet}
\inf_{F\in\Cb(\Iset)}
\frac{\Rev(\Psi,F)}{\opt(F)}.
\end{equation*}

It remains to prove the final upper bound. Let $R_{\Iset}
=
\max_{0\le i\le N+1}p_i \cdot q_i,$
which is strictly positive by assumption. Define the unrestricted and grid-supported maximin values
\begin{equation*}
V_{\Cb}
=
\sup_{\Psi\in\MechSet}
\inf_{F\in\Cb(\Iset)}
\frac{\Rev(\Psi,F)}{\opt(F)},
\qquad
V_{\Cb,\IN}
=
\sup_{\Psi\in\MechSet_{\IN}}
\inf_{F\in\Cb(\Iset)}
\frac{\Rev(\Psi,F)}{\opt(F)}.
\end{equation*}
By \Cref{prop:mech_discretization_gap_admissible},
\begin{equation*}
V_{\Cb}
\le
V_{\Cb,\IN}
+
\frac{\Delta(\IN)}{R_{\Iset}}.
\end{equation*}

We now distinguish the two distribution classes.

First, suppose $\Cb=\Gb$.  \Cref{lem:general_grid_exactness} implies that, $V_{\Gb,\IN}
=
\underline{\Lb}(\Gb(\Iset),\IN).$
Therefore,
\begin{equation*}
V_{\Gb}
\le
\underline{\Lb}(\Gb(\Iset),\IN)
+
\frac{\Delta(\IN)}{R_{\Iset}}.
\end{equation*}
Thus the desired upper bound holds for the general class with
\begin{equation*}
C_{\Gb}
=
\frac{1}{R_{\Iset}}.
\end{equation*}

Next, suppose $\Cb=\Fb$. We have
\begin{equation*}
V_{\Fb,\IN} \stackrel{(a)}{\leq} \overline{\Lb}(\Fb(\Iset),\IN)
\stackrel{(b)}{\le}
\underline{\Lb}(\Fb(\Iset),\IN)
+
\frac{
L_{\Rev}^{\Fb}(\Iset)+L_{\opt}^{\Fb}(\Iset)
}{
R_{\Iset}
}
\Delta(\IN),
\end{equation*}
where $(a)$ follows from \Cref{prop:optimistic_ub_discrete} and $(b)$ is a consequence of \Cref{thm:regular_lp_gap}.
Hence,
\begin{equation*}
V_{\Fb}
\le
\underline{\Lb}(\Fb(\Iset),\IN)
+
\frac{
L_{\Rev}^{\Fb}(\Iset)+L_{\opt}^{\Fb}(\Iset)+1
}{
R_{\Iset}
}
\Delta(\IN).
\end{equation*}
Thus the desired upper bound holds for the regular class with
\begin{equation*}
C_{\Fb}
=
\frac{
L_{\Rev}^{\Fb}(\Iset)+L_{\opt}^{\Fb}(\Iset)+1
}{
R_{\Iset}
}.
\end{equation*}

Combining the two cases proves the theorem. In particular, the constant in the statement can be
chosen as
\begin{equation*}
C
=
\begin{cases}
\dfrac{1}{R_{\Iset}}, & \text{if } \Cb=\Gb,\\[8pt]
\dfrac{
L_{\Rev}^{\Fb}(\Iset)+L_{\opt}^{\Fb}(\Iset)+1
}{
R_{\Iset}
}, & \text{if } \Cb=\Fb.
\end{cases}
\end{equation*}
Both constants depend only on the information set and the distribution class, and are independent
of the admissible grid $\IN$.
\end{proof}

\subsection{Near Maximin-Optimal Policy}

Recall the information set in \Cref{fig_feasible} defined as 
\begin{equation*}
    \{(p_i,q_i), \, i \in \{1,\ldots, 4\} \} = \{ (0.47,0.25), (0.70,0.055), (0.89,0.017), (0.98,0.006)   \}
\end{equation*}
and with $\lb = 0, \, \ub =1$. In what follows, we illustrate the shape of the set of candidate optimal prices $\mathcal{S}(\Iset)$ and the structure of the near-optimal maximin policy defined in \Cref{sec:value_information} for this information set under both the assumptions that the set of distributions are general and regular. 

\definecolor{regularpolicycolor}{HTML}{0F766E}
\definecolor{generalpolicycolor}{HTML}{C2410C}
\definecolor{pricepolicycolor}{HTML}{DC2626}
\definecolor{policygridcolor}{HTML}{D9E2F0}

\pgfplotsset{
  maximinlegend/.style={
    legend style={
      font=\small,
      row sep=2pt,
      draw=none,
      fill=white,
      fill opacity=0.9,
      text opacity=1,
      legend pos=south east
    }
  },
  maximinpolicyaxis/.style={
    width=0.8\linewidth,
    height=0.24\textheight,
    xmin=0,
    xmax=1,
    ymin=0,
    ymax=1.02,
    ylabel={CDF},
    ylabel style={font=\scriptsize},
    xtick={0,0.2,0.4,0.6,0.8,1},
    ytick={0,0.25,0.5,0.75,1},
    grid=both,
    major grid style={draw=policygridcolor, line width=0.4pt},
    minor grid style={draw=policygridcolor, line width=0.2pt},
    table/col sep=comma,
    clip=false,
    maximinlegend,
  }
}

\pgfplotstableread[col sep=comma]{Data/figure1_minimax_policy_feasible_set_regular.csv}\regularpolicyset
\pgfplotstableread[col sep=comma]{Data/figure1_minimax_policy_feasible_set_general.csv}\generalpolicyset

\newcommand{\DrawPolicyFeasibleBand}[2]{%
  \pgfplotstablegetrowsof{#1}%
  \pgfmathtruncatemacro{\PolicyLastRow}{\pgfplotsretval-1}%
  \pgfplotsinvokeforeach{0,...,\PolicyLastRow}{%
    \pgfplotstablegetelem{##1}{left}\of{#1}\edef\PolicyLeft{\pgfplotsretval}%
    \pgfplotstablegetelem{##1}{right}\of{#1}\edef\PolicyRight{\pgfplotsretval}%
    \pgfplotstablegetelem{##1}{y_bottom}\of{#1}\edef\PolicyBottom{\pgfplotsretval}%
    \pgfplotstablegetelem{##1}{y_top}\of{#1}\edef\PolicyTop{\pgfplotsretval}%
    \addplot[
      draw=none,
      fill=#2,
      fill opacity=0.10,
      forget plot,
    ] coordinates {
      (\PolicyLeft,\PolicyBottom)
      (\PolicyLeft,\PolicyTop)
      (\PolicyRight,\PolicyTop)
      (\PolicyRight,\PolicyBottom)
    } \closedcycle;
  }%
}

\begin{figure}[h!]
\centering

\begin{subfigure}{0.97\textwidth}
\centering
\begin{tikzpicture}
\begin{axis}[
  maximinpolicyaxis,
  title={Regular assumption},
  xlabel={},
]
\DrawPolicyFeasibleBand{\regularpolicyset}{regularpolicycolor}

\addplot[
  very thick,
  regularpolicycolor,
  forget plot,
] table[x=x, y=cdf] {Data/figure1_minimax_policy_cdf_regular.csv};

\addplot[
  only marks,
  mark=triangle*,
  mark size=4pt,
  draw=pricepolicycolor,
  fill=pricepolicycolor,
  forget plot,
] table[x=price, y=y_dot] {Data/figure1_minimax_policy_price_marks.csv};

\addlegendimage{area legend, fill=regularpolicycolor, fill opacity=0.14, draw=none}
\addlegendentry{Regular $S(I_{p,q})$}
\addlegendimage{only marks, mark=triangle*, mark options={draw=pricepolicycolor, fill=pricepolicycolor}}
\addlegendentry{Observed prices $p_i$}
\end{axis}
\end{tikzpicture}
\end{subfigure}

\vspace{0.8em}

\begin{subfigure}{0.97\textwidth}
\centering
\begin{tikzpicture}
\begin{axis}[
  maximinpolicyaxis,
  title={General assumption},
  xlabel={$p$},
]
\DrawPolicyFeasibleBand{\generalpolicyset}{generalpolicycolor}

\addlegendimage{area legend, fill=generalpolicycolor, fill opacity=0.14, draw=none}
\addlegendentry{General $S(I_{p,q})$}
\addlegendimage{only marks, mark=triangle*, mark options={draw=pricepolicycolor, fill=pricepolicycolor}}
\addlegendentry{Observed prices $p_i$}

\addplot[
  very thick,
  generalpolicycolor,
  forget plot,
] table[x=x, y=cdf] {Data/figure1_minimax_policy_cdf_general.csv};

\addplot[
  only marks,
  mark=triangle*,
  mark size=4pt,
  draw=pricepolicycolor,
  fill=pricepolicycolor,
  forget plot,
] table[x=price, y=y_dot] {Data/figure1_minimax_policy_price_marks.csv};

\end{axis}
\end{tikzpicture}
\end{subfigure}

\caption{\textbf{CDF of the near maximin-optimal policy.}
Each panel shows the cumulative distribution function of the near-optimal maximin pricing
policy under the regular or general assumption. The red markers on the x-axis indicate the
observed prices $p_i$, and the shaded region indicates the set $S(I_{p,q})$. (M=400)}
\label{fig:maximin-optimal-policy}
\end{figure}

\section{Algorithmic Details and Additional Results for \Cref{sec:applications}}
\label{sec:apx_algorithms}

In this appendix, we provide the algorithmic details for the dynamic pricing procedures discussed in \Cref{sec:ternary}. We describe both the classical ternary search algorithm and the meta-algorithm based on our value-of-information framework.

\subsection{Ternary Search Procedure}

We first recall the ternary search algorithm used as the baseline exploration procedure. This algorithm applies when the revenue function is unimodal, as is the case under regular demand. It is presented in Algorithm~\ref{algo:ts}.

\begin{algorithm}[h]
\caption{Ternary Search}\label{algo:ts}
\texttt{Input:} Acceptable sub-optimality $\epsilon$\;
{Initialize $a=\lb$, $b=\ub$, $\Iset^{1} = \{(\lb, 1), (\ub, 0)\}$\;
\While {$t < \frac{\log(\frac{\ub-\lb}{\lb \cdot \epsilon})}{\log(3/2)}$ \tcp*{Ternary search usual stopping criterion}} {
\tcc{Exploration procedure of the Ternary search algorithm}
    Price at $p_1 = a + \frac{b-a}{3}$, observe $q_1$ \;
    Price at $p_2 = a + \frac{2(b-a)}{3}$, observe $q_2$ \;
    Record observed conversion rates $\Iset^{t+1} = \Iset^{t} \cup \{(p_1,q_1), (p_2, q_2)\} $ \;
    
    \If {$p_1 \cdot q_1 < p_2 \cdot q_2$} {
        Set $a = p_1$ \;
    }
    \Else {
        Set $b = p_2$ \;
    }
    Increment $t$
}
\textrm{Return} $a$ if $a \cdot q_a \geq b \cdot q_b$ else return $b$ \tcp*{Exploitation price}
}
\end{algorithm}

The stopping criterion used in this procedure is non-adaptive and depends only on the target accuracy level $\epsilon$. It ensures that the final price lies in an interval of sufficiently small length, which implies near-optimal performance under unimodality. The next result characterizes the number of samples required to guarantee a given performance.

\begin{lemma}
\label{lem:ternary_stop}
Assume  $\lb > 0$. The ternary stopping criterion is a valid stopping criterion.
\end{lemma}

\begin{proof}[\textbf{Proof of \Cref{lem:ternary_stop}.}]
Denote by $T_N = [a_N, b_N]$, the interval obtained at iteration $N$ in Algorithm \ref{algo:ts}, by uni-modality of $\Rev\left(\cdot\vert F\right)$ we know that $\rs$ is in $[a_N, b_N]$, therefore we have:
\begin{align*}
    1-\frac{\Rev\left(a_N\vert F\right)}{\Rev\left(\rs\vert F\right)} &= \frac{\rs \cdot \bF(\rs-) - a_N \cdot \bF(a_N-)}{\Rev\left(\rs\vert F\right)} \\
    &\leq \frac{(\rs - a_N) \cdot \bF(\rs-)}{\Rev\left(\rs\vert F\right)}\leq \frac{b_N -a_N}{\Rev\left(\rs\vert F\right)} = \frac{\ub - \lb}{\Rev\left(\rs\vert F\right)} \left(\frac{2}{3}\right)^N
    \leq \frac{\ub - \lb}{\lb} \left(\frac{2}{3}\right)^N
\end{align*}
Therefore for $N \geq N(\epsilon) = \frac{\log(\frac{\ub-\lb}{\lb \cdot \epsilon})}{\log(3/2)}$, we obtain the desired level of the maximin ratio.
\end{proof}

\subsection{Meta-Algorithm with Data-Driven Stopping}

We now describe in Algorithm~\ref{algo:algo_with_stopping} the meta-algorithm introduced in \Cref{sec:ternary}, which replaces the stopping criterion of a generic exploration procedure with a data-driven rule based on the value of information.

\begin{algorithm}[h]
\caption{Meta-algorithm for dynamic pricing} \label{algo:algo_with_stopping}
\texttt{Input:} Dynamic pricing algorithm $\theta$; Acceptable sub-optimality $\epsilon$; parameter $M$ for discretization in \eqref{eq:minimax_lp}\;
Initialize $\Iset^{1}$\;
\While {$\underline{\Lb}(\Cb(\Iset^{t}), \IN) < 1-\epsilon$} {
    Use the exploration procedure of $\theta$\;
    Update $\Iset^{t+1}$ and increment $t$\;
}
 \textrm{Return} $(\psi^*_i)_{i \in \{1,\ldots,M\}}$ solution of the discretized problem \eqref{eq:minimax_lp} \;
\end{algorithm}

At each iteration, the algorithm evaluates the quantity $\underline{\Lb}(\Cb(\Iset^{t}), \IN)$, which provides a lower bound on the maximin ratio achievable given the current information set. The exploration phase stops as soon as this bound exceeds $1-\epsilon$, ensuring that the returned mechanism satisfies the desired performance guarantee.

\section{Experimental setting for Figure 6}
\label{sec:apx_figure6}

We now describe in more details the experimental setup used to generate \Cref{fig:true_ratio_vs_M}.

We consider grids of $K \in \{3,5,9\}$ prices, equally spaced over $[\lb,\ub]$, together with their associated true conversion rates $q_i = q(p_i)$. The resulting price-quantile pairs are:
\begin{itemize}
    \item $K = 3$: 
    \[
    (20, 0.599), \quad (60, 0.057), \quad (100, 0.002)
    \]
    \item $K = 5$: 
    \[
    (20, 0.599), \quad (40, 0.231), \quad (60, 0.057), \quad (80, 0.012), \quad (100, 0.002)
    \]
    \item $K = 9$: 
    \[
    (20, 0.599), \quad (30, 0.401), \quad (40, 0.231), \quad (50, 0.119), \quad (60, 0.057), 
    \]
    \[
    (70, 0.027), \quad (80, 0.012), \quad (90, 0.006), \quad (100, 0.002)
    \]
\end{itemize}

At each price $p_i$, we observe $T$ independent binary outcomes, where the number of purchases follows a Binomial distribution with parameters $(T, q_i)$. We vary the sample size across $T \in \{25, 50, 100, 250, 500, 750, 1000\}$.

For each configuration $(K,T)$, we generate $1000$ independent datasets. In each dataset, empirical conversion rates are first projected onto a feasible set consistent with the assumed distribution class, using isotonic regression in the general case and a convexity-constrained projection in the regular case. We then compute the corresponding pricing policy by solving the discretized linear program described in \Cref{sec:value_information} with $M=100$ support points.
\end{document}